	\documentclass[pra,aps,showpacs,amsfonts,superscriptaddress,onecolumn,nofootinbib,tikz]{revtex4-2}
\usepackage{tkz-graph}
\usetikzlibrary{graphs}
\usetikzlibrary{shapes,snakes}
\usepackage{color}
\definecolor{myurlcolor}{rgb}{0,0,0.4}
\definecolor{mycitecolor}{rgb}{0,0.5,0}
\definecolor{myrefcolor}{rgb}{0.5,0,0}
\definecolor{tropicalrainforest}{rgb}{0.0, 0.46, 0.37}
\definecolor{ufogreen}{rgb}{0.24, 0.82, 0.44}
\usepackage{hyperref}
\hypersetup{colorlinks,
	linkcolor=myrefcolor,
	citecolor=mycitecolor,
	urlcolor=myurlcolor}
\usepackage[mathscr]{euscript}
\usepackage{makecell}
\usepackage[draft]{fixme}
\usepackage{tikz}
\usepackage{tkz-graph}
\usepackage{subfigure}
\usepackage{graphicx}
\usepackage{amsfonts}
\usepackage{amsmath, amssymb, amsthm,verbatim,graphicx,bbm}
\usepackage{palatino}
\usepackage{mathpazo}
\usepackage{dsfont}
\usetikzlibrary{patterns}
\usepackage{soul}

\definecolor{yellow-green}{rgb}{0.6, 0.8, 0.2}
\definecolor{new-color}{rgb}{0.6, 0, 0.2}
\definecolor{new-color-2}{rgb}{0.2,0.2,0.8}
\definecolor{yellowgreen}{rgb}{0.8,0.6,0.5}
\newcommand{\beq}[0]{\begin{equation}}
	\newcommand{\eeq}[0]{\end{equation}}

\newcommand{\girth}{\mathbbm{g}}

\newcommand{\ket}[1]{|#1\rangle}
\newcommand{\bra}[1]{\langle#1|}

\newcommand{\repr}[1]{\ket{v_{#1}}}
\newcommand{\non}[0]{\nonumber \\}

\newcommand{\one}{\leavevmode\hbox{\small1\normalsize\kern-.33em1}}

\def\be{\begin{equation}}
	\def\ee{\end{equation}}
\def\ben{\begin{eqnarray}}
	\def\een{\end{eqnarray}}
\def\eea{\end{array}}
\def\bea{\begin{array}}

\newcommand{\bei}{\begin{itemize}}
	\newcommand{\eei}{\end{itemize}}

\newcommand{\zero}{\ket{0}}
\newcommand{\jeden}{\ket{1}}
\newcommand{\dwa}{\ket{2}}

\renewcommand{\emph}[1]{\textbf{#1}}

\theoremstyle{plain}
\newtheorem{thm}{Theorem}

\newtheorem{lem}{Lemma}
\newtheorem{fakt}{Fact}
\newtheorem{obs}{Observation}

\newtheorem{defn}[thm]{Definition}

\theoremstyle{definition}

\theoremstyle{remark}

\begin{document}
	\title{
Progress in the study of the (non)existence of genuinely unextendible product bases}
	\author{Maciej Demianowicz}
	\affiliation{Institute of Physics and Applied Computer Science, Faculty of
		Applied Physics and Mathematics, Gda\'nsk University of Technology,
		Narutowicza 11/12, 80-233 Gda\'nsk, Poland}
		\email{maciej.demianowicz@pg.edu.pl}
	\date{\today}
	
	\begin{abstract}
We investigate the open problem of the existence of genuinely unextendible product bases (GUPBs), that is, multipartite unextendible product bases (UPBs) which remain unextendible even with respect to biproduct vectors across all bipartitions of the parties. To this end, we exploit the well-known connection between UPBs and graph theory through orthogonality graphs and orthogonal representations, together with recent progress in this framework, and employ forbidden induced subgraph characterizations to single out the admissible local orthogonality graphs for GUPBs. Using this approach, we establish that GUPBs of size thirteen in three-qutrit systems—the smallest candidate GUPBs—do not exist. We further provide a partial characterization of graphs relevant to larger bases and systems with ququart subsystems.
	\end{abstract}
	
	\maketitle
	
	\section{Introduction}
	
	Unextendible product bases (UPBs) \cite{upb,big-upb} constitute a very important notion
	in quantum information science, being relevant in a broad array of theoretical and practical areas.   Specifically, one of the main results from the
	theory of entanglement related to UPBs is that a state supported on the
	orthogonal subspace is entangled and has positive partial transpose (PPT), i.e.,
	it is PPT bound entangled \cite{ppt-bound}. UPBs also display the effect of nonlocality without
	entanglement \cite{nwe,rinaldis,cerf,strong-nwe, strong-nwe-2}. Interestingly, UPBs found applications in the domain of Bell nonlocality, where they were exploited in the construction of Bell-type inequalities with no quantum violation \cite{no-quantum,no-quantum-2}.

	Much effort has been devoted to the characterization of UPBs, including establishing bounds on their sizes (see, e.g., \cite{djokovic-qubit,johnston,min-upb-johnston}) and developing explicit constructions (see, e.g., \cite{djokovic,7-qubit,tile-bipartite}). In particular, the construction of UPBs with specific uncompletability \cite{niekompletowalne,sucb} or unextendibility \cite{upb-to-ges,wang,k-ces} properties has attracted considerable attention in the community. 	Despite this progress, a certain class of UPBs has remained elusive---genuinely unextendible product bases (GUPBs), i.e.,  UPBs which are unextendible
	in the strongest possible sense, that is even with biproduct vectors
	\cite{upb-to-ges,MD-PRA}.  In fact, while it has been shown that one can lower-bound the cardinalities of GUPBs in a non-trivial way, i.e., provide bounds which beat the ones stemming only from the theory of bipartite UPBs \cite{MD-PRA,kiribela}, their actual existence, even in simple systems, remains an important open problem.
		
	It was early recognized that UPBs share deep connections to combinatorial
	mathematics and graph theory \cite{big-upb,alon-lovasz,feng}. Specifically, the notions of an orthogonality graph and an orthogonal representation of a graph prove very useful.
	Very recently, several strong general results in this framework were shown by Shi et
	al.  in \cite{kiribela}. Their concise graph-theoretic characterization of UPBs also pushed forward our understanding of GUPBs. One particularly appealing result is
that for certain minimal GUPBs their orthogonality graphs corresponding to single-partite subsystems are regular. This strong property severely narrows the set of eligible orthogonality graphs in these cases. Together with the fact that valid orthogonal representations of these graphs must satisfy certain spanning (or, saturation) properties, this paved way to investigate the existence of GUPBs in an organized fashion.
Indeed, Shi et al. \cite{kiribela} put forward a concrete route to search for minimal, size thirteen, GUPBs in tripartite systems with qutrit subsystems---the smallest systems in which such bases may in principle exist.
The proposition was to consider all thirteen-vertex four-regular graphs and check whether they admit orthogonal representations complying with certain requirements on their spanning properties. The realization of this idea in \cite{kiribela}, relying on an optimization procedure in the search for orthogonal representations applied to all graphs from the relevant class, however, proved inconclusive, leaving the GUPB existence problem unsolved. 

In this paper, we introduce the method of forbidden induced subgraph characterization into the search for GUPBs. The idea is to identify a small set of graphs, possibly with only a few vertices, that lack orthogonal representations in a given dimension. The relevant graph family is then pruned by discarding all graphs that contain any of these as induced subgraphs. Only the remaining graphs must then be checked for the existence of an orthogonal representation. With this aim, it is often possible to  exploit further  the properties of some smaller graphs and extract efficiently relevant  features of representations  and decide on their validity for the considered task. We apply this approach to the case of minimal GUPBs in three-qutrit systems and find that only two candidate graphs admit orthogonal representations, but these turn out not to be valid, leading us to conclude that the sought-for GUPBs do not exist. While we were able to obtain this result without investing much computational effort, in general, for larger systems or basis sizes, the approach will still require it. We discuss these limitations and the applicability of the method to other cases, and provide some initial results in this regard.

	The paper is organized as follows. In Sec. \ref{intro} we introduce the relevant notions and notation. Sec. \ref{main-section} focuses on the analysis of minimal three-qutrit GUPBs and forms the core of the paper. In Secs. \ref{kandydaci}–\ref{polaczone} we analyze in detail all candidate graphs and argue that only two of them admit faithful orthogonal representations, which, nevertheless, are shown not to be valid for GUPBs. The main result of the paper is presented in Sec. \ref{konkluzja}. In Sec. \ref{ogolniej} we discuss how the proposed approach can be applied to other cases, in particular to systems with qutrit or ququart subsystems. Finally, we conclude in Sec. \ref{konkluzje}.
	
	\section{Preliminaries} \label{intro}
		
	Let $\mathcal{H}_{\bold{A}}=\mathcal{H}_{A_1}\otimes \cdots \otimes
	\mathcal{H}_{A_N}$ with $\mathcal{H}_{A_i}=\mathbb{C}^{d}$, where $\bold{A}=A_1
	A_2 \dots A_N$ label subsystems. A state $\ket{\psi}_{\bf{A}} \in
	\mathcal{H}_{\bold{A}}$ is said to be product if it can be written as 
	%
		$		\ket{\psi}_{\bold{A}}= \ket{\varphi_1}_{A_1}\otimes \ket{\varphi_2}_{A_2}
		\otimes \cdots \otimes \ket{\varphi_N}_{A_N}.$
	%
	If a state is not product it is called entangled. Product states belong to a
	wider class of biproduct vectors which can be written as
	%
		$		\ket{\psi}_{\bold{A}}=\ket{\varphi}_S \otimes \ket{\bar{\varphi}}_{\bar{S}}$
	%
	for some bipartition  $S|\bar{S}$ of the subsystems ($S$ and $\bar{S}$ are
	disjoint sets such that $S\cup \bar{S} = \bold{A}$). A state which is not
	biproduct is called genuinely multipartite entangled (GME). 
	Subspaces of multipartite Hilbert spaces whose all vectors are entangled are
	called completely entangled (CES) \cite{partha,bhat,walgate,cubitt,johnston-x,johnston-xx}; a special class of such
	subspaces is comprised of those composed only of GME states --- one calls them
	genuinely entangled subspaces (GESs)
	\cite{upb-to-ges,halder,approach,ent-of-ges,makuta,DemianowiczQuantum}. In what
	follows we will omit normalization factors in states.

	\subsection{(Genuinely) unextendible product bases}
	
	We now briefly recall the notions of unextendible product basis  \cite{upb,big-upb}
	and genuinely unextendible product basis \cite{upb-to-ges,MD-PRA}.
	
	\begin{defn}\label{upb-def}
		An unextendible product basis (UPB) is a set of product, mutually orthogonal
		vectors spanning  a proper subspace of a given multipartite Hilbert space with
		the property that there does not exist a product vector orthogonal to all the
		elements of the set.
	\end{defn}
	An exemplary UPB in $\mathbb{C}^3\otimes \mathbb{C}^3\otimes \mathbb{C}^3$ is
	given by the following set of nineteen vectors \cite{gupb-kutrity}: 
	\begin{align}
		&	\ket{\varphi_i}_{A_1}\ket{0}_{A_2}\ket{\psi_j}_{A_3},
		\ket{\varphi_i}_{A_1}\ket{\psi_j}_{A_2}\ket{2}_{A_3},
		\ket{2}_{A_1}\ket{\varphi_i}_{A_2}\ket{\psi_j}_{A_3},\nonumber\\
		&\ket{\psi_j}_{A_1}\ket{2}_{A_2}\ket{\varphi_i}_{A_3}, 
		\ket{\psi_j}_{A_1}\ket{\varphi_i}_{A_2}\ket{0}_{A_3},
		\ket{0}_{A_1}\ket{\psi_j}_{A_2}\ket{\varphi_i}_{A_3}, \ket{S}_{A_1A_2A_3} ,
	\end{align}
$(i,j)\in \{(0,1),(1,0),(1,1)$\}, where $\ket{\varphi_p}=\ket{1}+(-1)^p\ket{2}$,
	$\ket{\psi_p}=\ket{0}+(-1)^p\ket{1}$, and
	$\ket{S}=(\ket{0}+\ket{1}+\ket{2})^{\otimes 3}$ (so-called stopper state). 
	
	It follows from Definition \ref{upb-def} that the orthocomplement of a subspace
	spanned by a UPB is a CES. One should bear in mind that the converse does not hold in general.

	Our scope is on a particular class of UPBs.
	\begin{defn}
		A genuinely unextendible product basis (GUPB) is a set of product, mutually
		orthogonal vectors spanning  a proper subspace of a given multipartite Hilbert
		space with the property that there does not exist a \textbf{biproduct} vector orthogonal
		to all the elements of the set. In other words, a GUPB is a UPB for all
		bipartitions of the subsystems.
	\end{defn}
	
	It is evident that a GUPB defines a GES in the orthogonal subspace. An
	immediate realization is that GUPBs cannot exist in systems with a qubit
	subsystem since bipartite UPBs in systems $\mathbb{C}^2 \otimes \mathbb{C}^n$ do
	not exist \cite{big-upb}. Therefore, the smallest system one needs to consider is with
	$N=3$ and $d=3$.
	
	While it is still an open question whether GUPBs exist at all, 
	it has been recently shown that their sizes $\mathfrak{G}$ must be bounded from
	below as follows \cite{MD-PRA,kiribela}: 
	%
	%
	
		%
	\begin{align}\displaystyle \label{kiribela-bound}
		\mathfrak{G}(d,N) \ge  \frac{N d^{N-1}-1}{N-1}.
	\end{align}

	In some cases this bound can be slightly refined \cite{kiribela}. Clearly, the RHS of Eq. \eqref{kiribela-bound} could involve the ceiling function, however, GUPBs saturating the lower bound in the current form share a certain important graph-theoretic feature (see Sec. \ref{graph-character}), which will serve as a basis for our considerations.
	
	\subsection{Graphs}\label{graphs}

	We now move to a terse review of the relevant graph theory notions (see, e.g.,
	\cite{bondy-graph}).
	
	An undirected simple graph $G=(V,E)$ consists of a nonempty finite set $V$,
	whose elements are called vertices (or nodes), and a finite set $E$ of different
	unordered pairs of elements of $V$, whose elements are called edges. Any graph
	in the current paper is an undirected and simple one with nonempty $E$. We say
	that two vertices $v_1$ and $v_2$ of a graph are adjacent if there is an edge
	connecting them, i.e., $\{v_1, v_2\} \in E$. In such case, $v_1$ is said to be a
	neighbor of $v_2$ and vice versa. 
	The neighborhood $\mathcal{N}(v)$ of a given vertex $v$ in a graph is the set
	of all its neighbors. The degree $\deg(v)$ of a vertex $v$ in a graph is the
	size of its neighborhood, i.e., $\deg(v)=|\mathcal{N}(v)|$.
	A graph whose all vertices have the same degree is called regular. In
	particular, if this degree is $r$ then the graph is called $r$-regular.  An
	$(n-1)$-regular graph with $n$ vertices is said to be complete. In such graphs,
	all vertices are adjacent. Regular graphs will be of our main interest.
	
	An adjacency matrix of graph $G$, $\mathcal{A}(G)$, is the matrix whose
	$(i,j)$-th entry equals to $1$  if $\{v_i,v_j\} \in E$ and $0$ if $\{v_i,v_j\}
	\not\in E$. All the diagonal elements of $\mathcal{A}$ are equal to $0$.

	A clique of a graph is a subset of its vertices with the property that any two
	vertices are adjacent; oftentimes a clique is considered to be a graph itself
	and this is the approach we will employ. A $k$-clique, denoted henceforward by
	$C_k$, is a clique with $k$ vertices. The clique number $\omega (G)$ of a graph
	$G$ is given by the number of vertices in a maximum clique of the graph.

	Two vertices $v_1$ and $v_2$ are said to be connected if there is a path
	\footnote{A walk in graph $G=(V,E)$ is given by a (finite) sequence of edges
		$\{v_0,v_1\},\{v_1,v_2\},\dots, \{v_{m-1},v_m\}$ in which any two consecutive
		edges are either adjacent or the same. A walk gives a sequence of vertices
		$v_0v_1 \dots v_{m-1}v_m$. A walk with $v_0=v_m$ is said to be closed;
		otherwise, it is open. The number of edges in a walk is called its length. A
		walk with all edges distinct is called a trail. Further, if vertices $v_1,\dots,
		v_{m-1}$ are different then the trail is called a (simple) path.} from $v_1$ to $v_2$. A graph is 
	connected if any two of its vertices are connected. Otherwise, it is called
	disconnected.
	
	 A closed path
	with at least one edge is called a cycle. The length of a shortest cycle in a graph is called its girth, denoted $\girth$.
	
	A subgraph of $G$ is any graph, whose all vertices belong to $V$ and all edges
	belong to $E$. Given two graphs, $G_1=(V_1,E_1)$ and $G_2=(V_2,E_2)$, their
	union is the graph $G_1 \cup G_2 =(V_1 \cup V_2, E_1 \cup E_2)$. A decomposition
	of a graph $G=(V,E)$ is a set of its subgraphs $\{G_i=(V,E_i)\}_{i=1}^m$ such
	that $\bigcup _{i=1}^m G_i =G$. 
	The disjoint union of graphs will be denoted as $G_1+G_2$.
	
	Let $W \subseteq V$ and $E_W \subseteq E$ be the subset of all edges in $G$
	with both endpoints in  $W$. Then, $G[W]=(W,E_W)$ is called an induced subgraph
	of $G$.
	In an induced subgraph isomorphism problem, one is given two graphs 
	$G$ and $H$ and must determine whether there exists an induced subgraph of 
	$G$ that is isomorphic to $H$. In full generality, it is an NP-complete
	problem.

	\subsubsection{Orthogonal representations and orthogonality graphs}
	
	An orthogonal representation (OR) in dimension $d$ of a graph $G=(V,E)$ assigns
	vectors $\ket{v_i}\in \mathbb{C}^d$ to all vertices $v_i \in V$
	in such a manner that $\bra{v_i} v_j \rangle =0$ for adjacent vertices. The
	representation is said to be faithful if the vectors are orthogonal only for
	adjacent vertices. In fact, we will be exclusively interested in the existence
	of faithful ORs (FORs) for certain graphs. We will use 
	FOR(d) to denote a FOR in dimension $d$ of a graph. 
	An important problem in the field is finding   the minimal dimension
	$d_{\mathrm{min}}$ for which an orthogonal representation for a given graph
	exists. While it is in general very hard to find $d_{\mathrm{min}}$ \cite{zero-forcing,msr},  
	a simple useful lower bound is
	%
		$	d_{\mathrm{min}} \ge \omega(G).$
	%
	%

	Somewhat conversely, given a set of vectors $\{ \ket{\varphi_i} \}_{i=1}^{k}$
	from $\mathbb{C}^d$, their orthogonality graph (OG) is the graph $G=(V,E)$ with
	$V= \{v_1,v_2,\dots, v_{k}\}$ and $E=\{ \{v_i,v_j\}: \bra{\varphi_i}
	\varphi_j\rangle=0  \}$.

Notice that we consider here the notion of an orthogonal representation commonly used in quantum information theory, rather than the one introduced by Lov\'asz \cite{lovasz} and widely spread throughout the mathematical literature, whereby the orthogonality of vectors is imposed for  non-adjacent vertices (usually as an iff condition). The relation between these two approaches is evident: given graph $G$, the Lov\'asz-type representation of the complement of $G$ is the one we are interested in the present paper.

\subsubsection{Forbidden induced subgraph characterization}\label{fis-sekcja}

	In a forbidden graph characterization (see, e.g., \cite{forbidden1,forbidden2,forbidden3}), one specifies completely a family of
graphs through a set of subgraphs (or, more generally, substructures), called an
obstruction set, such that none of the graphs from the family
admits a subgraph  (or a relevant substructure) from the obstruction set. Our main tool will be a variant of this approach, namely a forbidden induced subgraph characterization of graphs admitting a FOR. It relies on the following simple, yet powerful observation.
\begin{obs}\label{fis}
Let $G$ be a graph without a FOR(d). Then any graph $H$ containing $G$ as an induced subgraph does not admit a FOR(d) either.	
\end{obs}

For a given dimension, our aim will be to identify a set of small graphs without FORs and then use Observation \ref{fis}.

	\subsubsection{Orthogonal representation of the square graph and the diamond graph}

We will repeatedly use the following simple result characterizing (faithful)
orthogonal representations of certain elementary graphs.

\begin{lem}\label{4cykl}
	(i) The square graph  has at least one repeating vector in
	a FOR(3). (ii) The diamond graph has exactly one repeating vector in a FOR(3). 
\end{lem}
\begin{proof} See Fig. \ref{4cycle-vecs} for the denotations.
	
	(i)	Let
	%
	$	\ket{w_1}=\ket{0}$ and $\ket{w_2}=\ket{1}.$	
%
It must then hold
%
	$\ket{w_3}= \ket{0}+ \alpha \ket{2}$ and $\ket{w_4}=\ket{1}+\beta \ket{2}$,
%
and these vectors must be orthogonal, implying that $\alpha \beta = 0$.

(ii) It is obvious that the orthogonal representation with 	$\ket{w_1}=\ket{0}$, $\ket{w_2}=\ket{1}$, $\ket{w_2}=\ket{w_4}=\dwa$, which is faithful, is unique up to a unitary in $d=3$.
\begin{figure}[h]
	\centering		
	\begin{tikzpicture}[scale=0.5]
		\begin{scope}
			\node[circle, draw, fill=yellow, minimum size=8pt,inner sep=2pt] (v1) at
			({45+0*90}:4cm) {$\ket{w_1}=\ket{0}$};
			\node[circle, draw, fill=green, minimum size=8pt,inner sep=2pt] (v2) at
			({45+1*90}:4cm) {$\ket{w_2}=\ket{1}$};
			\node[circle, draw, fill=yellow, minimum size=8pt,inner sep=-2pt] (v3) at
			({45+2*90}:4cm) {$\begin{array}{c}
					\ket{w_3}\!\!=\!\!\ket{0}\!\!+\!\alpha \ket{2} \\ (\alpha \beta\! =\! 0) \\
				\end{array}$};
			\node[circle, draw, fill=green, minimum size=8pt,inner sep=-2pt] (v4) at
			({45+3*90}:4cm) {$\begin{array}{c}
					\ket{w_4}\!=\!\!\ket{1}\!\!+\!\!\beta \ket{2} \\ (\alpha\beta\! =\! 0) 
				\end{array}$};
			%
			\draw (v1) -- (v2) -- (v3) -- (v4) -- (v1);
		\end{scope}
		\begin{scope}[xshift=14cm]
			\node[circle, draw, fill=white, minimum size=8pt,inner sep=2pt] (v1) at
			({45+0*90}:4cm) {$\ket{w_1}=\ket{0}$};
			\node[circle, draw, fill=brown, minimum size=8pt,inner sep=2pt] (v2) at
			({45+1*90}:4cm) {$\ket{w_2}=\ket{2}$};
			\node[circle, draw, fill=white, minimum size=8pt,inner sep=2pt] (v3) at
			({45+2*90}:4cm) {
				$\ket{w_3}= \ket{1}$};
			\node[circle, draw, fill=brown, minimum size=8pt,inner sep=2pt] (v4) at
			({45+3*90}:4cm) {
				$\ket{w_4}=\ket{2}$};
			
			\draw (v1) -- (v2) -- (v3) -- (v4) -- (v1);
			\draw (v1) -- (v3);
		\end{scope}
	\end{tikzpicture}
	\caption{(left) There is necessarily at least on repeating vector in a FOR(3) of the
		square graph (Lemma \ref{4cykl}) -- for $w_2$ and $w_4$ (green) or $w_1$ and
		$w_3$ (yellow), or both. (right) There is necessarily exactly one pair of equal vectors, for $w_2$ and $w_4$ (brown), in a FOR(3) of the
		diamond graph.
	}\label{4cycle-vecs}
\end{figure}
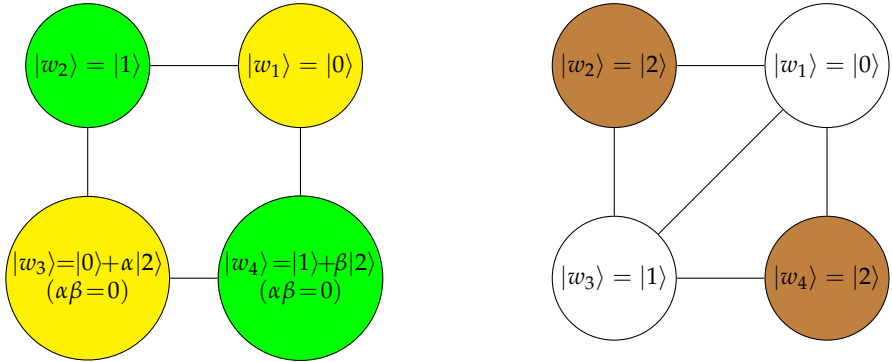
\end{proof}

Notice that the square graph has in fact a pair of identical vectors in any OR
in $d=3$.
Also, it is easy to observe that the necessary presence of repeated vectors  does not hold for $d\ge 4$; we
could, for example, set  $\ket{w_3}=\ket{0}+\ket{2}+\ket{3}$ and
$\ket{w_4}=\ket{1}+\ket{2}-\ket{3}$. Similarly for the diamond graph, where we could have $\ket{w_2}=\dwa$ and $\ket{w_4}=\dwa + \ket{3}$.

We will later use the square graph and the diamond graph to construct small graphs without FORs.

\subsection{Graph-theoretic characterization of GUPBs} \label{graph-character}

Let $\{\ket{\varphi_1^{(j)}}_{A_1}\otimes \ket{\varphi_2^{(j)}}_{A_2} \otimes
\cdots \otimes \ket{\varphi_N^{(j)}}_{A_N}\}_{j=1}^k$,
$\ket{\varphi_i^{(j)}}_{A_i} \in \mathbb{C}^d$, be a UPB. The number $k$ is called the size or the cardinality of a UPB.
One defines local orthogonality graphs (LOGs), $G_i$, $i=1,2,\dots,N$, for this
set as follows: $G_i=(V,E_i)$, where $V=\{v_1,v_2,\dots,v_k\}$ and $E_i=\{
\{v_m,v_n\}: \bra{\varphi_i^{(m)}}\varphi_i^{(n)} \rangle _{A_i}  =0\}$. LOGs
represent orthogonality of vectors in a UPB as seen by individual subsystems.
Clearly, the union of LOGs $G_i$, which is an OG for the UPB, is a complete
graph with $k$ vertices, $\bigcup_{i=1}^N G_i=C_k$.

There is a simple necessary condition regarding spanning (saturation) properties of local vectors of a GUPB \cite{MD-PRA,kiribela}.
\begin{fakt} \label{general-fakt}
	A $k$-element $N$-qudit GUPB must share the following property: for any subsystem $A_m$, any $(k-d^{N-1}+1)$-tuple of local vectors
	$\ket{\varphi_m^{(j)}}_{A_m}$ span the full $d$-dimensional subspace.
\end{fakt}

A similar condition may also be formulated for the composed subsystems (containing up to $N-1$ subsystems). All of them follow directly from the observation already made by Bennett et al. \cite{upb}.

Further, we have the following result due to Shi et al. \cite{kiribela}. 

\begin{fakt} \label{orto-regular}
	LOGs for GUPBs saturating the  lower bound in Eq. \eqref{kiribela-bound} are all
	$\left(\mathfrak{G}-d^{N-1}\right)$-regular.
\end{fakt}
This powerful result  places, in a sense, GUPBs on equal footing with UPBs, as
a similar feature holds \cite{kiribela} for minimal UPBs of size $N(d-1)+1$
\cite{upb}. Furthermore, LOGs for those minimal GUPBs have non overlapping edge sets as can be realized through a direct count.
Notably, a general result about the degrees of vertices in UPBs was
also demonstrated in \cite{kiribela}.

\section{Minimal three-qutrit GUPBs}\label{main-section}

We now move the main part of the present paper and  concentrate on the specific case of  three-qutrit ($N=3$
and $d=3$) GUPBs. According to Eq. \eqref{kiribela-bound}, their minimal
permissible cardinality is (notice that this is also the smallest theoretically 
possible size of a GUPB in any setup)
\begin{align}
	\mathfrak{G}_{\min}(3,3)=13.
\end{align}

By Fact \ref{orto-regular}, LOGs for a minimal GUPB, which are now $13$-vertex graphs, must be $4$-regular. By Fact \ref{general-fakt}, for any subsystem $A_m$, any $5$-tuple of local vectors
$\ket{\varphi_m^{(j)}}_{A_m}$ in their FORs(3)  span the full three-dimensional space.
If a set of product vectors fulfills these conditions and in addition for any pair of subsystems $A_mA_n$, any $9$-tuple of  vectors $\ket{\varphi_m^{(j)}}_{A_m}\otimes \ket{\varphi_n^{(j)}}_{A_n}$ span the full
$9$-dimensional subspace then this set is a GUPB.  This motivated Shi et al \cite{kiribela} to propose a concrete route for the search of minimal three-qutrit GUPBs . 
First, the complete graph $C_{13}$  --- representing an OG for
a GUPB --- needs to be decomposed into three $4$-regular graphs, which play the
role of LOGs. Second, one needs to find FORs for these LOGs sharing the single and two-partite spanning 
properties discussed above. 

On the other hand, to disprove  the existence of the minimal GUPBs 
it suffices to show that (a)  certain candidate graphs do not have a FOR(3) and/or (b)
FORs(3) for the remaining graphs do not lead to the fulfillment of the  condition
of Fact \ref{general-fakt}. We will follow this path.
Observe that the ORs  must necessarily be
faithful as otherwise they would lead back to OGs which are not $4$-regular
anymore, contradicting Fact \ref{orto-regular}.

\subsection{Candidate LOGs classification}\label{kandydaci}

We need to analyze all non-isomorphic $4$-regular graphs with $13$ vertices.
Their total number is  $10\:786$ \cite{all-quartic} and they can be classified
based on whether they are connected or not. More precisely, we can divide the
graphs into:
\begin{itemize}
	\item  disconnected (Sec. \ref{rozlaczne}) -- $8$ graphs
	\cite{disconnected-quartic},
	\item  connected  (Sec. \ref{polaczone}) -- $10\:778$ graphs
	\cite{connected-quartic}.
\end{itemize}

The connected graphs will henceforth be denoted by $M_i$, where $i$ is the index
of a graph consistent with the enumeration obtained with GENREG software
\cite{meringer,meringer2}. A file with the adjacency matrices for these graphs
is available at \cite{pliki}, where all other necessary resources can also be
found.

We note that the connected 
graphs exhibit the following properties: (i) 31 graphs have girth $4$ 
(precisely, graphs $M_{10748}$ through $M_{10778}$)
and (ii) the remaining graphs have girth $3$ (they have $3$-cliques, i.e.,
triangles, as subgraphs). The fact that there are no graphs with larger girths is what makes the square and the diamond graphs of particular relevance in our analyses. 

\subsection{Disconnected $4$-regular graphs with $13$ vertices}\label{rozlaczne}

We begin with the simpler case of disconnected graphs. As noted in the previous
section, there are $8$ such graphs. 
They are of two types in terms of the number of vertices in the disconnected
groups:
\begin{itemize}
	\item[(a)]  type I ($6$ graphs): $5+8$ vertices, 
	\item[(b)] type II ($2$ graphs): $6+7$ vertices.
\end{itemize}

Let us analyze these graphs in detail.

(a) Type I. There are six $4$-regular graphs with $8$ vertices, however, there
is only one such graph with $5$ vertices, namely $C_5$ (Fig. \ref{5-6-vertices}),
which obviously does not have a FOR(3). In turn, this case gets discarded.

\begin{figure}[h]
	\centering

	\begin{tikzpicture}[scale=0.6]
		\begin{scope}
			\node[circle, draw, fill=white, minimum size=12pt] (v1) at ({90+0*72}:3cm)
			{$v_1$};
			\node[circle, draw, fill=white, minimum size=12pt] (v2) at ({90+1*72}:3cm)
			{$v_2$};
			\node[circle, draw, fill=white, minimum size=12pt] (v3) at ({90+2*72}:3cm)
			{$v_3$};
			\node[circle, draw, fill=white, minimum size=12pt] (v4) at ({90+3*72}:3cm)
			{$v_4$};
			\node[circle, draw, fill=white, minimum size=12pt] (v5) at ({90+4*72}:3cm)
			{$v_5$};
			%
			\draw (v1) -- (v2) -- (v3) -- (v4) -- (v5) -- (v1);
			\draw (v1) -- (v3) -- (v5) -- (v2) -- (v4) -- (v1);
		\end{scope}
		\begin{scope}[xshift=4.5cm,yshift=2.5cm,scale=0.8]
			\draw[red, line width=2.5pt, line cap=round, line join=round]
			(-0.5,0.95) -- (0.3,0) ;
			\draw[red, line width=2.5pt, line cap=round, line join=round]	(-0.5,0) --
			(0.3,0.95);
		\end{scope}
		\begin{scope}[xshift=12cm]
			\node[circle, draw, fill=brown, minimum size=12pt] (v1) at
			({90+0*60}:3cm) {$v_1$};
			\node[circle, draw, fill=cyan, minimum size=12pt] (v2) at ({90+1*60}:3cm)
			{$v_2$};
			\node[circle, draw, fill=lime, minimum size=12pt] (v3) at ({90+2*60}:3cm)
			{$v_3$};
			\node[circle, draw, fill=lime, minimum size=12pt] (v4) at ({90+3*60}:3cm)
			{$v_4$};
			\node[circle, draw, fill=cyan, minimum size=12pt] (v5) at ({90+4*60}:3cm)
			{$v_5$};
			\node[circle, draw, fill=brown, minimum size=12pt] (v6) at
			({90+5*60}:3cm) {$v_6$};
			%
			\draw (v1) -- (v2) -- (v3) -- (v5) -- (v6) -- (v4) -- (v2);
			\draw (v1) -- (v3) -- (v6);
			\draw (v1) -- (v4) -- (v5);
			\draw (v2) -- (v6);
			\draw (v1) -- (v5);
		\end{scope}
		\begin{scope}[xshift=16cm,yshift=2.5cm,scale=0.9]
			\draw[ufogreen, line width=2.5pt, line cap=round, line join=round]
			(-0.3,0.4) -- (0,0) -- (0.6,0.95);
		\end{scope}
	\end{tikzpicture}
	\caption{(left) $C_5$ is the unique $4$-regular graph on $5$ vertices.
		Trivially, no FOR(3) exists for this graph (indicated by the cross mark).
		(right) The unique $4$-regular  graph $D_6$ on $6$ vertices. This graph admits a
		FOR(3) (indicated by the checkmark) with necessarily repeating vectors for the
		colored vertices.
	}\label{5-6-vertices}
\end{figure}
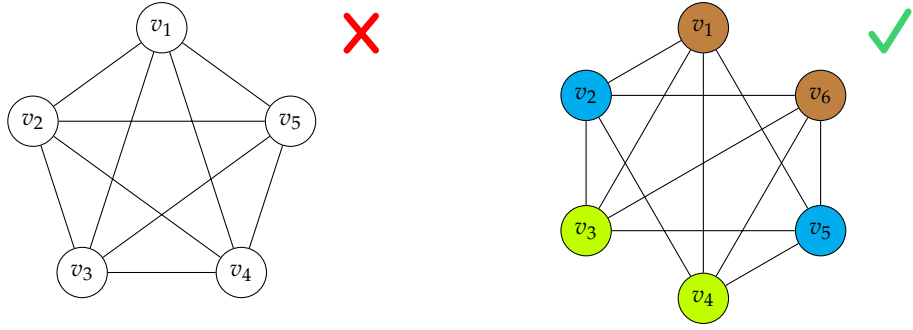

(b) Type II.  Graph $D_6$ (Fig. \ref{5-6-vertices}) is the unique graph with six vertices and there are two
graphs with $7$ vertices, $D_{7,a}$ and $D_{7,b}$ (Fig. \ref{7-vertices}). 

Let us start with $D_6$. 
Considering the diamond subgraph with vertices $(v_1,v_2,v_3,v_4)$ 
immediately leads to 
$\repr{3}=\repr{4}$ (lime vertices in Fig. \ref{5-6-vertices}). In a similar manner we get $\repr{2}=\repr{5}$ (blue) and
$\repr{1}=\repr{6}$ (brown). All three conditions can be satisfied simultaneously  and we conclude that
this graph does have a FOR(3) and it is unique up to a unitary.

We now consider $D_{7,a}$. Let $\ket{v_1}=\ket{0}$ and $\ket{v_2}=\ket{1}$. By Lemma \ref{4cykl}, one
immediately obtains  $\repr{3}=\repr{4}=\repr{5}=\ket{2}$ (brown vertices in Fig. \ref{7-vertices}). Further, $\repr{6}=\alpha
\ket{0}+\beta\ket{1}$, and $\repr{7}=\beta^* \ket{0} -\alpha^* \ket{1}$ with
$\alpha,\beta \ne 0$. This constitutes a FOR(3), which is unique up to a
unitary. 

In case of $D_{7,b}$ there are several ways to arrive at a contradiction. For
example, we readily get $\repr{3}=\repr{4}$, which is, however, impossible in a
FOR as $\mathcal{N}(v_3) \ne \mathcal{N}(v_4)$. In turn, $D_{7,b}$ does not admit
a FOR(3).  In the framework of forbidden induced subgraph characterization (see Sec. \ref{polaczone}), we identify $D_{7,b}$ as one admitting the forbidden subgraph $H_5$ (e.g., on vertices $v_1,v_2,v_3,v_6,v_7$).

Concluding, only one disconnected graph, $D_6+D_{7,a}$ (type II), has a
FOR(3). Nevertheless, this candidate graph gets discarded as there is twice the
same vector in a FOR(3) for $D_6$ and thrice for $D_{7,a}$, violating the  condition
of Fact \ref{general-fakt}.
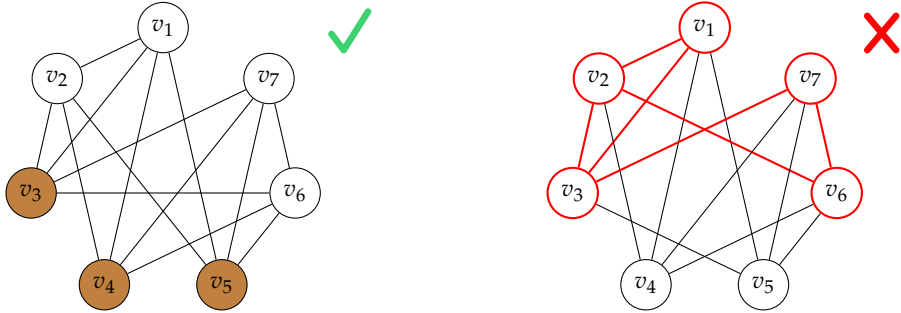
\begin{figure}[h!]
	\centering
	\begin{tikzpicture}[scale=0.6]
		\begin{scope}
			\GraphInit[vstyle=Classic]
			\SetGraphUnit{2.5}
			\node[circle, draw, fill=white, minimum size=12pt] (v1) at ({90+0*360/7}:3cm)
			{$v_1$};
			\node[circle, draw, fill=white, minimum size=12pt] (v2) at ({90+1*360/7}:3cm)
			{$v_2$};
			\node[circle, draw, fill=brown, minimum size=12pt] (v3) at ({90+2*360/7}:3cm)
			{$v_3$};
			\node[circle, draw, fill=brown, minimum size=12pt] (v4) at ({90+3*360/7}:3cm)
			{$v_4$};
			\node[circle, draw, fill=brown, minimum size=12pt] (v5) at ({90+4*360/7}:3cm)
			{$v_5$};
			\node[circle, draw, fill=white, minimum size=12pt] (v6) at ({90+5*360/7}:3cm)
			{$v_6$};
			\node[circle, draw, fill=white, minimum size=12pt] (v7) at ({90+6*360/7}:3cm)
			{$v_7$};
			\draw (v1) -- (v2);			\draw (v1) -- (v3);
			\draw (v1) -- (v4);			\draw (v1) -- (v5);
			\draw (v2) -- (v3);			\draw (v2) -- (v4);
			\draw (v2) -- (v5);			\draw (v3) -- (v6);
			\draw (v3) -- (v7);			\draw (v4) -- (v6);
			\draw (v4) -- (v7);			\draw (v5) -- (v6);
			\draw (v5) -- (v7);			\draw (v6) -- (v7);		
	\end{scope}
	\begin{scope}[xshift=4cm,yshift=2.5cm,scale=0.9]
		\draw[ufogreen, line width=2.5pt, line cap=round, line join=round]
		(-0.3,0.4) -- (0,0) -- (0.6,0.95);
	\end{scope}
	%
		\begin{scope}[xshift=12cm]
			\GraphInit[vstyle=Classic]
			\SetGraphUnit{2.5}
			\node[circle, draw= red,thick, fill=white, minimum size=12pt] (v1) at ({90+0*360/7}:3cm)
			{$v_1$};
			\node[circle, draw=red,thick, fill=white, minimum size=12pt] (v2) at ({90+1*360/7}:3cm)
			{$v_2$};
			\node[circle, draw=red,thick, fill=white, minimum size=12pt] (v3) at
			({90+2*360/7}:3cm) {$v_3$};
			\node[circle, draw, fill=white, minimum size=12pt] (v4) at
			({90+3*360/7}:3cm) {$v_4$};
			\node[circle, draw, fill=white, minimum size=12pt] (v5) at ({90+4*360/7}:3cm)
			{$v_5$};
			\node[circle, draw=red,thick, fill=white, minimum size=12pt] (v6) at ({90+5*360/7}:3cm)
			{$v_6$};
			\node[circle, draw=red,thick, fill=white, minimum size=12pt] (v7) at ({90+6*360/7}:3cm)
			{$v_7$};	
			\draw[red,thick] (v1) -- (v2);			\draw[red,thick] (v1) -- (v3);
			\draw (v1) -- (v4);			\draw (v1) -- (v5);
			\draw[red,thick] (v2) -- (v3);			\draw (v2) -- (v4);
			\draw[red,thick] (v2) -- (v6);			\draw (v3) -- (v5);
			\draw[red,thick] (v3) -- (v7);			\draw (v4) -- (v6);
			\draw (v4) -- (v7);			\draw (v5) -- (v6);
			\draw (v5) -- (v7);			\draw[red,thick] (v6) -- (v7);
		\end{scope}
		\begin{scope}[xshift=16cm,yshift=2.5cm,scale=0.8]
			\draw[red, line width=2.5pt, line cap=round, line join=round]
			(-0.5,0.95) -- (0.3,0) ;
			\draw[red, line width=2.5pt, line cap=round, line join=round]	(-0.5,0) --
			(0.3,0.95);
		\end{scope}
	\end{tikzpicture}
	\caption{Two non-isomorphic $4$-regular graphs on $7$ vertices: $D_{7,a}$
		(left) and $D_{7,b}$ (right). $D_{7,a}$ admits a FOR(3) with the same vectors
		for $v_3, v_4,v_5$; there is no FOR(3) for $D_{7,b}$ as $H_5$ is its induced subgraph (highlighted in red).} \label{7-vertices}
\end{figure} 

For larger cases, there will be significantly more disconnected graphs (also with more vertices) and their analysis should involve the approach put forward in the next section.

\subsection{Connected $4$-regular graphs with $13$ vertices} \label{polaczone}

In this section we analyze all connected $4$-regular graphs $M_i$,
$i=1,2,\dots, 10\:778$ (cf. Sec. \ref{kandydaci}). With this purpose, we will
use the approach of forbidden induced subgraph characterization (see Sec. \ref{fis-sekcja}). 
 We will identify a small number of graphs (with up to six
vertices) not having a FOR(3) (an obstruction set). By Observation \ref{fis}, if any of these graphs
exists as an induced subgraph in graph $M_i$, then this graph obviously does not
have a FOR(3) either and it gets discarded as a candidate LOG.
We will be able to narrow the set of graphs having a FOR(3) down to a single
graph. This graph will further be ruled out by Fact \ref{general-fakt}
as not having the required property.

Since the problem we deal here with is still rather small, we can employ a
brute-force strategy for the relevant induced subgraph isomorphism problems. For
a given subgraph, we verify whether its adjacency matrix $\mathcal{A}$ or $P
\mathcal{A} P^T$ ($P$ is a permutation matrix)  is a submatrix along the
diagonal of the adjacency matrix of a graph from the analyzed set. This is
performed for all subgraphs (see the list below) and all $M_i$'s. This is clearly very inefficient but it is particularly easy to implement and it is sufficient for the considered problem. We also implemented an alternative, more efficient procedure, using an in-built function for the subgraph isomorphism problem of the software we used \cite{pliki}.

\subsubsection{Forbidden induced subgraphs} \label{forbidden-induced}

We will consider the following subgraphs:
\begin{itemize}
	\item $4$-clique $C_4$ (Fig. \ref{klika}),
	\item house graph $H_5$ (Fig. \ref{house}),
	\item kite graph $K_5$ (Fig. \ref{kite}),
	\item $A$-graph $A_6$ (Fig. \ref{A-graph}).
\end{itemize}
Their adjacency matrices are as follows:
\begin{align}
	\mathcal{A}(C_4)= \left(
	\begin{array}{cccc}
		0 & 1 & 1 & 1 \\
		1 &0 & 1 & 1 \\
		1 & 1 & 0 & 1 \\
		1 & 1 & 1 & 0
	\end{array} \right)\!\! , \:\:
	\mathcal{A}(H_5) =	\left(
	\begin{array}{ccccc}
		0 & 1 & 1 & 1 & 0 \\
		1 & 0 & 1 & 0 & 1 \\
		1 & 1 & 0 & 0 & 0 \\
		1 & 0 & 0 & 0 & 1 \\
		0 & 1 & 0 & 1 & 0
	\end{array}
	\right)\!\! , \:\:
	\mathcal{A}(K_5)=	\left(
	\begin{array}{ccccc}
		0 & 1 & 1 & 1 &0 \\
		1 & 0 & 1 & 1 &0\\
		1 & 1 & 0 & 0 &1\\
		1 & 1 & 0 & 0& 0\\
		0 &0 & 1 & 0& 0
	\end{array} \right)\!\! , \:\:
	\mathcal{A}(A_6)=\left(
	\begin{array}{cccccc}
		0 & 0 & 0 & 1 & 1 & 1 \\
		0 & 0 & 0 & 0 & 1 & 1 \\
		0 & 0 & 0 & 0 & 0 & 1 \\
		1 & 0 & 0 & 0 & 0 & 0 \\
		1 & 1 & 0 & 0 & 0 & 0 \\
		1 & 1 & 1 & 0 & 0 & 0 \\
	\end{array}
	\right).
\end{align}
\begin{figure}[htp]
	\begin{tikzpicture}[scale=0.45]
		\centering
		\begin{scope}[scale=0.9]
			\node[circle, draw, fill=white, minimum size=12pt] (w1) at ({45}:3cm)
			{$w_1$};
			\node[circle, draw, fill=white, minimum size=12pt] (w2) at ({135}:3cm)
			{$w_2$};
			\node[circle, draw, fill=white, minimum size=12pt] (w3) at ({225}:3cm)
			{$w_3$};
			\node[circle, draw, fill=white, minimum size=12pt] (w4) at ({315}:3cm)
			{$w_4$};
			%
			\draw (w1) -- (w2) -- (w3) -- (w4) -- (w1);
			\draw (w2) -- (w4);
			\draw (w1) -- (w3);
		\end{scope}
		\begin{scope}[xshift=12cm,scale=0.9]
			\GraphInit[vstyle=Classic]
			\SetGraphUnit{2.5}
			\node[circle, draw=red, thick, fill=white, minimum size=10pt, inner sep=1.8pt] (v1) at
			({90+0*360/13}:6.5cm) {$v_1$};
			\node[circle, draw=red,thick, fill=white, minimum size=12pt, inner sep=1.8pt] (v2) at
			({90+1*360/13}:6.5cm) {$v_2$};
			\node[circle, draw=red,thick, fill=white, minimum size=12pt, inner sep=1.8pt] (v3) at
			({90+2*360/13}:6.5cm) {$v_3$};
			\node[circle, draw=red, thick, fill=white, minimum size=12pt, inner sep=1.8pt] (v4) at
			({90+3*360/13}:6.5cm) {$v_4$};
			\node[circle, draw, fill=white, minimum size=12pt, inner sep=1.8pt] (v5) at
			({90+4*360/13}:6.5cm) {$v_5$};
			\node[circle, draw, fill=white, minimum size=12pt, inner sep=1.8pt] (v6) at
			({90+5*360/13}:6.5cm) {$v_6$};
			\node[circle, draw, fill=white, minimum size=12pt, inner sep=1.8pt] (v7) at
			({90+6*360/13}:6.5cm) {$v_7$};
			\node[circle, draw, fill=white, minimum size=12pt, inner sep=1.8pt] (v8) at
			({90+7*360/13}:6.5cm) {$v_8$};	
			\node[circle, draw, fill=white, minimum size=12pt, inner sep=1.8pt] (v9) at
			({90+8*360/13}:6.5cm) {$v_9$};	
			\node[circle, draw, fill=white, minimum size=12pt,inner sep=0.6pt] (v10) at
			({90+9*360/13}:6.5cm) {$v_{10}$};	
			\node[circle, draw, fill=white, minimum size=12pt,inner sep=0.6pt] (v11) at
			({90+10*360/13}:6.5cm) {$v_{11}$};	
			\node[circle, draw, fill=white, minimum size=12pt,inner sep=0.6pt] (v12) at
			({90+11*360/13}:6.5cm) {$v_{12}$};	
			\node[circle, draw, fill=white, minimum size=12pt,inner sep=0.6pt] (v13) at
			({90+12*360/13}:6.5cm) {$v_{13}$};		
			\draw[red,thick] (v1) -- (v2); \draw[red,thick](v1) -- (v3);
			\draw[red,thick](v1) -- (v4); \draw(v1) -- (v5);
			\draw[red,thick] (v2) -- (v3); \draw[red,thick](v2) -- (v4); \draw(v2) --
			(v5);
			\draw [red,thick](v3) -- (v4); \draw(v3) -- (v5);
			\draw (v4) -- (v6);
			\draw (v5) -- (v6);
			\draw (v6) -- (v7); \draw(v6) -- (v8);
			\draw (v7) -- (v9); \draw(v7) -- (v10); \draw(v7) -- (v11);
			\draw (v8) -- (v9); \draw(v8) -- (v10); \draw(v8) -- (v11);
			\draw (v9) -- (v12); \draw(v9) -- (v13);
			\draw (v10) -- (v12);\draw (v10) -- (v13);
			\draw (v11) -- (v12); \draw(v11) -- (v13);
			\draw (v12) -- (v13);
		\end{scope}
		\begin{scope}[xshift=19.5cm,yshift=5cm,scale=1.2]
			\draw[red, line width=2.5pt, line cap=round, line join=round]
			(-0.5,0.95) -- (0.3,0) ;
			\draw[red, line width=2.5pt, line cap=round, line join=round]	(-0.5,0) --
			(0.3,0.95);
		\end{scope}
	\end{tikzpicture}
	\caption{(left) $4$-clique  $C_4$. Graphs with $4$-cliques  obviously do not
		have FORs(3). (right) $C_4$ highlighted in red in $M_{4}$. } \label{klika}
\end{figure}
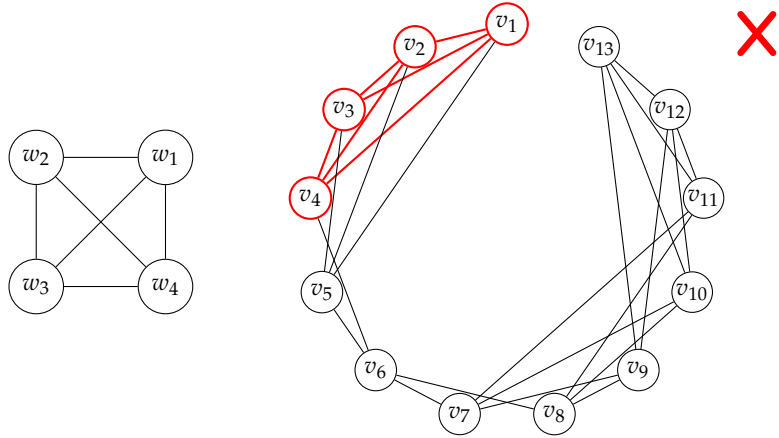
\begin{figure}[htp]
	\begin{tikzpicture}[scale=0.45]
		\begin{scope}[yshift=-0.8cm,scale=0.7]
			\node[circle, draw, fill=white, minimum size=12pt] (w1) at ({90+0*72}:5cm)
			{$w_1$};
			\node[circle, draw, fill=white, minimum size=12pt] (w2) at ({135}:3cm)
			{$w_2$};
			\node[circle, draw, fill=white, minimum size=12pt] (w3) at ({225}:3cm)
			{$w_3$};
			\node[circle, draw, fill=white, minimum size=12pt] (w4) at ({315}:3cm)
			{$w_4$};
			\node[circle, draw, fill=white, minimum size=12pt] (w5) at ({45}:3cm)
			{$w_5$};
			%
			\draw (w1) -- (w2) -- (w3) -- (w4) -- (w5) -- (w1);
			\draw (w2) -- (w5);
		\end{scope}
		\begin{scope}[xshift=12cm,scale=0.9]
			\node[circle, draw=red,thick, fill=white, minimum size=10pt, inner sep=1.8pt] (v1) at
			({90+0*360/13}:6.5cm) {$v_1$};
			\node[circle, draw=red,thick, fill=white, minimum size=12pt, inner sep=1.8pt] (v2) at
			({90+1*360/13}:6.5cm) {$v_2$};
			\node[circle, draw=red,thick, fill=white, minimum size=12pt, inner sep=1.8pt] (v3) at
			({90+2*360/13}:6.5cm) {$v_3$};
			\node[circle, draw=red,thick, fill=white, minimum size=12pt, inner sep=1.8pt] (v4) at
			({90+3*360/13}:6.5cm) {$v_4$};
			\node[circle, draw, fill=white, minimum size=12pt, inner sep=1.8pt] (v5) at
			({90+4*360/13}:6.5cm) {$v_5$};
			\node[circle, draw=red,thick, fill=white, minimum size=12pt, inner sep=1.8pt] (v6) at
			({90+5*360/13}:6.5cm) {$v_6$};
			\node[circle, draw, fill=white, minimum size=12pt, inner sep=1.8pt] (v7) at
			({90+6*360/13}:6.5cm) {$v_7$};
			\node[circle, draw, fill=white, minimum size=12pt, inner sep=1.8pt] (v8) at
			({90+7*360/13}:6.5cm) {$v_8$};	
			\node[circle, draw, fill=white, minimum size=12pt, inner sep=1.8pt] (v9) at
			({90+8*360/13}:6.5cm) {$v_9$};	
			\node[circle, draw, fill=white, minimum size=12pt,inner sep=0.6pt] (v10) at
			({90+9*360/13}:6.5cm) {$v_{10}$};	
			\node[circle, draw, fill=white, minimum size=12pt,inner sep=0.6pt] (v11) at
			({90+10*360/13}:6.5cm) {$v_{11}$};	
			\node[circle, draw, fill=white, minimum size=12pt,inner sep=0.6pt] (v12) at
			({90+11*360/13}:6.5cm) {$v_{12}$};	
			\node[circle, draw, fill=white, minimum size=12pt,inner sep=0.6pt] (v13) at
			({90+12*360/13}:6.5cm) {$v_{13}$};		
			\draw[red,thick] (v1) -- (v2);\draw[red,thick]  (v1) -- (v3);\draw[red,thick] 
			(v1) -- (v4);
			\draw (v1) -- (v5);\draw[red,thick]  (v2) -- (v3);
			\draw[red,thick]  (v2) -- (v6);\draw (v2) -- (v7);
			\draw (v3) -- (v8);\draw (v3) -- (v9);\draw (v4) -- (v5);
			\draw[red,thick]  (v4) -- (v6);\draw (v4) -- (v7);\draw (v5) -- (v8);\draw
			(v5) -- (v9);
			\draw (v6) -- (v10);\draw (v6) -- (v11);\draw (v7) -- (v10);\draw (v7) --
			(v11);
			\draw (v8) -- (v12);\draw (v8) -- (v13);\draw (v9) -- (v12);\draw (v9) --
			(v13);
			\draw (v10) -- (v12);\draw (v10) -- (v13);\draw (v11) -- (v12);\draw (v11) --
			(v13);
		\end{scope}
		\begin{scope}[xshift=19.5cm,yshift=5cm,scale=1.2]
			\draw[red, line width=2.5pt, line cap=round, line join=round]
			(-0.5,0.95) -- (0.3,0) ;
			\draw[red, line width=2.5pt, line cap=round, line join=round]	(-0.5,0) --
			(0.3,0.95);
		\end{scope}
	\end{tikzpicture}
	\caption{(left) House graph $H_5$ --- a $5$-vertex graph without a FOR(3) [Lemma \ref{grafy-bez-for} (b)]. (right) $H_5$
		highlighted in $M_{9048}$.
	}\label{house}
\end{figure}
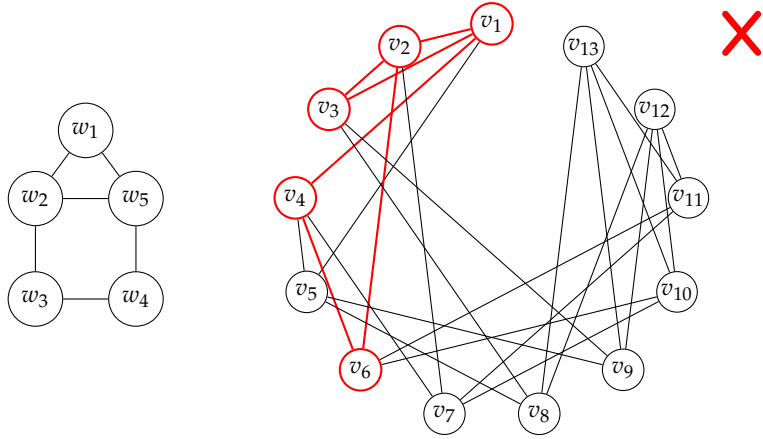
\begin{figure}[h]
	\centering
	\begin{tikzpicture}[scale=0.45]
		\begin{scope}[yshift=-0.7cm,scale=0.7]
			\node[circle, draw, fill=white, minimum size=12pt] (w1) at (0,-4) {$w_1$};
			\node[circle, draw, fill=white, minimum size=12pt] (w2) at (0,0) {$w_2$};
			\node[circle, draw, fill=white, minimum size=12pt] (w3) at ({45}:4cm)
			{$w_3$};
			\node[circle, draw, fill=white, minimum size=12pt] (w4) at (0,5.7) {$w_4$};
			\node[circle, draw, fill=white, minimum size=12pt] (w5) at ({135}:4cm)
			{$w_5$};
			%
			\draw (w1) -- (w2) -- (w3) -- (w4) -- (w5) ;
			\draw (w2) -- (w5);
			\draw (w3) -- (w5);
		\end{scope}
		\begin{scope}[xshift=12cm,scale=0.9]
			\node[circle, draw=red,thick, fill=white, minimum size=10pt, inner sep=1.8pt] (v1) at
			({90+0*360/13}:6.5cm) {$v_1$};
			\node[circle, draw=red,thick, fill=white, minimum size=12pt, inner sep=1.8pt] (v2) at
			({90+1*360/13}:6.5cm) {$v_2$};
			\node[circle, draw=red,thick, fill=white, minimum size=12pt, inner sep=1.8pt] (v3) at
			({90+2*360/13}:6.5cm) {$v_3$};
			\node[circle, draw=red,thick, fill=white, minimum size=12pt, inner sep=1.8pt] (v4) at
			({90+3*360/13}:6.5cm) {$v_4$};
			\node[circle, draw, fill=white, minimum size=12pt, inner sep=1.8pt] (v5) at
			({90+4*360/13}:6.5cm) {$v_5$};
			\node[circle, draw, fill=white, minimum size=12pt, inner sep=1.8pt] (v6) at
			({90+5*360/13}:6.5cm) {$v_6$};
			\node[circle, draw, fill=white, minimum size=12pt, inner sep=1.8pt] (v7) at
			({90+6*360/13}:6.5cm) {$v_7$};
			\node[circle, draw=red,thick, fill=white, minimum size=12pt, inner sep=1.8pt] (v8) at
			({90+7*360/13}:6.5cm) {$v_8$};	
			\node[circle, draw, fill=white, minimum size=12pt, inner sep=1.8pt] (v9) at
			({90+8*360/13}:6.5cm) {$v_9$};	
			\node[circle, draw, fill=white, minimum size=12pt,inner sep=0.6pt] (v10) at
			({90+9*360/13}:6.5cm) {$v_{10}$};	
			\node[circle, draw, fill=white, minimum size=12pt,inner sep=0.6pt] (v11) at
			({90+10*360/13}:6.5cm) {$v_{11}$};	
			\node[circle, draw, fill=white, minimum size=12pt,inner sep=0.6pt] (v12) at
			({90+11*360/13}:6.5cm) {$v_{12}$};	
			\node[circle, draw, fill=white, minimum size=12pt,inner sep=0.6pt] (v13) at
			({90+12*360/13}:6.5cm) {$v_{13}$};		
			\draw[red,thick] (v1) -- (v2);	\draw[red,thick]  (v1) -- (v3);
			\draw[red,thick]  (v1) -- (v4);
			\draw (v1) -- (v5);	\draw[red,thick]  (v2) -- (v3);	\draw[red,thick]  (v2) --
			(v4);
			\draw (v2) -- (v5);	\draw (v3) -- (v6);	\draw (v3) -- (v7);	\draw (v4) --
			(v6);
			\draw[red,thick]  (v4) -- (v8);	\draw (v5) -- (v6);	\draw (v5) -- (v9);	\draw
			(v6) -- (v7);
			\draw (v7) -- (v8);	\draw (v7) -- (v10);	\draw (v8) -- (v11);	\draw (v8) --
			(v12);
			\draw (v9) -- (v11);	\draw (v9) -- (v12);	\draw (v9) -- (v13);
			\draw (v10) -- (v11);	\draw (v10) -- (v12);	\draw (v10) -- (v13);	
			\draw (v11) -- (v13);	\draw (v12) -- (v13);
		\end{scope}
		\begin{scope}[xshift=19.5cm,yshift=5cm,scale=1.2]
			\draw[red, line width=2.5pt, line cap=round, line join=round]
			(-0.5,0.95) -- (0.3,0) ;
			\draw[red, line width=2.5pt, line cap=round, line join=round]	(-0.5,0) --
			(0.3,0.95);
		\end{scope}
	\end{tikzpicture}
	\caption{(left) Kite graph $K_5$ --- a $5$-vertex graph without a FOR(3) [Lemma \ref{grafy-bez-for} (c)]. (right) $K_3$ highlighted in $M_{503}$.}
	\label{kite}
\end{figure}
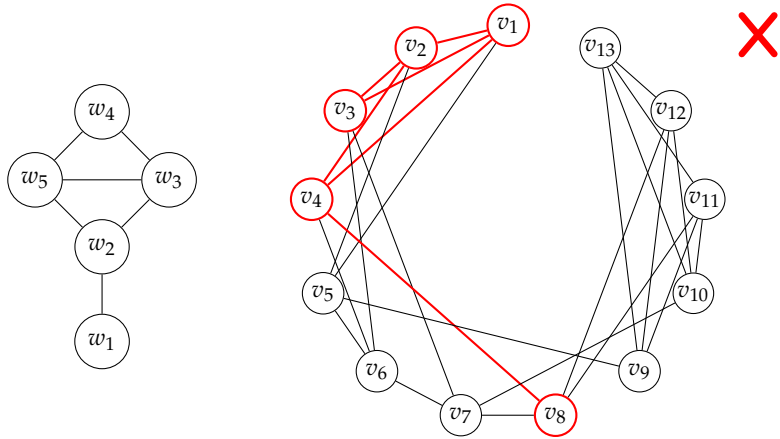

\begin{lem} \label{grafy-bez-for}
	Graphs: (a) $C_4$, (b) $H_5$, (c) $K_5$, (d) $A_6$ do not admit a FOR(3). 
\end{lem}	
\begin{proof}(a) Obvious. There is no orthogonal representation in $d\le 3$ for
	this graph. \\ 
	(b) By Lemma \ref{4cykl}, it must be that $\ket{w_2}=\ket{w_4}$ or $\ket{w_3}=\ket{w_5}$. However, this is
	impossible to be achieved for a graph with $\mathcal{N}(w_2) \ne
	\mathcal{N}(w_4)$ and   $\mathcal{N}(w_3) \ne \mathcal{N}(w_5)$ along with the
	condition that $\ket{w_2}\ne \ket{w_5}$ (there is edge $\{w_2,w_5\}$), which is
	the case here.\\
	(c)	It immediately follows that it must hold $\ket{w_2} = \ket{w_4}$ in any FOR(3).  However, $\mathcal{N}(w_2) \ne  \mathcal{N}
	(w_4)$, a contradiction. \\
	(d) By Lemma \ref{4cykl} it must hold $\ket{w_3}=\ket{w_5}$  or $\ket{w_2}=\ket{w_4}$.  However, it holds
	$\mathcal{N}(w_2) \ne \mathcal{N}(w_4)$ and $\mathcal{N}(w_3) \ne
	\mathcal{N}(w_5)$, preventing the graph to have a FOR(3).
\end{proof}

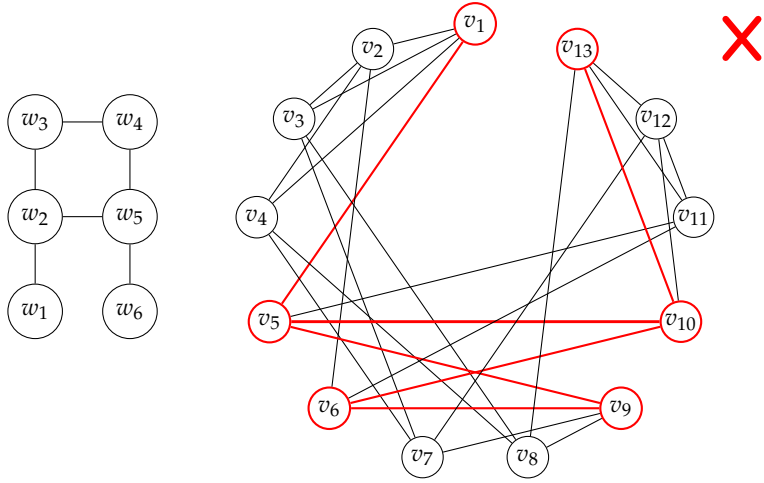
\begin{figure}[htp]
	\begin{tikzpicture}[scale=0.45]
		\centering
		\begin{scope}[yshift=-2cm,scale=0.7]
			\node[circle, draw, fill=white, minimum size=12pt] (w1) at (0,0) {$w_1$};
			\node[circle, draw, fill=white, minimum size=12pt] (w2) at (0,4) {$w_2$};
			\node[circle, draw, fill=white, minimum size=12pt] (w3) at (0,8) {$w_3$};
			\node[circle, draw, fill=white, minimum size=12pt] (w4) at (4,8) {$w_4$};
			\node[circle, draw, fill=white, minimum size=12pt] (w5) at (4,4) {$w_5$};
			\node[circle, draw, fill=white, minimum size=12pt] (w6) at (4,0) {$w_6$};
			%
			\draw (w1) -- (w2) -- (w3) -- (w4) -- (w5) -- (w6);
			\draw (w2) -- (w5);
		\end{scope}
		\begin{scope}[xshift=13cm]
			\GraphInit[vstyle=Classic]
			\SetGraphUnit{2.5}
			\node[circle, draw=red,thick, fill=white, minimum size=10pt, inner sep=1.8pt] (v1) at
			({90+0*360/13}:6.5cm) {$v_1$};
			\node[circle, draw, fill=white, minimum size=12pt, inner sep=1.8pt] (v2) at
			({90+1*360/13}:6.5cm) {$v_2$};
			\node[circle, draw, fill=white, minimum size=12pt, inner sep=1.8pt] (v3) at
			({90+2*360/13}:6.5cm) {$v_3$};
			\node[circle, draw, fill=white, minimum size=12pt, inner sep=1.8pt] (v4) at
			({90+3*360/13}:6.5cm) {$v_4$};
			\node[circle, draw=red,thick, fill=white, minimum size=12pt, inner sep=1.8pt] (v5) at
			({90+4*360/13}:6.5cm) {$v_5$};
			\node[circle, draw=red,thick, fill=white, minimum size=12pt, inner sep=1.8pt] (v6) at
			({90+5*360/13}:6.5cm) {$v_6$};
			\node[circle, draw, fill=white, minimum size=12pt, inner sep=1.8pt] (v7) at
			({90+6*360/13}:6.5cm) {$v_7$};
			\node[circle, draw, fill=white, minimum size=12pt, inner sep=1.8pt] (v8) at
			({90+7*360/13}:6.5cm) {$v_8$};	
			\node[circle, draw=red,thick, fill=white, minimum size=12pt, inner sep=1.8pt] (v9) at
			({90+8*360/13}:6.5cm) {$v_9$};	
			\node[circle, draw=red,thick, fill=white, minimum size=12pt,inner sep=0.6pt] (v10) at
			({90+9*360/13}:6.5cm) {$v_{10}$};	
			\node[circle, draw, fill=white, minimum size=12pt,inner sep=0.6pt] (v11) at
			({90+10*360/13}:6.5cm) {$v_{11}$};	
			\node[circle, draw, fill=white, minimum size=12pt,inner sep=0.6pt] (v12) at
			({90+11*360/13}:6.5cm) {$v_{12}$};	
			\node[circle, draw=red,thick, fill=white, minimum size=12pt,inner sep=0.6pt] (v13) at
			({90+12*360/13}:6.5cm) {$v_{13}$};		
			\draw (v1) -- (v2);	\draw (v1) -- (v3);	\draw (v1) -- (v4);	\draw[thick,red]
			(v1) -- (v5);
			\draw (v2) -- (v3);	\draw (v2) -- (v4);
			\draw (v2) -- (v6);	\draw (v3) -- (v7);	\draw (v3) -- (v8);
			\draw (v4) -- (v7);	\draw (v4) -- (v8);	\draw[thick,red] (v5) -- (v9);
			\draw[red, line width=1.1pt] (v5) -- (v10);\draw (v5) -- (v11);
			\draw[thick,red] (v6) -- (v9);	\draw[thick,red] (v6) -- (v10);
			\draw (v6) -- (v11);\draw (v7) -- (v9);	\draw (v7) -- (v12);\draw (v8) --
			(v9);	\draw (v8) -- (v13);
			\draw (v10) -- (v12);\draw[thick,red] (v10) -- (v13);	\draw (v11) -- (v12);
			\draw (v11) -- (v13);\draw (v12) -- (v13);
		\end{scope}
		\begin{scope}[xshift=21cm,yshift=5.5cm,scale=1.2]
			\draw[red, line width=2.5pt, line cap=round, line join=round]
			(-0.5,0.95) -- (0.3,0) ;
			\draw[red, line width=2.5pt, line cap=round, line join=round]	(-0.5,0) --
			(0.3,0.95);
		\end{scope}
	\end{tikzpicture}
	\caption{(left) A-graph $A_6$ --- a $6$-vertex graph without a FOR(3) [Lemma \ref{grafy-bez-for} (d)]. (right) $A_6$
		highlighted in graph $M_{5059}$. } \label{A-graph}
\end{figure}

We just note that it is a trivial task to find an OR in $d=3$ for graphs $A_6,
K_5,H_5$.

Notably, $C_4$ is isomorphic to the wheel graph $W_4$. Wheel graph $W_n$ is given by an $(n-1)$-cycle with an extra vertex, called a hub, adjacent to all other vertices.

\subsubsection{Forbidden induced subgraph characterization of graphs with FOR(3)}

We define the following set
\begin{align}\label{obstrukcja}
	\mathcal{O}_3=\{C_4,H_5,K_5,A_6\}.
\end{align}

We have verified through a direct automated search (see \cite{pliki}) that
$\mathcal{O}_3$ is an obstruction set for $13$-vertex $4$-regular graphs
admitting a FOR(3). Interestingly, the latter family consists of a single graph
-- $M_{5057}$. In the next subsection we present a FOR(3) for this graph and
argue that the graph cannot be a LOG for a minimal GUPB.

The results of the elimination procedure are shown in Tab.
\ref{eliminacja-przez-obstruction} and in more detail in Tab. \ref{tabelka-summary} in Appendix \ref{details-13-4}.

\begin{widetext}
	\begin{center}
		\begin{table}[h]
			\begin{tabular}{ccccc}
				&	\makecell{\boxed{\textit{$13$-vertex $4$-regular connected graphs}}\\{}}&& \\
				\hline\hline 
				\makecell{forbidden \\ induced  \\ subgraph }
				&\makecell{\# graphs w/ \\ induced subgraph } & \makecell{cumulative \\ \#
					eliminated }
				& \#  left\\ \hline \\
				A-graph $A_6$ & $10\:672$  &  $10\:672$  & $106$ \\\\
				Kite graph $K_5$ & $8\:919$  &  $10\:767$  & $11$  \\
				\\
				House graph $H_5$ & $10\:662$  &  $10\:776$  & $2$ \\\\
				$4$-clique $C_4$ & $671$  &  $10\:777$  & $1$     \\ \\
				\hline\hline
			\end{tabular}
			\caption{Summary of the forbidden induced subgraph characterization of $13$-vertex $4$-regular graphs
				with a FOR(3). Obstruction set $\mathcal{O}_3$ [Eq. \eqref{obstrukcja}] defines a single
				graph.}\label{eliminacja-przez-obstruction}
		\end{table}
	\end{center}
\end{widetext}

We have also verified that no three-element subset of $\mathcal{O}_3$ is enough
to characterize the set of graphs with a FOR(3). 
It is possible that
$\mathcal{O}_3$ is in fact minimal with regard to the considered task.

\subsubsection{FOR(3) for $M_{5 057}$}\label{graff-5057}

Graph  $M_{5 057}$ (see Fig. \ref{graf-5057}) has a FOR(3). An exemplary one is as follows: 
%
\begin{align} \label{orto-repra-5057}
	& \ket{v_1}=    \ket{0}, \quad \ket{v_2}=  \ket{1},\non 
	& \ket{v_3}=  \ket{v_4}=
	\ket{2} ,        \nonumber \\
	& \ket{v_5}=   2\ket{1}-3\ket{2},   \quad \ket{v_6}=  2\ket{0}-\ket{2},       
	\nonumber \\
	& \ket{v_7}=     \ket{0}-\ket{1},  \quad \ket{v_8}=     \ket{0}+\ket{1},    \\
	& \ket{v_9}= \ket{v_{10}}=\ket{v_{11}}= \ket{0}+3\ket{1}+2\ket{2},         
	\nonumber \\
	& \ket{v_{12}}= \ket{0}+\ket{1}-2\ket{2},  \quad \ket{v_{13}}= 
	\ket{0}-\ket{1}+\ket{2}        .\nonumber
\end{align}
One immediately notices that  this representation is not permissible according
to Fact \ref{general-fakt} because there is a $5$-tuple of vectors whose span
is only $2$-dimensional, precisely:
\begin{align}\label{za-malo}
	\dim \mathrm{span} \{\ket{v_3},\ket{v_4},\ket{v_9},\ket{v_{10}},\ket{v_{11}}\}=2.
\end{align}

The FOR we gave above is not unique, but the presence of repeating vectors
leading to the above property  is common to all representations of the graph.

\begin{lem}\label{repra-5057}
	In any  FOR(3) of $M_{5057}$ it holds:
	%
	(a) $	\ket{v_3}=  \ket{v_4}$ and (b) $	\ket{v_9}=  \ket{v_{10}}=\ket{v_{11}}$.
\end{lem}
\begin{proof}
	Case (a) is obvious as vertices  $v_3$ and $v_4$ belong to the diamond graph with vertices $(v_1,v_2,v_3,v_4)$  (highlighted in orange in Fig.
	\ref{graf-5057}). 
	
	For case (b), we need to consider (bipartite complete) subgraph with vertices
	$(v_5,v_6,v_{9},v_{10},v_{11})$ and edges $\{v_i,v_j\}$ with $i=5,6$,
	$j=9,10,11$ (highlighted in bold in Fig. \ref{graf-5057}). We will now exploit
	Lemma \ref{4cykl} again.
	Consider the square graph  with vertices $(v_5, v_{10},v_6,v_{11})$ (blue in Fig.
	\ref{graf-5057}). It must be that  (i) $\ket{v_5}=\ket{v_6}$ or (ii)
	$\ket{v_{10}}=\ket{v_{11}}$. Since $\mathcal{N}(v_5)\ne \mathcal{N}(v_6)$ in
	graph $M_{5057}$, case (i) can be discarded and case (ii) holds.
	Considering the square graph with $(v_5,v_9,v_6,v_{11})$ (thick, dashed) 
	immediately leads then to $\ket{v_{9}}=\ket{v_{11}}$. Case (b) holds. As a note,
	observe that we could equivalently take $v_{12}$ and $v_{13}$ instead of $v_5$
	and $v_6$ to arrive at the same conclusion.
	%
	%
\end{proof}
\begin{figure}[htp]
	\begin{tikzpicture}[scale=0.5]
		\centering
		\begin{scope}
			\GraphInit[vstyle=Classic]
			\SetGraphUnit{2.5}
			\node[circle, draw=orange,thick, fill=white, minimum size=10pt, inner
			sep=1.8pt] (v1) at ({90+0*360/13}:6.5cm) {$v_1$};
			\node[circle, draw=orange,thick, fill=white, minimum size=12pt, inner
			sep=1.8pt] (v2) at ({90+1*360/13}:6.5cm) {$v_2$};
			\node[circle, draw=orange,thick, fill=brown, minimum size=12pt, inner
			sep=1.8pt] (v3) at ({90+2*360/13}:6.5cm) {$v_3$};
			\node[circle, draw=orange,thick, fill=brown, minimum size=12pt, inner
			sep=1.8pt] (v4) at ({90+3*360/13}:6.5cm) {$v_4$};
			\node[circle, dashed, draw=blue,thick, fill=white, minimum size=12pt, inner
			sep=1.8pt] (v5) at ({90+4*360/13}:6.5cm) {$v_5$};
			\node[circle, dashed, draw=blue,thick, fill=white, minimum size=12pt, inner
			sep=1.8pt] (v6) at ({90+5*360/13}:6.5cm) {$v_6$};
			\node[circle, draw, fill=white, minimum size=12pt, inner sep=1.8pt] (v7) at
			({90+6*360/13}:6.5cm) {$v_7$};
			\node[circle, draw, fill=white, minimum size=12pt, inner sep=1.8pt] (v8) at
			({90+7*360/13}:6.5cm) {$v_8$};	
			\node[circle, dashed,draw=black,thick, fill=lime, minimum size=12pt, inner
			sep=1.8pt] (v9) at ({90+8*360/13}:6.5cm) {$v_9$};	
			\node[circle, draw=blue,thick, fill=lime, minimum size=12pt,inner sep=0.6pt]
			(v10) at ({90+9*360/13}:6.5cm) {$v_{10}$};	
			\node[circle, dashed, draw=blue,thick, fill=lime, minimum size=12pt,inner
			sep=0.6pt] (v11) at ({90+10*360/13}:6.5cm) {$v_{11}$};	
			\node[circle, draw, fill=white, minimum size=12pt,inner sep=0.6pt] (v12) at
			({90+11*360/13}:6.5cm) {$v_{12}$};	
			\node[circle, draw, fill=white, minimum size=12pt,inner sep=0.6pt] (v13) at
			({90+12*360/13}:6.5cm) {$v_{13}$};		
			\draw[orange,thick] (v1) -- (v2);		\draw[orange,thick] (v1) -- (v3);
			\draw[orange,thick] (v1) -- (v4);		\draw (v1) -- (v5);
			\draw[orange,thick] (v2) -- (v3);		\draw[orange,thick] (v2) -- (v4);
			\draw (v2) -- (v6);		\draw (v3) -- (v7);
			\draw (v3) -- (v8);		\draw (v4) -- (v7);
			\draw (v4) -- (v8);		\draw[line width=1.4pt,dashed] (v5) -- (v9);
			\draw[blue,line width =1.4pt,dashed] (v5) -- (v11); 	\draw[line width=1.4pt,
			blue](v5) -- (v10);
			\draw[line width=1.7pt,dashed] (v6) -- (v9);		\draw[blue,thick] (v6) -- (v10);
			\draw[blue,line width=1.4pt,dashed] (v6) -- (v11);		\draw (v7) -- (v8);
			\draw (v7) -- (v12);		\draw (v8) -- (v13);
			\draw (v9) -- (v12);		\draw (v9) -- (v13);
			\draw (v10) -- (v12);		\draw (v10) -- (v13);
			\draw (v11) -- (v12);		\draw (v11) -- (v13);
		\end{scope}
		\begin{scope}[xshift=8.5cm,yshift=5.5cm,scale=1.3]
			\draw[ufogreen, line width=2.5pt, line cap=round, line join=round]
			(-0.3,0.4) -- (0,0) -- (0.6,0.95);
		\end{scope}
	\end{tikzpicture}
	\caption{Graph $M_{5 057}$ with the relevant subgraphs highlighted (see the
		proof of Lemma \ref{repra-5057}). There exist FORs(3) for this graph,
		e.g., the one given in Eq. (\ref{orto-repra-5057}); nevertheless, none of them
		complies with the requirement given in Fact \ref{general-fakt} as the span of
		vectors corresponding to the colored vertices is always only $2$-dimensional.}
	\label{graf-5057}
\end{figure}
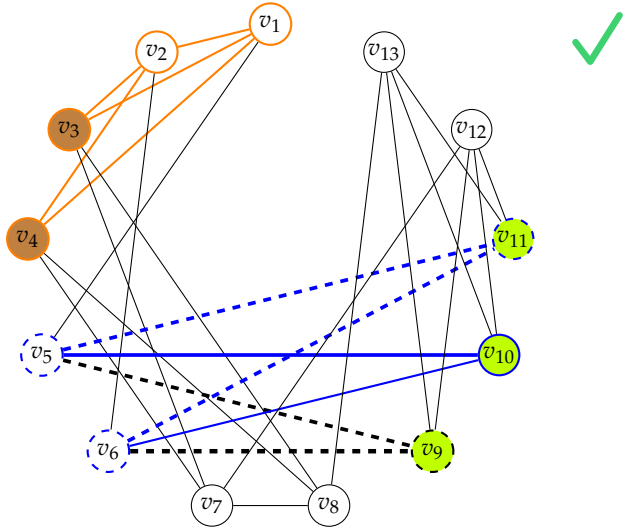

In consequence, candidate graph $M_{5057}$, while admitting FORs(3), gets discarded by  Fact
\ref{general-fakt}.

\subsection{Thirteen element three-qutrit GUPB -- conclusion} \label{konkluzja}

We have found  that only two $4$-regular graphs with $13$ vertices have a
faithful orthogonal representation in dimension three --- $D_6 + D_{7,a}$
(disconnected) and $M_{5057}$ (connected).
However, none of them can be a LOG for a GUPB by Fact \ref{general-fakt}. This
leads to the main result of the present paper.
\newline
\newline
\noindent \textbf{Main result.}	\textit{There does not exist a three-qutrit GUPB
	with a cardinality of $13$. }

\section{General GUPBs } \label{ogolniej}

It is evident that the proposed method of forbidden induced subgraph
characterization of graphs with a FOR is not fundamentally limited to the
particular case of three-partite qutrit GUPBs of the so-far considered size. In principle, it
can be applied universally for arbitrary $N$ and $d$ (or even different local
dimensions for each subsystem; for clarity, we will omit these cases in further
discussion) and cardinalities of the sets of vectors (not necessarily the
minimal ones). The actual implementation in such cases, however, soon becomes
intractable as the complexity of the problem grows very quickly (cf. Sec. \ref{ququarty}).  
In this section, we provide some initial results regarding the cases outside the one considered in Sec. \ref{main-section}.

\subsection{Qutrits}

We can use  set $\mathcal{O}_3$ to discard a number  of candidate graphs (with $\girth=3,4$) in any case with
$d=3$.  
However, this set is not a universal obstruction set. 
For example, for $14$-vertex $4$-regular graphs (Sec. \ref{14-vert-4-reg} and Appendix \ref{eny}) we have identified  an additional structure which should be included; it is relatively large as its vertex set has as many as eleven elements. On the other hand, we found that for some classes of $3$-regular graphs (Sec. \ref{remarksy} and Appendices \ref{3-reg-connected}-\ref{3-reg-disconnected}) the set can be reduced, which may also involve its modification.

\subsubsection{$4$-regular graphs with $14$ vertices}\label{14-vert-4-reg}

We have performed an analysis analogous to the presented above in case of $14$-vertex $4$-regular graphs. 
One finds that there are $25$ disconnected \cite{disconnected-quartic} and $88\:186$ connected \cite{connected-quartic} such graphs. 

In case of disconnected graphs we now have three types: (a) type I ($16$ graphs): $5+9$ vertices, (b) type II ($6$ graphs): $6+8$ vertices, and (c) type III ($3$ graphs): $7+7$ vertices. 

We can use directly some of the results from Sec. \ref{rozlaczne} . First, we can immediately discard type I graphs. Further, we find that there is a single type III graph with a FOR(3): $D_{7,a}+D_{7,a}$ (cf. Fig. \ref{7-vertices}). Finally, considering type II graphs, we recall that there is a unique $6$-vertex $4$-regular graph, $D_6$, which does admit a FOR(3), so we need to check the $6$ graphs $G_{8,i}$ with $8$ vertices (see Fig. \ref{8-vertices}). We find that graphs $G_{8,1}$ through $G_{8,5}$ have an induced subgraph isomorphic to $H_5$ and as such do not admit FORs(3). It is trivial to realize that $G_{8,6}$ (which is a complete bipartite graph $K_{4,4}$) actually admits a FOR(2) as we can set $\repr{2}=\repr{3}=\repr{4}=\repr{5}=\ket{0}$ and $\repr{1}=\repr{6}=\repr{7}=\repr{8}=\ket{1}$. Clearly, this is not a unique (up to a unitary) FOR(3), however, it is easy to show that in any FOR(3) there will be four identical vectors (in Fig. \ref{8-vertices}, these will be vectors corresponding to green or yellow vertices). Consider vertex $v_1$; from $\mathcal{N}(v_1)=(v_2,v_3,v_4,v_5)$ we infer that $\mathfrak{d}:=\dim\mathrm{span} \{\repr{2},\repr{3},\repr{4},\repr{5}\}  \le 2$. If $\mathfrak{d}=2$, then it obviously holds $\dim\mathrm{span} \{\repr{1},\repr{6},\repr{7},\repr{8}\} =1$ and the claim follows; if $\mathfrak{d}=1$, the result trivially holds.

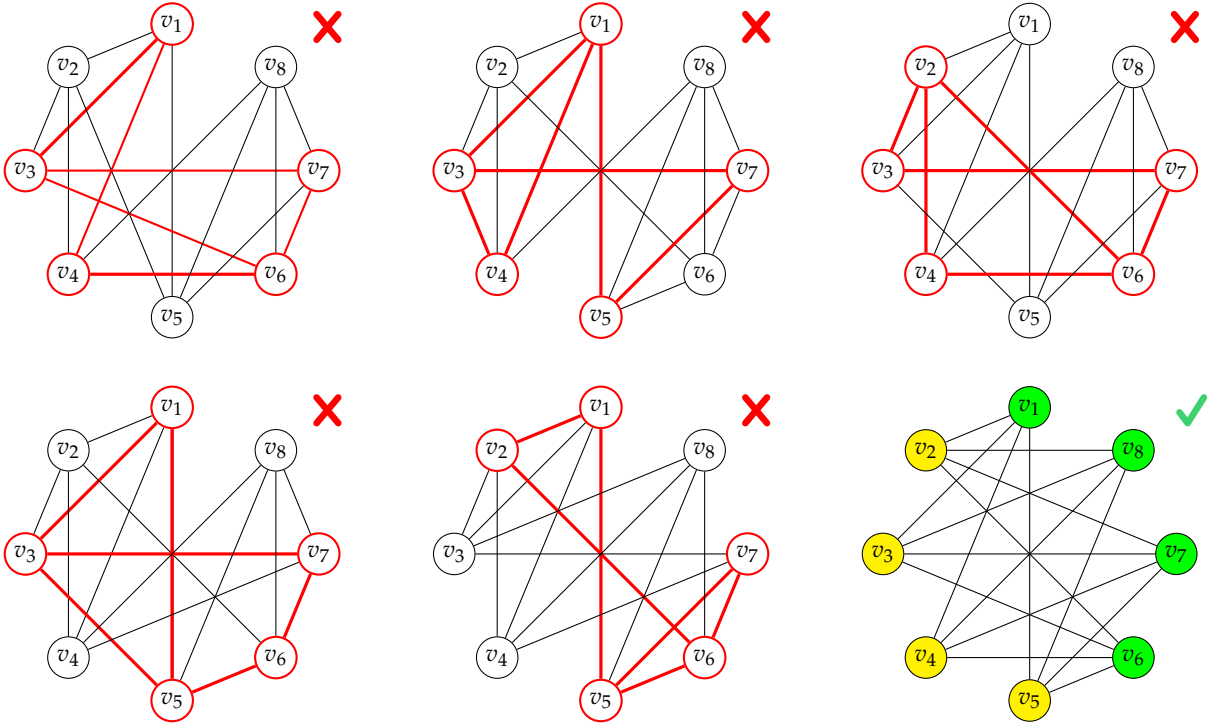
\begin{figure}[htp]
	\begin{tikzpicture}[scale=0.3]
		\centering
		\begin{scope}[xshift=7cm,yshift=5.8cm,scale=1.2]
			\draw[red, line width=2.5pt, line cap=round, line join=round]
			(-0.5,0.95) -- (0.3,0) ;
			\draw[red, line width=2.5pt, line cap=round, line join=round]	(-0.5,0) --
			(0.3,0.95);
		\end{scope}
		\begin{scope}[xshift=26cm,yshift=5.8cm,scale=1.2]
			\draw[red, line width=2.5pt, line cap=round, line join=round]
			(-0.5,0.95) -- (0.3,0) ;
			\draw[red, line width=2.5pt, line cap=round, line join=round]	(-0.5,0) --
			(0.3,0.95);
		\end{scope}
		\begin{scope}[xshift=45cm,yshift=5.8cm,scale=1.2]
			\draw[red, line width=2.5pt, line cap=round, line join=round]
			(-0.5,0.95) -- (0.3,0) ;
			\draw[red, line width=2.5pt, line cap=round, line join=round]	(-0.5,0) --
			(0.3,0.95);
		\end{scope}
		\begin{scope}[xshift=7cm,yshift=-11.2cm,scale=1.2]
			\draw[red, line width=2.5pt, line cap=round, line join=round]
			(-0.5,0.95) -- (0.3,0) ;
			\draw[red, line width=2.5pt, line cap=round, line join=round]	(-0.5,0) --
			(0.3,0.95);
		\end{scope}
		\begin{scope}[xshift=26cm,yshift=-11.2cm,scale=1.2]
			\draw[red, line width=2.5pt, line cap=round, line join=round]
			(-0.5,0.95) -- (0.3,0) ;
			\draw[red, line width=2.5pt, line cap=round, line join=round]	(-0.5,0) --
			(0.3,0.95);
		\end{scope}
		\begin{scope}[xshift=45cm,yshift=-11.2cm,scale=1.2]
			\draw[ufogreen, line width=2.5pt, line cap=round, line join=round]
		(-0.3,0.4) -- (0,0) -- (0.6,0.95);
		\end{scope}
		\begin{scope}
			\GraphInit[vstyle=Classic]
			\SetGraphUnit{2.5}
			\node[circle, draw=red,thick, fill=white, minimum size=10pt, inner sep=1.8pt] (v1) at ({90+0*360/8}:6.5cm) {$v_1$};
			\node[circle, draw, fill=white, minimum size=12pt, inner sep=1.8pt] (v2) at ({90+1*360/8}:6.5cm) {$v_2$};
			\node[circle, draw=red,thick, fill=white, minimum size=12pt, inner sep=1.8pt] (v3) at ({90+2*360/8}:6.5cm) {$v_3$};
			\node[circle, draw=red,thick, fill=white, minimum size=12pt, inner sep=1.8pt] (v4) at ({90+3*360/8}:6.5cm) {$v_4$};
			\node[circle, draw, fill=white, minimum size=12pt, inner sep=1.8pt] (v5) at ({90+4*360/8}:6.5cm) {$v_5$};
			\node[circle, draw=red,thick, fill=white, minimum size=12pt, inner sep=1.8pt] (v6) at ({90+5*360/8}:6.5cm) {$v_6$};
			\node[circle, draw=red,thick, fill=white, minimum size=12pt, inner sep=1.8pt] (v7) at ({90+6*360/8}:6.5cm) {$v_7$};
			\node[circle, draw, fill=white, minimum size=12pt, inner sep=1.8pt] (v8) at ({90+7*360/8}:6.5cm) {$v_8$};			
		\draw(v1) -- (v2);		\draw[red, line width=1.2pt](v1) -- (v3);		\draw[red, thick] (v1) -- (v4);		\draw (v1) -- (v5);
		\draw (v2) -- (v3);		\draw (v2) -- (v4);		\draw (v2) -- (v5);
		\draw[red, thick] (v3) -- (v6);		\draw[red, thick] (v3) -- (v7);
		\draw[red, line width=1.2pt] (v4) -- (v6);		\draw (v4) -- (v8);		\draw (v5) -- (v7);		\draw (v5) -- (v8);
		\draw[red, thick] (v6) -- (v7);		\draw (v6) -- (v8);
		\draw (v7) -- (v8);
		\end{scope}
		\begin{scope}[xshift=19cm]
			\GraphInit[vstyle=Classic]
			\SetGraphUnit{2.5}
			\node[circle, draw=red,thick, fill=white, minimum size=10pt, inner sep=1.8pt] (v1) at ({90+0*360/8}:6.5cm) {$v_1$};
			\node[circle, draw, fill=white, minimum size=12pt, inner sep=1.8pt] (v2) at ({90+1*360/8}:6.5cm) {$v_2$};
			\node[circle, draw=red,thick, fill=white, minimum size=12pt, inner sep=1.8pt] (v3) at ({90+2*360/8}:6.5cm) {$v_3$};
			\node[circle, draw=red,thick, fill=white, minimum size=12pt, inner sep=1.8pt] (v4) at ({90+3*360/8}:6.5cm) {$v_4$};
			\node[circle, draw=red,thick, fill=white, minimum size=12pt, inner sep=1.8pt] (v5) at ({90+4*360/8}:6.5cm) {$v_5$};
			\node[circle, draw, fill=white, minimum size=12pt, inner sep=1.8pt] (v6) at ({90+5*360/8}:6.5cm) {$v_6$};
			\node[circle, draw=red,thick, fill=white, minimum size=12pt, inner sep=1.8pt] (v7) at ({90+6*360/8}:6.5cm) {$v_7$};
			\node[circle, draw, fill=white, minimum size=12pt, inner sep=1.8pt] (v8) at ({90+7*360/8}:6.5cm) {$v_8$};			
		\draw (v1) -- (v2);		\draw[red, line width=1.2pt] (v1) -- (v3);		\draw[red, line width=1.2pt] (v1) -- (v4);		\draw[red, line width=1.2pt] (v1) -- (v5);
		\draw (v2) -- (v3);		\draw (v2) -- (v4);		\draw (v2) -- (v6);
		\draw[red, line width=1.2pt] (v3) -- (v7); \draw[red, line width=1.2pt](v3) -- (v4);
		\draw (v4) -- (v8);
		\draw (v5) -- (v6);		\draw[red, line width=1.2pt] (v5) -- (v7);		\draw (v5) -- (v8);
		\draw (v6) -- (v7);		\draw (v6) -- (v8);
		\draw (v7) -- (v8);

		\end{scope}
		\begin{scope}[xshift=38cm]
			\GraphInit[vstyle=Classic]
			\SetGraphUnit{2.5}
			\node[circle, draw, fill=white, minimum size=10pt, inner sep=1.8pt] (v1) at ({90+0*360/8}:6.5cm) {$v_1$};
			\node[circle, draw=red,thick, fill=white, minimum size=12pt, inner sep=1.8pt] (v2) at ({90+1*360/8}:6.5cm) {$v_2$};
			\node[circle, draw=red,thick, fill=white, minimum size=12pt, inner sep=1.8pt] (v3) at ({90+2*360/8}:6.5cm) {$v_3$};
			\node[circle, draw=red,thick, fill=white, minimum size=12pt, inner sep=1.8pt] (v4) at ({90+3*360/8}:6.5cm) {$v_4$};
			\node[circle, draw, fill=white, minimum size=12pt, inner sep=1.8pt] (v5) at ({90+4*360/8}:6.5cm) {$v_5$};
			\node[circle, draw=red,thick, fill=white, minimum size=12pt, inner sep=1.8pt] (v6) at ({90+5*360/8}:6.5cm) {$v_6$};
			\node[circle, draw=red,thick, fill=white, minimum size=12pt, inner sep=1.8pt] (v7) at ({90+6*360/8}:6.5cm) {$v_7$};
			\node[circle, draw, fill=white, minimum size=12pt, inner sep=1.8pt] (v8) at ({90+7*360/8}:6.5cm) {$v_8$};			
			\draw(v1)--(v2);			\draw (v1)--(v3);			\draw (v1)--(v4);			\draw (v1)--(v5);
			\draw[red, line width=1.2pt] (v2)--(v3);			\draw[red, line width=1.2pt] (v2)--(v4);			\draw[red, line width=1.2pt] (v2)--(v6);
			\draw (v3)--(v5);			\draw[red, line width=1.2pt] (v3)--(v7);
			\draw[red, line width=1.2pt] (v4)--(v6);			\draw (v4)--(v8);
			\draw (v5)--(v7);			\draw (v5)--(v8);
			\draw[red, line width=1.2pt] (v6)--(v7);			\draw (v6)--(v8);
			\draw (v7)--(v8);
		\end{scope}
			\begin{scope}[yshift=-17cm]
			\GraphInit[vstyle=Classic]
			\SetGraphUnit{2.5}
			\node[circle, draw=red,thick, fill=white, minimum size=10pt, inner sep=1.8pt] (v1) at ({90+0*360/8}:6.5cm) {$v_1$};
			\node[circle, draw, fill=white, minimum size=12pt, inner sep=1.8pt] (v2) at ({90+1*360/8}:6.5cm) {$v_2$};
			\node[circle, draw=red,thick, fill=white, minimum size=12pt, inner sep=1.8pt] (v3) at ({90+2*360/8}:6.5cm) {$v_3$};
			\node[circle, draw, fill=white, minimum size=12pt, inner sep=1.8pt] (v4) at ({90+3*360/8}:6.5cm) {$v_4$};
			\node[circle, draw=red,thick, fill=white, minimum size=12pt, inner sep=1.8pt] (v5) at ({90+4*360/8}:6.5cm) {$v_5$};
			\node[circle, draw=red,thick, fill=white, minimum size=12pt, inner sep=1.8pt] (v6) at ({90+5*360/8}:6.5cm) {$v_6$};
			\node[circle, draw=red,thick, fill=white, minimum size=12pt, inner sep=1.8pt] (v7) at ({90+6*360/8}:6.5cm) {$v_7$};
			\node[circle, draw, fill=white, minimum size=12pt, inner sep=1.8pt] (v8) at ({90+7*360/8}:6.5cm) {$v_8$};			
			\draw (v1)--(v2);			\draw[red, line width=1.2pt] (v1)--(v3);			\draw (v1)--(v4);			\draw[red, line width=1.3pt] (v1)--(v5);
			\draw (v2)--(v3);			\draw (v2)--(v4);			\draw (v2)--(v6);
			\draw[red, line width=1.2pt] (v3)--(v5);			\draw[red, line width=1.2pt] (v3)--(v7);
			\draw (v4)--(v7);			\draw (v4)--(v8);
			\draw[red, line width=1.2pt] (v5)--(v6);			\draw (v5)--(v8);
			\draw[red, line width=1.2pt] (v6)--(v7);			\draw (v6)--(v8);
			\draw (v7)--(v8);
		\end{scope}
		\begin{scope}[xshift=19cm,yshift=-17cm]
			\GraphInit[vstyle=Classic]
			\SetGraphUnit{2.5}
			\node[circle, draw=red,thick, fill=white, minimum size=10pt, inner sep=1.8pt] (v1) at ({90+0*360/8}:6.5cm) {$v_1$};
			\node[circle, draw=red,thick, fill=white, minimum size=12pt, inner sep=1.8pt] (v2) at ({90+1*360/8}:6.5cm) {$v_2$};
			\node[circle, draw, fill=white, minimum size=12pt, inner sep=1.8pt] (v3) at ({90+2*360/8}:6.5cm) {$v_3$};
			\node[circle, draw, fill=white, minimum size=12pt, inner sep=1.8pt] (v4) at ({90+3*360/8}:6.5cm) {$v_4$};
			\node[circle, draw=red,thick, fill=white, minimum size=12pt, inner sep=1.8pt] (v5) at ({90+4*360/8}:6.5cm) {$v_5$};
			\node[circle, draw=red,thick, fill=white, minimum size=12pt, inner sep=1.8pt] (v6) at ({90+5*360/8}:6.5cm) {$v_6$};
			\node[circle, draw=red,thick, fill=white, minimum size=12pt, inner sep=1.8pt] (v7) at ({90+6*360/8}:6.5cm) {$v_7$};
			\node[circle, draw, fill=white, minimum size=12pt, inner sep=1.8pt] (v8) at ({90+7*360/8}:6.5cm) {$v_8$};			
		\draw[red, line width=1.2pt] (v1)--(v2);		\draw (v1)--(v3);		\draw (v1)--(v4);		\draw[red, line width=1.2pt] (v1)--(v5);
		\draw (v2)--(v3);		\draw (v2)--(v4);		\draw[red, line width=1.2pt] (v2)--(v6);
		\draw (v3)--(v7);		\draw (v3)--(v8);
		\draw (v4)--(v7);		\draw (v4)--(v8);
		\draw[red, line width=1.2pt] (v5)--(v6);		\draw[red, line width=1.2pt] (v5)--(v7);		\draw (v5)--(v8);
		\draw[red, line width=1.2pt] (v6)--(v7);		\draw (v6)--(v8);
		\end{scope}
		\begin{scope}[xshift=38cm,yshift=-17cm]
			\GraphInit[vstyle=Classic]
			\SetGraphUnit{2.5}
			\node[circle, draw, fill=green, minimum size=10pt, inner sep=1.8pt] (v1) at ({90+0*360/8}:6.5cm) {$v_1$};
			\node[circle, draw, fill=yellow, minimum size=12pt, inner sep=1.8pt] (v2) at ({90+1*360/8}:6.5cm) {$v_2$};
			\node[circle, draw, fill=yellow, minimum size=12pt, inner sep=1.8pt] (v3) at ({90+2*360/8}:6.5cm) {$v_3$};
			\node[circle, draw, fill=yellow, minimum size=12pt, inner sep=1.8pt] (v4) at ({90+3*360/8}:6.5cm) {$v_4$};
			\node[circle, draw, fill=yellow, minimum size=12pt, inner sep=1.8pt] (v5) at ({90+4*360/8}:6.5cm) {$v_5$};
			\node[circle, draw, fill=green, minimum size=12pt, inner sep=1.8pt] (v6) at ({90+5*360/8}:6.5cm) {$v_6$};
			\node[circle, draw, fill=green, minimum size=12pt, inner sep=1.8pt] (v7) at ({90+6*360/8}:6.5cm) {$v_7$};
			\node[circle, draw, fill=green, minimum size=12pt, inner sep=1.8pt] (v8) at ({90+7*360/8}:6.5cm) {$v_8$};			
		\draw (v1)--(v2);	\draw (v1)--(v3);		\draw (v1)--(v4);		\draw (v1)--(v5);
		\draw (v2)--(v6);		\draw (v2)--(v7);		\draw (v2)--(v8);
		\draw (v3)--(v6);		\draw (v3)--(v7);		\draw (v3)--(v8);
		\draw (v4)--(v6);		\draw (v4)--(v7);		\draw (v4)--(v8);
		\draw (v5)--(v6);		\draw (v5)--(v7);		\draw (v5)--(v8);
		\end{scope}
	\end{tikzpicture}
	\caption{All six $4$-regular graphs with $8$ vertices. Top row (from the left): $G_{8,1}$,  $G_{8,2}$, $G_{8,3}$, bottom row (from the left): $G_{8,4}$,  $G_{8,5}$, $G_{8,6}$. Induced subgraphs of $G_{8,i}$, $i=1,2,3,4,5$, which are  isomorphic to $H_5$ are highlighted in red; none of these graphs can have a FOR(3). $G_{8,6}$ does not have any of the graphs from $\mathcal{O}_3$ as an induced subgraph. This graph admits a FOR(3) with at least four identical vector; see the main text.} \label{8-vertices}
\end{figure}

Let us now move to the connected graphs. We will denote these graphs by $N_i$ with their enumeration in accordance with the GENREG output \cite{meringer2}. Among the connected graphs, we find that  $220$ graphs have $\girth=4$ and the remaining ones have $\girth=3$. We can use set $\mathcal{O}_3$ from Eq. \eqref{obstrukcja} for graph elimination, which leaves us only $7$ graphs (depicted in Fig. \ref{14-vertices-appendix} in Appendix \ref{eny}):
\begin{align} \label{grafy-N}
 N_{2359}, N_{11743}, N_{36919}, N_{80015}, N_{87949},  N_{87956},
N_{87957}. 
\end{align}
We analyze these graphs in Appendix \ref{eny}, where we find that only $N_{80015}$ does not have a FOR(3). The induced subgraph $\widehat{N}_{11}$ (see Fig. \ref{redraw-graph}), which eliminates this graph is much larger than the elements of $\mathcal{O}_3$ as it has eleven vertices. We have the following obstruction set for the analyzed case (see Tab. \ref{eliminacja-przez-obstruction-2}):
\begin{align}
	\widehat{\mathcal{O}}_3= \mathcal{O}_3 \cup \{\widehat{N}_{11}\} = \{C_4,H_5,K_5,A_6,\widehat{N}_{11} \}.
\end{align}

\subsubsection{Fourteen-element three-qutrit GUPBs - remarks}\label{remarksy}

Let us now turn our attention to the case of $14$-element three-qutrit GUPBs. Exploiting Fact \ref{general-fakt}, we can discard $D_{7,a}+ D_{7,a}$ and $D_6 + D_{8,6}$ (disconnected)
and $N_{2359}$ and $N_{87949}$ (connected) as LOGs in this case. At the same time, we have not been able to find an argument which could also eliminate the remaining connected graphs and thus their status is undecided.

It must be noted that  $4$-regular graphs are not sufficient in this case, i.e., not all LOGs must be of this type. This can be seen  through a counting mismatch:
each $4$‑regular graph on $14$ vertices has $28$ edges, so three such graphs could only account for thrice this number of edges, i.e., $84$, while the complete graph $C_{14}$ has $91$ edges.
However, by the same argument, decomposition of $C_{14}$ into two $4$-regular graphs  and one $5$-regular graph (with $35$ edges) or one $3$-regular graph (with $21$ edges) and two $5$-regular graphs is possible.

We find that there are $509$ connected $3$-regular graphs \cite{connected-cubic}. Among them there are $110$ graphs with $\girth \ge 4$, $9$ with $\girth \ge 5$, including one with $\girth = 6$; the remaining ones have $\girth=3$ (there are no graphs with $\girth \ge 7$). Obstruction set $\mathcal{O}_3$ (in fact, set $\mathcal{O}_3 \setminus C_4$ as the connected $3$-regular graphs do not contain $4$-cliques) defines $57$ graphs -- $42$ graphs with $\girth=3$, six graphs with $\girth=4$, and obviously the nine graphs with  $\girth \ge 5$. We found that all the forty-eight graphs with $\girth=3$ and $\girth=4$ do have permissible FORs(3). This is also the case for the eight graphs with $\girth=5$. We have not been able to show whether the single graph with $\girth=6$ has a FOR(3). In Appendix \ref{3-reg-connected}, we gathered the relevant data.

There are also $31$ disconnected $3$-regular graphs \cite{disconnected-cubic}. These are of three types: (a) $4+10$ vertices ($1\times 19=19$ graphs), (b) $6+8$ vertices ($2\times 5=10$ graphs), (c) $4+4+6$ ($2$ graphs). Types (a) and (c) get immediately eliminated as $3$-regular graphs with four vertices are just $4$-cliques. In Appendix \ref{3-reg-disconnected}, we provide a brief analysis of the graphs that remain after this initial elimination.

There are further  $3\:459\:383$ connected (seven graphs with $\girth=4$ and the remaining ones with $\girth=3$)  and $3$ disconnected $5$-regular graphs (the disconnected ones get immediately discarded as they correspond to the division $6+8$ of vertices; there are no $5$-regular graphs with seven vertices) \cite{quintic-connected,quintic-disconnected}.   We have not analyzed the connected graphs with our approach and leave it for future study. 

In general, by the results of \cite{kiribela}, we find that the degrees of vertices in LOGs are non-trivially bounded and belong to the set $\{3,4,5\}$. 
Using \verb|nauty| \cite{nauty}, we have found that there are $80\;020\;903\;462$ non-isomorphic graphs of this kind. One also needs to bear in mind that the degrees of vertices of two-party orthogonality graphs comply with certain non-trivial conditions.
\begin{widetext}
	\begin{center}
		\begin{table}[h]
			\begin{tabular}{ccccc}
			&	\makecell{\boxed{\textit{$14$-vertex $4$-regular connected graphs}}\\{}}&& \\
				\hline\hline 
				\makecell{forbidden \\ induced  \\ subgraph }
				&\makecell{\# graphs w/ \\ induced subgraph } & \makecell{cumulative \\ \#
					eliminated }
				& \#  left\\ \hline \\
				A-graph $A_6$ & $87\:868$  &  $87\:868$  & $300$ \\\\
				Kite graph $K_5$ & $71\:322$  &  $88\:139$  & $29$  \\
				\\
				House graph $H_5$ & $87\:537$  &  $88\:158$  & $10$ \\\\
				$4$-clique $C_4$ & $4\:184$  &  $88\:161$  & $7$     \\ \\
				$\widehat{N}_{11}$ (Fig. \ref{redraw-graph})  & 33 & $88\:162$ & $6$ \\ \\ 
				\hline\hline
			\end{tabular}
			\caption{Summary of the forbidden induced subgraph characterization of $14$-vertex $4$-regular graphs
				with a FOR(3). Obstruction set $\mathcal{O}_3$ defines seven
				graphs  (first four rows of the table). One more graph gets eliminated by its forbidden induced subgraph $\widehat{N}_{11}$ leaving six graphs, which we show to have FORs(3).}\label{eliminacja-przez-obstruction-2}
		\end{table}
	\end{center}
\end{widetext}

\subsection{Ququarts}

We now provide some partial results concerning bases with ququart subsystems ($d=4$).

\subsubsection{Small graphs without a FOR(4)}

Our procedure exploits the fact that there are necessarily repeated
vectors in FORs for certain smaller subgraphs (the non-existence of FORs for
cliques can also be explained from this point of view). This allows for the
elimination of candidate graphs but also proves useful in the characterization
of FORs for the graphs possessing them. In case of qutrits, the
	elementary structures with the property are the square graph and the diamond graph (Lemma \ref{4cykl}).
In this subsection, we give some simple graphs not admitting FORs in dimension
	$d=4$, priorly identifying graphs with repeating vectors in a FOR(4). 
	
	We begin with a consideration of a wheel graph and graphs obtained from it: envelope graphs $E_6$ and $E_6^|$, and $W_7^{||}$ (see Fig. \ref{wheel}). The derived graphs are direct analogs of the house graph $H_5$ and A-graph $A_6$, respectively.
	
	\begin{lem}\label{wheel-graf}
		(a) Wheel graph $W_4$ has at least one repeated vector in a FOR(4). (b) Envelope graph $E_6$ and graph $E_6^|$ do not have FORs(4). (c) Graph $W^{||}_7$ does not admit a FOR(4).
	\end{lem}
	\begin{proof}
		(a) W.l.o.g. we set $\ket{w_1}=\zero$, $\ket{w_2}=\jeden$, and $\ket{w_3}=\dwa$. Then it must be $\ket{w_4}=\jeden+\alpha \ket{3}$ and $\ket{w_5}=\zero+\beta \ket{3}$; these graphs are orthogonal implying $\alpha \beta=0$.
		
		(b, c)  For the three graphs the argument is the same: neighborhoods of vertices with equal vectors are different, a contradiction.

		\end{proof}
	\begin{figure}[h]
		\centering
		\begin{tikzpicture}[scale=0.45]
			\begin{scope}[scale=0.7,xshift=-6cm]
			\node[circle, draw, fill=yellow, minimum size=8pt,inner sep=1.5pt] (w1) at
		(0,0) {$\ket{w_1}=\ket{0}$};
		\node[circle, draw, fill=green, minimum size=8pt,inner sep=2pt] (w2) at
		(0,10) {$\ket{w_2}=\ket{1}$};
		\node[circle, draw, fill=green, minimum size=8pt,inner sep=-2pt] (w5) at
		(10,0) {$\begin{array}{c}
				\ket{w_4}\!\!=\!\!\ket{1}\!\!+\!\alpha \ket{3} \\ (\alpha \beta\! =\! 0) \\
			\end{array}$};
		\node[circle, draw, fill=yellow, minimum size=8pt,inner sep=-2pt] (w4) at
		(10,10) {$\begin{array}{c}
				\ket{w_5}\!=\!\!\ket{0}\!\!+\!\!\beta \ket{3} \\ (\alpha\beta\! =\! 0) 
			\end{array}$};
			\node[circle, draw, fill=white, minimum size=8pt,inner sep=2pt] (w3) at (5,5) {$\ket{w_3}=\ket{2}$};
		\draw (w1)-- (w2) -- (w4) -- (w5) -- (w1); \draw (w1)--(w3)--(w4); \draw (w2)--(w3)--(w5);	%
			\end{scope}
			\begin{scope}[xshift=10cm,yshift=0cm,scale=0.9]
				\node[circle, draw, fill=white, minimum size=12pt,inner sep=2pt] (w1) at
				(0,0) {$w_1$};
				\node[circle, draw, fill=white, minimum size=12pt,inner sep=2pt] (w2) at
				(0,5) {$w_2$};
				\node[circle, draw, fill=white, minimum size=12pt,inner sep=2pt] (w4) at
				(5,0) {$w_4$};
				\node[circle, draw, fill=white, minimum size=12pt,inner sep=2pt] (w5) at
				(5,5) {$w_5$};
				\node[circle, draw, fill=white, minimum size=12pt,inner sep=2pt] (w3) at (2.5,2.5) {$w_3$};
				\node[circle, draw, fill=white, minimum size=12pt,inner sep=2pt] (w6) at (2.5,8) {$w_6$};
				%
				\draw (w1)-- (w2) -- (w5) -- (w4) -- (w1); \draw (w1)--(w3)--(w4); \draw (w2)--(w3)--(w5);	%
				\draw (w2)--(w6); \draw (w5)--(w6); \draw[dashed] (w3)--(w6);
			\end{scope}
			%
			\begin{scope}[xshift=15.5cm,yshift=7cm,scale=1.2]
				\draw[red, line width=2.5pt, line cap=round, line join=round]
				(-0.5,0.95) -- (0.3,0) ;
				\draw[red, line width=2.5pt, line cap=round, line join=round]	(-0.5,0) --
				(0.3,0.95);
			\end{scope}
		\begin{scope}[xshift=20cm,yshift=3cm,scale=0.9]
		\node[circle, draw, fill=white, minimum size=12pt,inner sep=2pt] (w1) at
	(0,0) {$w_1$};
	\node[circle, draw, fill=white, minimum size=12pt,inner sep=2pt] (w2) at
	(0,5) {$w_2$};
	\node[circle, draw, fill=white, minimum size=12pt,inner sep=2pt] (w4) at
	(5,0) {$w_4$};
	\node[circle, draw, fill=white, minimum size=12pt,inner sep=2pt] (w5) at
	(5,5) {$w_5$};
	\node[circle, draw, fill=white, minimum size=12pt,inner sep=2pt] (w3) at (2.5,2.5) {$w_3$};
		\node[circle, draw, fill=white, minimum size=12pt,inner sep=2pt] (w6) at (0,-3) {$w_6$};
			\node[circle, draw, fill=white, minimum size=12pt,inner sep=2pt] (w7) at (5,-3) {$w_7$};
	\draw (w1)-- (w2) -- (w5) -- (w4) -- (w1); \draw (w1)--(w3)--(w4); \draw (w2)--(w3)--(w5);	%
	\draw (w1)--(w6); \draw (w4)--(w7);
		\end{scope}
			%
			\begin{scope}[xshift=26.5cm,yshift=7cm,scale=1.2]
				\draw[red, line width=2.5pt, line cap=round, line join=round]
				(-0.5,0.95) -- (0.3,0) ;
				\draw[red, line width=2.5pt, line cap=round, line join=round]	(-0.5,0) --
				(0.3,0.95);
			\end{scope}
		\end{tikzpicture}
		\caption{(left) Wheel graph $W_5$ --- an exemplary small graph with a necessarily repeated vector in a FOR(4) [Lemma \ref{wheel-graf} (a)]. (middle) Envelope graphs $E_6$ (solid edges only) and $E_6^|$ (solid edges with additional dashed edge $\{w_3,w_6\}$) do not admit FORs(4) [Lemma \ref{wheel-graf} (b)]. (right) Graph $W^{||}_7$ --- a $7$-vertex graph without a FOR(4) [Lemma \ref{wheel-graf} (c)].}
		\label{wheel}
	\end{figure}
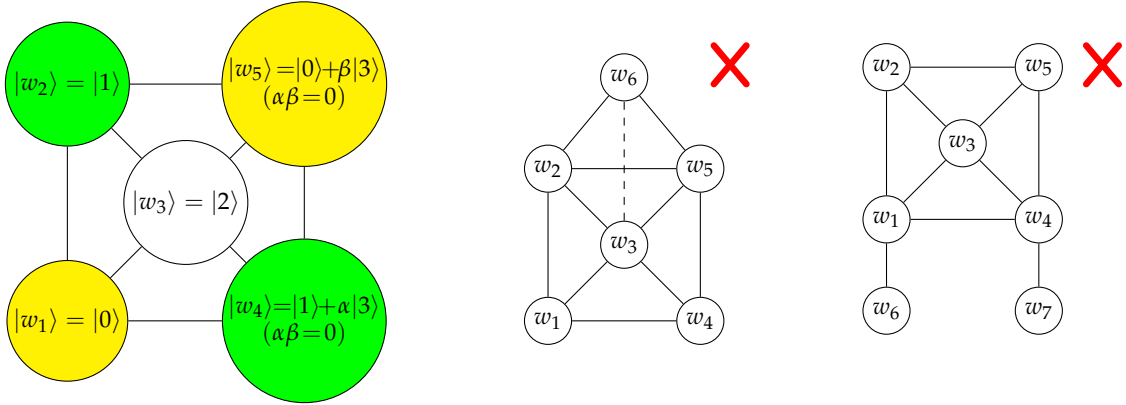

	We further have another graph leading to forbidden induced subgraphs (Fig. \ref{balonowy}).

\begin{lem}\label{balon-graf} 
	(a) Balloon graph $B_5$ has exactly one repeated vector in a FOR(4).  (b) Graph $B^{\circ}_6$ does not admit a FOR(4).   (c) $C_4^{(3)}$ does admit a FOR(4).
	\end{lem}
	\begin{proof} Both assertions are rather trivial. (a) We can set $\ket{w_1}=\zero, \ket{w_2}=\jeden, \ket{w_3}=\dwa, \ket{w_4}=\ket{3}$, as the corresponding graph is the clique $C_4$. It immediately follows that $\ket{w_5}=\ket{w_2}=\ket{1}$. (b) Neighbors of vertices with the same vector in a FOR are different, a contradiction. (c) The graph is composed of three cliques $C_4$, sharing certain vertices and edges (that is why the denotation for the graph). It follows that it must be $\ket{w_2}=\ket{w_5}$ and $\ket{w_3}=\ket{w_6}$, but $w_5$ and $w_6$ are not neighbors, while $w_2$ and $w_3$ are ones.
		\end{proof}

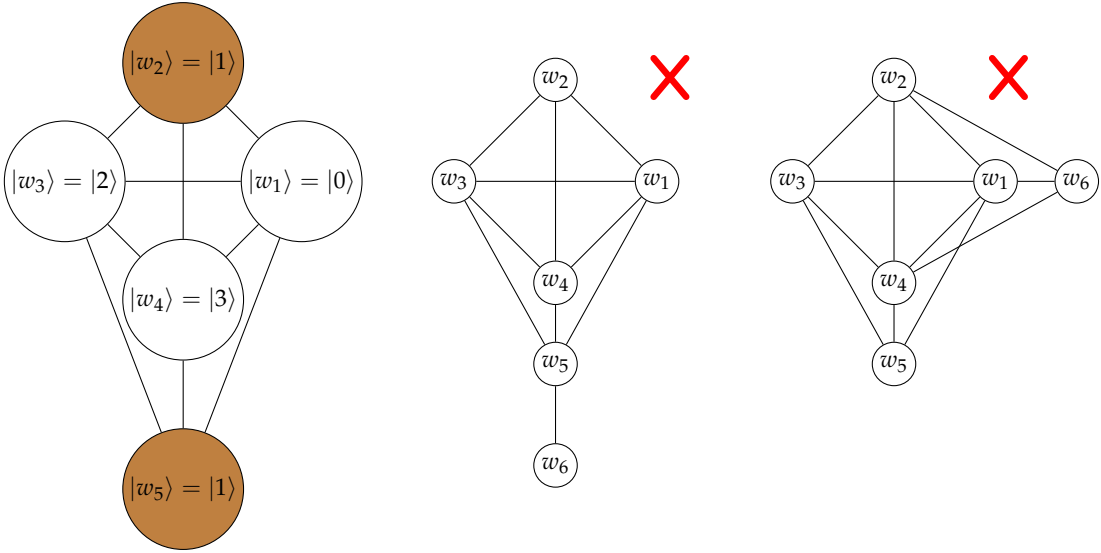
\begin{figure}[h]
	\centering
	\begin{tikzpicture}[scale=0.45]
		\begin{scope}[xshift=-1cm,scale=0.7]
			\node[circle, draw, fill=white, minimum size=8pt,inner sep=1.5pt] (w1) at
			(5,0) {$\ket{w_1}=\zero$};
			\node[circle, draw, fill=brown, minimum size=8pt,inner sep=1.5pt] (w2) at
			(0,5) {$\ket{w_2}=\jeden$};
			\node[circle, draw, fill=white, minimum size=8pt,inner sep=1.5pt] (w3) at
			(-5,0) {$\ket{w_3}=\dwa$};
			\node[circle, draw, fill=white, minimum size=8pt,inner sep=1.5pt] (w4) at
			(0,-5) {$\ket{w_4}=\ket{3}$};
			\node[circle, draw, fill=brown, minimum size=8pt,inner sep=1.5pt] (w5) at (0,-13) {$\ket{w_5}=\ket{1}$};
			\draw (w1)-- (w2) -- (w3) -- (w4) -- (w1); \draw (w3)--(w5)--(w1); 	%
			\draw (w1)--(w3); \draw (w2)--(w4)--(w5);
		\end{scope}
		\begin{scope}[xshift=10cm,scale=0.6]
			\node[circle, draw, fill=white, minimum size=1pt,inner sep=1.5pt] (w1) at
		(5,0) {$w_1$};
		\node[circle, draw, fill=white, minimum size=12pt,inner sep=1.5pt] (w2) at
		(0,5) {$w_2$};
		\node[circle, draw, fill=white, minimum size=12pt,inner sep=1.5pt] (w3) at
		(-5,0) {$w_3$};
		\node[circle, draw, fill=white, minimum size=12pt,inner sep=1.5pt] (w4) at
		(0,-5) {$w_4$};
		\node[circle, draw, fill=white, minimum size=12pt,inner sep=1.5pt] (w5) at (0,-9) {$w_5$};
			\node[circle, draw, fill=white, minimum size=12pt,inner sep=1.5pt] (w6) at (0,-14) {$w_6$};
		\draw (w1)-- (w2) -- (w3) -- (w4) -- (w1); \draw (w3)--(w5)--(w1); 	%
		\draw (w1)--(w3); \draw (w2)--(w4)--(w5)--(w6);
		\end{scope}
		\begin{scope}[xshift=20cm,scale=0.6]
			\node[circle, draw, fill=white, minimum size=1pt,inner sep=1.5pt] (w1) at
			(5,0) {$w_1$};
			\node[circle, draw, fill=white, minimum size=12pt,inner sep=1.5pt] (w2) at
			(0,5) {$w_2$};
			\node[circle, draw, fill=white, minimum size=12pt,inner sep=1.5pt] (w3) at
			(-5,0) {$w_3$};
			\node[circle, draw, fill=white, minimum size=12pt,inner sep=1.5pt] (w4) at
			(0,-5) {$w_4$};
			\node[circle, draw, fill=white, minimum size=12pt,inner sep=1.5pt] (w5) at (0,-9) {$w_5$};
			\node[circle, draw, fill=white, minimum size=12pt,inner sep=1.5pt] (w6) at (9,0) {$w_6$};
			\draw (w6)--(w1)-- (w2) -- (w3) -- (w4) -- (w1); \draw (w3)--(w5)--(w1); 	%
			\draw (w1)--(w3); \draw (w2)--(w4)--(w5); \draw (w4) -- (w6) -- (w2);
		\end{scope}
		%
		\begin{scope}[xshift=23.5cm,yshift=2.5cm,scale=1.2]
			\draw[red, line width=2.5pt, line cap=round, line join=round]
			(-0.5,0.95) -- (0.3,0) ;
			\draw[red, line width=2.5pt, line cap=round, line join=round]	(-0.5,0) --
			(0.3,0.95);
		\end{scope}
		%
		\begin{scope}[xshift=13.5cm,yshift=2.5cm,scale=1.2]
			\draw[red, line width=2.5pt, line cap=round, line join=round]
			(-0.5,0.95) -- (0.3,0) ;
			\draw[red, line width=2.5pt, line cap=round, line join=round]	(-0.5,0) --
			(0.3,0.95);
		\end{scope}
	\end{tikzpicture}
	\caption{(left) Balloon graph $B_5$: there is one repeated vector in its FOR(4) [Lemma \ref{balon-graf} (a)]. (middle) $B_6^{\circ}$ -- a $6$-vertex graph without a FOR(4) [Lemma \ref{balon-graf} (b)]. (right) Graph $C_4^{(3)}$ does not have a FOR(4) [Lemma \ref{balon-graf} (c)]. }
	\label{balonowy}
\end{figure}

\subsubsection{Three-ququart minimal GUPBs - remarks}\label{ququarty}

In the three-ququart case ($N=3$, $d=4$), it holds
$\mathfrak{C}_{\mathrm{min}}(4,3)=24$ and thus graphs with $24$ vertices are the objects of interest. Importantly, the minimal GUPBs do not saturate the lower bound in Eq. \eqref{kiribela-bound} and we find that the degrees of vertices in LOGs are either $7$ or $8$ \cite{kiribela}. The number of valid degree sequences in this case is $25$. Focusing only on relevant regular graphs, we find the following. They all have girth bounded as $\girth \le 4$. There are $141\;515\;621\:596\;238\;755\;266\;884\;806\;115\;631$ connected \cite{7-connected} and  $733\:460\:349\:818$ disconnected  \cite{7-disconnected} $7$-regular graphs. We consider disconnected $8$-regular, i.e., octic graphs more closely below.

There are $1\:473\:822\:243$ disconnected octic graphs on $24$ vertices \cite{disconnected-octic}. 
Their classification according to the division of vertices into two groups is as follows (dividing into more than two groups is impossible in this case):
\begin{itemize}
	\item[(a)] type A: $9+15$ vertices ($1\:470\:293\:676$ graphs). These graphs get immediately discarded as there is the unique octic graph with $9$ vertices --- the clique $C_9$,
	\item[(b)] type B: $10+14$ vertices ($3\:459\:386$ graphs).  This case also gets discarded as there is one octic graph with $10$ vertices and it contains $C_5$ as a subgraph (this can be verified by direct inspection or using Tur\'an's bound \cite{turan}),
	\item[(c)] type C: $11+13$ vertices ($6\times 10\:778 = 64\:716$ graphs). There is only one graph among the six $11$-vertex graphs, which does not contain $C_5$, but it contains $C_4^{(3)}$ as an induced subgraph (see Fig. \ref{type-C}); consequently, type C graphs are eliminated,
	\item[(d)] type D: $12+12$ vertices. There are $94$  octic graphs on $12$ vertices, giving rise to the total of $4\:465$ graphs.  Among the said twelve-vertex octic graphs we find $6$ graphs with $C_6$ and $75$ graphs with $C_5$, which are eliminated at once (accounting for the total of $81$ graphs). 
	In turn, the remaining thirteen graphs give rise to $91$ type D candidate graphs, which need to be checked for validity as LOGs. We have found that only two graphs, both shown in Fig. \ref{type-D}, do not get discarded by the induced forbidden subgraph $C_4^{(3)}$. Using again the enumeration of graphs from GENREG, these are graphs number $70$ (with maximum clique $C_4$) and $94$ (with maximum clique $C_3$), denoted henceforth as  $L_{70}$ and $L_{94}$, respectively. We further find that $E_6^{|}$ is an induced subgraph of $L_{70}$ and this graph gets discarded. In turn, the only relevant graph from this category is $L_{94}+L_{94}$. We deal with $L_{94}$ below.
\end{itemize}

Graph $L_{94}$ is a complete $3$-partite graph and as such it has a FOR even in dimension $d=3$:
\begin{align}
&	\repr{1}=\repr{10}=\repr{11}=\repr{12}=\zero, \non 
&	\repr{2}=\repr{7}=\repr{8}=\repr{9}=\jeden, \\
&	\repr{3}=\repr{4}=\repr{5}=\repr{6}=\dwa. \nonumber
	\end{align}
We will now show that even exploiting the extra dimension does not lead to a representation with appropriate spanning (saturation) properties. Let $\mathcal{A}_1=\{\repr{1}, \repr{10}, \repr{11}, \repr{12}\}$, $\mathcal{A}_2=\{\repr{2}, \repr{7}, \repr{8}, \repr{9}\}$, and $\mathcal{A}_3=\{\repr{3}, \repr{4}, \repr{5}, \repr{6} \}$. While the obvious upper bound on the dimension of the subspaces spanned by the sets is three, we can immediately infer that it actually must hold that $\sum_k \dim\mathrm{span} \mathcal{A}_k \le 4$ , as the division of vectors into three sets corresponds to a certain decomposition of the four dimensional Hilbert space, or its subspace, into orthogonal subspaces (the upper bound on the sum of dimensions is $d$ in the general case of complete $d$-partite graphs). This means that there are always  eight vectors in a representation that span only a two-dimensional subspace.
Consequently, any FOR(4) of $L_{94}+L_{94}$ does not comprise with the condition imposed by Fact \ref{general-fakt}.
 An exemplary FOR(4) for $L_{94}$ saturating dimensions of the subspaces spanned by the set is as follows:
\begin{align}
	&	\repr{1}=\repr{10}=\repr{11}=\repr{12}=\zero, \non 
	&	\repr{2}=\repr{7}=\repr{8}=\repr{9}=\jeden, \\
	&	\repr{3}=\dwa, \quad \repr{4}=\dwa+\ket{3}, \non 
	& \repr{5}=\dwa+2\ket{3},\quad \repr{6}=\dwa+3\ket{3}. \nonumber
\end{align}
\begin{figure}[h]
	\centering	
	\begin{tikzpicture}[scale=0.6]
		\begin{scope}
			\GraphInit[vstyle=Classic]
		\SetGraphUnit{2.5}
		\node[circle, draw=red,thick, fill=white, minimum size=10pt, inner sep=1.8pt] (v1) at ({90+0*360/11}:6.5cm) {$v_1$};
		\node[circle, draw=red,thick, fill=white, minimum size=12pt, inner sep=1.8pt] (v2) at ({90+1*360/11}:6.5cm) {$v_2$};
		\node[circle, draw=red,thick, fill=white, minimum size=12pt, inner sep=1.8pt] (v3) at ({90+2*360/11}:6.5cm) {$v_3$};
		\node[circle, draw=red,thick, fill=white, minimum size=12pt, inner sep=1.8pt] (v4) at ({90+3*360/11}:6.5cm) {$v_4$};
		\node[circle, draw, fill=white, minimum size=12pt, inner sep=1.8pt] (v5) at ({90+4*360/11}:6.5cm) {$v_5$};
		\node[circle, draw, fill=white, minimum size=12pt, inner sep=1.8pt] (v6) at ({90+5*360/11}:6.5cm) {$v_6$};
		\node[circle, draw, fill=white, minimum size=12pt, inner sep=1.8pt] (v7) at ({90+6*360/11}:6.5cm) {$v_7$};
		\node[circle, draw, fill=white, minimum size=12pt, inner sep=1.8pt] (v8) at ({90+7*360/11}:6.5cm) {$v_8$};	
		\node[circle, draw=red,thick, fill=white, minimum size=12pt, inner sep=1.8pt] (v9) at ({90+8*360/11}:6.5cm) {$v_9$};	
		\node[circle, draw=red,thick, fill=white, minimum size=12pt,inner sep=0.6pt] (v10) at ({90+9*360/11}:6.5cm) {$v_{10}$};	
		\node[circle, draw, fill=white, minimum size=12pt,inner sep=0.6pt] (v11) at ({90+10*360/11}:6.5cm) {$v_{11}$};		
	\draw[red,thick] (v1) -- (v2)	(v1) -- (v3)	(v1) -- (v4) 	(v1) -- (v9);
\draw 	(v1) -- (v5)	(v1) -- (v6)	(v1) -- (v7)	(v1) -- (v8);
	\draw[red,thick] (v2) -- (v3)
	(v2) -- (v4); \draw	(v2) -- (v5)
	(v2) -- (v6)	(v2) -- (v7)
	(v2) -- (v8); \draw[red,thick]	(v2) -- (v10);
	\draw[red,thick] (v3) -- (v4) (v3) -- (v9); \draw	(v3) -- (v5)
	(v3) -- (v6)	; 
\draw[red,line width=1.2pt]	(v3) -- (v10); \draw	(v3) -- (v11);
	\draw (v4) -- (v7)	(v4) -- (v8)
; \draw[red,thick] 	(v4) -- (v9)	(v4) -- (v10);\draw	(v4) -- (v11);
	\draw (v5) -- (v7)	(v5) -- (v8)	(v5) -- (v9)	(v5) -- (v10)	(v5) -- (v11);
	\draw (v6) -- (v7)	(v6) -- (v8)	(v6) -- (v9)	(v6) -- (v10)	(v6) -- (v11);
	\draw (v7) -- (v9)	(v7) -- (v10)	(v7) -- (v11);
	\draw (v8) -- (v9)	(v8) -- (v10)	(v8) -- (v11);
	\draw (v9) -- (v11);	\draw (v10) -- (v11);
		\end{scope}
		\begin{scope}[xshift=6.75cm,yshift=5.5cm,scale=1]
			\draw[red, line width=2.5pt, line cap=round, line join=round]
			(-0.5,0.95) -- (0.3,0) ;
			\draw[red, line width=2.5pt, line cap=round, line join=round]	(-0.5,0) --
			(0.3,0.95);
		\end{scope}
	\end{tikzpicture}
	\caption{The unique $11$-vertex octic graph without $C_5$ as a subgraph. The graph does not admit a FOR(4) as $C_4^{(3)}$ is its induced subgraph (highlighted in red).
	}\label{type-C}
\end{figure}
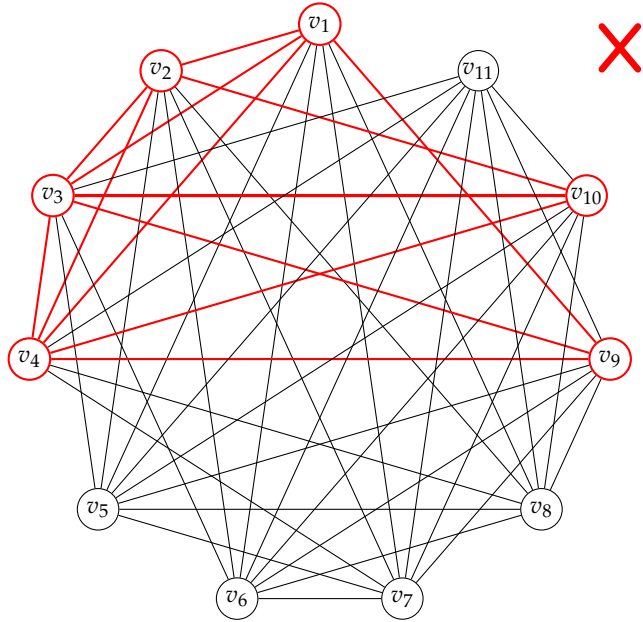

\begin{figure}[h]
	\centering

	\begin{tikzpicture}[scale=0.65]
		\begin{scope}[scale=0.8]
				\GraphInit[vstyle=Classic]
			\SetGraphUnit{2.5}
			\node[circle, draw=red,thick, fill=white, minimum size=10pt, inner sep=1.8pt] (v1) at ({90+0*360/12}:6.5cm) {$v_1$};
			\node[circle, draw=red,thick, fill=white, minimum size=12pt, inner sep=1.8pt] (v2) at ({90+1*360/12}:6.5cm) {$v_2$};
			\node[circle, draw=red,thick, fill=white, minimum size=12pt, inner sep=1.8pt] (v3) at ({90+2*360/12}:6.5cm) {$v_3$};
			\node[circle, draw=red,thick, fill=white, minimum size=12pt, inner sep=1.8pt] (v4) at ({90+3*360/12}:6.5cm) {$v_4$};
			\node[circle, draw, fill=white, minimum size=12pt, inner sep=1.8pt] (v5) at ({90+4*360/12}:6.5cm) {$v_5$};
			\node[circle, draw, fill=white, minimum size=12pt, inner sep=1.8pt] (v6) at ({90+5*360/12}:6.5cm) {$v_6$};
			\node[circle, draw, fill=white, minimum size=12pt, inner sep=1.8pt] (v7) at ({90+6*360/12}:6.5cm) {$v_7$};
			\node[circle, draw, fill=white, minimum size=12pt, inner sep=1.8pt] (v8) at ({90+7*360/12}:6.5cm) {$v_8$};	
			\node[circle, draw=red,thick, fill=white, minimum size=12pt, inner sep=1.8pt] (v9) at ({90+8*360/12}:6.5cm) {$v_9$};	
			\node[circle, draw=red,thick, fill=white, minimum size=12pt,inner sep=0.6pt] (v10) at ({90+9*360/12}:6.5cm) {$v_{10}$};	
			\node[circle, draw, fill=white, minimum size=12pt,inner sep=0.6pt] (v11) at ({90+10*360/12}:6.5cm) {$v_{11}$};	
				\node[circle, draw, fill=white, minimum size=12pt,inner sep=0.6pt] (v12) at ({90+11*360/12}:6.5cm) {$v_{12}$};	
		\draw[red,thick] (v1) -- (v2) (v1) -- (v3) (v1) -- (v4); \draw (v1) -- (v5) (v1) -- (v6) (v1) -- (v7) (v1) -- (v8); \draw[red,thick] (v1) -- (v9);
		\draw[red,thick] (v2) -- (v3) (v2) -- (v4);\draw (v2) -- (v5) (v2) -- (v6) (v2) -- (v7) (v2) -- (v8);\draw[red,thick] (v2) -- (v10);
		\draw[red,thick] (v3) -- (v4);\draw (v3) -- (v5) (v3) -- (v6) (v3) -- (v7) (v3) -- (v11) (v3) -- (v12);
		\draw (v4) -- (v8); \draw[red,thick] (v4) -- (v9) (v4) -- (v10); \draw (v4) -- (v11) (v4) -- (v12);
		\draw (v5) -- (v8) (v5) -- (v9) (v5) -- (v10) (v5) -- (v11) (v5) -- (v12);
		\draw (v6) -- (v8) (v6) -- (v9) (v6) -- (v10) (v6) -- (v11) (v6) -- (v12);
		\draw (v7) -- (v8) (v7) -- (v9) (v7) -- (v10) (v7) -- (v11) (v7) -- (v12);
		\draw (v8) -- (v11) (v8) -- (v12);
		\draw[red,thick] (v9) -- (v10); (v9) -- (v11) (v9) -- (v12);
		\draw (v10) -- (v11) (v10) -- (v12);
		\end{scope}
		\begin{scope}[xshift=5cm,yshift=5cm,scale=0.8]
			\draw[red, line width=2.5pt, line cap=round, line join=round]
			(-0.5,0.95) -- (0.3,0) ;
			\draw[red, line width=2.5pt, line cap=round, line join=round]	(-0.5,0) --
			(0.3,0.95);
		\end{scope}
	\begin{scope}[scale=0.8,xshift=18cm]
		\GraphInit[vstyle=Classic]
		\SetGraphUnit{2.5}
		\node[circle, draw, fill=white, minimum size=10pt, inner sep=1.8pt] (v1) at ({90+0*360/12}:6.5cm) {$v_1$};
	\node[circle, draw, fill=white, minimum size=12pt, inner sep=1.8pt] (v2) at ({90+1*360/12}:6.5cm) {$v_2$};
	\node[circle, draw, fill=white, minimum size=12pt, inner sep=1.8pt] (v3) at ({90+2*360/12}:6.5cm) {$v_3$};
	\node[circle, draw, fill=white, minimum size=12pt, inner sep=1.8pt] (v4) at ({90+3*360/12}:6.5cm) {$v_4$};
	\node[circle, draw, fill=white, minimum size=12pt, inner sep=1.8pt] (v5) at ({90+4*360/12}:6.5cm) {$v_5$};
	\node[circle, draw, fill=white, minimum size=12pt, inner sep=1.8pt] (v6) at ({90+5*360/12}:6.5cm) {$v_6$};
	\node[circle, draw, fill=white, minimum size=12pt, inner sep=1.8pt] (v7) at ({90+6*360/12}:6.5cm) {$v_7$};
	\node[circle, draw, fill=white, minimum size=12pt, inner sep=1.8pt] (v8) at ({90+7*360/12}:6.5cm) {$v_8$};	
	\node[circle, draw, fill=white, minimum size=12pt, inner sep=1.8pt] (v9) at ({90+8*360/12}:6.5cm) {$v_9$};	
	\node[circle, draw, fill=white, minimum size=12pt,inner sep=0.6pt] (v10) at ({90+9*360/12}:6.5cm) {$v_{10}$};	
	\node[circle, draw, fill=white, minimum size=12pt,inner sep=0.6pt] (v11) at ({90+10*360/12}:6.5cm) {$v_{11}$};	
	\node[circle, draw, fill=white, minimum size=12pt,inner sep=0.6pt] (v12) at ({90+11*360/12}:6.5cm) {$v_{12}$};	
	\draw (v1) -- (v2) (v1) -- (v3) (v1) -- (v4) (v1) -- (v5) (v1) -- (v6) (v1) -- (v7) (v1) -- (v8) (v1) -- (v9);
	\draw (v2) -- (v3) (v2) -- (v4) (v2) -- (v5) (v2) -- (v6) (v2) -- (v10) (v2) -- (v11) (v2) -- (v12);
	\draw (v3) -- (v7) (v3) -- (v8) (v3) -- (v9) (v3) -- (v10) (v3) -- (v11) (v3) -- (v12);
	\draw (v4) -- (v7) (v4) -- (v8) (v4) -- (v9) (v4) -- (v10) (v4) -- (v11) (v4) -- (v12);
	\draw (v5) -- (v7) (v5) -- (v8) (v5) -- (v9) (v5) -- (v10) (v5) -- (v11) (v5) -- (v12);
	\draw (v6) -- (v7) (v6) -- (v8) (v6) -- (v9) (v6) -- (v10) (v6) -- (v11) (v6) -- (v12);
	\draw (v7) -- (v10) (v7) -- (v11) (v7) -- (v12);
	\draw (v8) -- (v10) (v8) -- (v11) (v8) -- (v12);
	\draw (v9) -- (v10) (v9) -- (v11) (v9) -- (v12);
	
	\end{scope}
		\begin{scope}[xshift=19.5cm,yshift=4.5cm,scale=1.3]
			\draw[ufogreen, line width=2.5pt, line cap=round, line join=round]
			(-0.3,0.4) -- (0,0) -- (0.6,0.95);
		\end{scope}
	\end{tikzpicture}
	\caption{(left) Forbidden induced subgraph $E_{6}^|$ (cf. Fig. \ref{wheel}) highlighted in graph $L_{70}$. (right)  Graph $L_{94}$ has FORs(4) with eight vectors spanning only a two dimensional subspace (see Sec. \ref{ququarty}).
	}\label{type-D}
\end{figure}
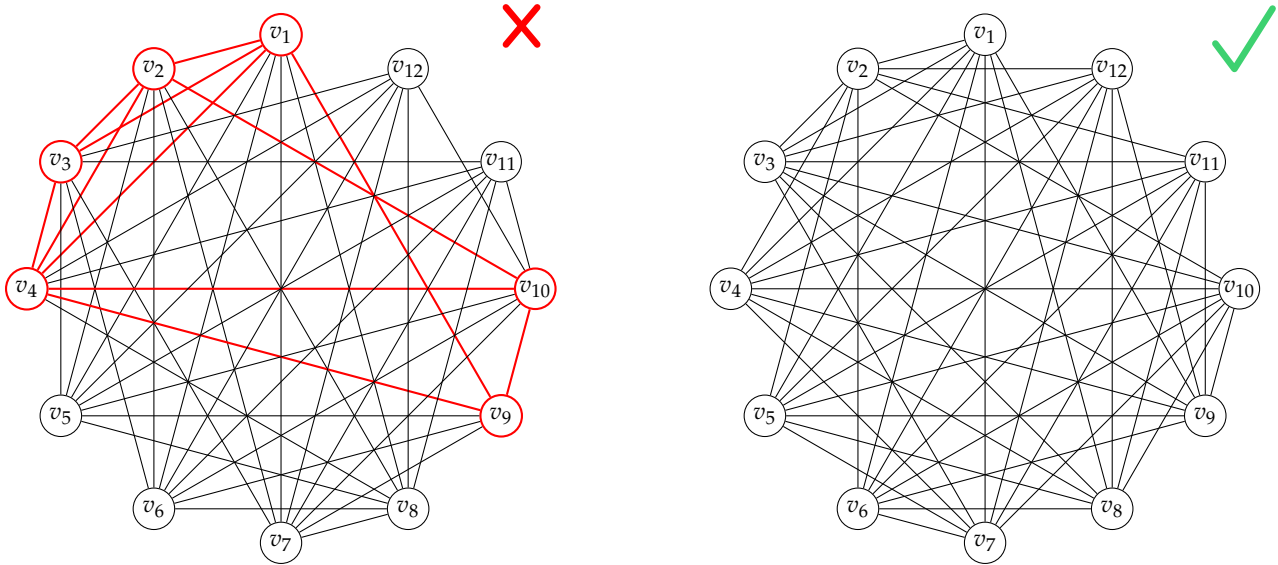

We have
not been able to find the  number of the connected graphs and we refer the
reader to the database \cite{8-regular} for the cases $N\le 22$, only noting here
that for $N=22$ this number is around $2.1 \times 10^{29}$.  

\section{Conclusions} \label{konkluzje}

We demonstrated that three‐qutrit genuinely unextendible product bases (GUPBs)
with the minimal theoretically permissible cardinality of 
thirteen do not exist. To this end, we exploited the recent characterization of
local orthogonality graphs for GUPBs  and employed
a forbidden induced subgraph characterization to efficiently prune the set of
candidate graphs. This approach allowed us to show that only two eligible graphs
have orthogonal representations, yet neither of them exhibits the required
spanning properties. We also discussed partial results concerning graphs associated with larger bases and systems involving ququart subsystems.

The proposed approach proved effective for small systems, but its direct application
becomes increasingly challenging as system dimensions and sizes of the bases grow. 
The method thus calls for improvements allowing for treating graphs with more
	vertices. First, one should apply other methods (see, e.g.,
	\cite{motif1,motif2,g-tries})  for the induced
	subgraph isomorphism problems. Second, one could also look into the
	possibility of graph elimination already at the
	construction stage of the candidate graphs, that is, instead of analyzing full
	graphs, smaller ones could be investigated and large sets of graphs could
	potentially be rejected early in the process. Nevertheless, it is rather clear
	that even a blend of such adjustments would still be profitable only in the
	problems of the limited scale. One must also bear in mind that the elimination
	of graphs is only part of the procedure --- it is still necessary to verify whether
	every graph in the set obtained at some stage of the process admits FORs. If
	they do (i.e., we have constructed an obstruction set) these FORs must be
	characterized and their validity for GUPBs confirmed 
 by checking certain properties, such as those outlined in Fact \ref{general-fakt}. It is likely that integrating advanced analytical tools with stronger graph
 characterizations---yet to be discovered---will be required.
 We believe that continued development along these lines may be key to
	overcoming the current limitations and advancing further the field.
	
	While the present results clarify and restrict several aspects of the problem, it appears too early to attempt to draw any far-reaching conclusions or to pose conjectures regarding the existence or nonexistence of GUPBs, even in the narrowed three-qutrit setting. Nevertheless, regardless of the current status of the problem, studying the consequences of both their existence and nonexistence is a worthwhile research direction that could reveal currently unexplored features of multipartite entanglement. In particular, the nonexistence case could be especially informative, possibly with unexpected consequences, for example in the form of no-go theorems for certain information-theoretic tasks other than those traditionally considered in this context, such as local discrimination of states.

\section{Acknowledgments}

The research was supported by Centre of Informatics Tricity Academic
Supercomputer \& Network (CI TASK) through the computing grant ''Unextendible
product bases and entangled subspaces''.


\appendix

\section{$3$-ladder graph}\label{drabina}

Here we consider the $3$-ladder graph $L_6$, see Fig. \ref{kwadraty}. 
\begin{figure}[htp]
\begin{tikzpicture}[scale=0.47]
	\centering
	\begin{scope}[scale=0.9]
		\node[circle, draw, fill=white, minimum size=12pt] (v1) at (2,0) {$w_1$};
		\node[circle, draw, fill=white, minimum size=12pt] (v2) at (2,3.5) {$w_2$};
		\node[circle, draw, fill=white, minimum size=12pt] (v3) at (-2,3.5)
		{$w_3$};
		\node[circle, draw, fill=white, minimum size=12pt] (v4) at (-2,0) {$w_4$};
		\node[circle, draw, fill=white, minimum size=12pt] (v5) at (-2,-3.5)
		{$w_5$};
		\node[circle, draw, fill=white, minimum size=12pt] (v6) at (2,-3.5) {$w_6$};
		%
		\draw (v1) -- (v2) -- (v3) -- (v4) -- (v1) ;
		\draw (v4) -- (v5) --(v6)--(v1);
	\end{scope}
	\begin{scope}[xshift=10.5cm,scale=0.9]
		\node[circle, draw=red,thick, fill=white, minimum size=10pt, inner sep=1.8pt] (v1) at
		({90+0*360/13}:6.5cm) {$v_1$};
		\node[circle, draw=red,thick, fill=white, minimum size=12pt, inner sep=1.8pt] (v2) at
		({90+1*360/13}:6.5cm) {$v_2$};
		\node[circle, draw=red,thick, fill=white, minimum size=12pt, inner sep=1.8pt] (v3) at
		({90+2*360/13}:6.5cm) {$v_3$};
		\node[circle, draw, fill=white, minimum size=12pt, inner sep=1.8pt] (v4) at
		({90+3*360/13}:6.5cm) {$v_4$};
		\node[circle, draw, fill=white, minimum size=12pt, inner sep=1.8pt] (v5) at
		({90+4*360/13}:6.5cm) {$v_5$};
		\node[circle, draw=red,thick, fill=white, minimum size=12pt, inner sep=1.8pt] (v6) at
		({90+5*360/13}:6.5cm) {$v_6$};
		\node[circle, draw, fill=white, minimum size=12pt, inner sep=1.8pt] (v7) at
		({90+6*360/13}:6.5cm) {$v_7$};
		\node[circle, draw=red,thick, fill=white, minimum size=12pt, inner sep=1.8pt] (v8) at
		({90+7*360/13}:6.5cm) {$v_8$};	
		\node[circle, draw, fill=white, minimum size=12pt, inner sep=1.8pt] (v9) at
		({90+8*360/13}:6.5cm) {$v_9$};	
		\node[circle, draw, fill=white, minimum size=12pt,inner sep=0.6pt] (v10) at
		({90+9*360/13}:6.5cm) {$v_{10}$};	
		\node[circle, draw, fill=white, minimum size=12pt,inner sep=0.6pt] (v11) at
		({90+10*360/13}:6.5cm) {$v_{11}$};	
		\node[circle, draw=red,thick, fill=white, minimum size=12pt,inner sep=0.6pt] (v12) at
		({90+11*360/13}:6.5cm) {$v_{12}$};	
		\node[circle, draw, fill=white, minimum size=12pt,inner sep=0.6pt] (v13) at
		({90+12*360/13}:6.5cm) {$v_{13}$};		
		\draw[red,thick] (v1)--(v2); 	\draw[red,thick]  (v1)--(v3);	\draw 
		(v1)--(v4);	\draw  (v1)--(v5);
		\draw[red,thick] (v2)--(v6);	\draw  (v2)--(v7);	\draw[red,thick] 
		(v2)--(v8);
		\draw[red,thick] (v3)--(v6);	\draw  (v3)--(v9);	\draw  (v3)--(v10);
		\draw (v4)--(v7);	\draw  (v4)--(v11) ;	\draw (v4)--(v12);
		\draw (v5)--(v9); 	\draw (v5)--(v11);	\draw  (v5)--(v13);
		\draw (v6)--(v11)	;\draw[red,thick]  (v6)--(v12);
		\draw (v7)--(v9);	\draw  (v7)--(v13);
		\draw (v8)--(v9);	\draw  (v8)--(v10);	\draw[red,thick]  (v8)--(v12);
		\draw (v9)--(v11);
		\draw (v10)--(v11);	\draw  (v10)--(v13);
		\draw (v12)--(v13);
	\end{scope}
	\begin{scope}[xshift=16cm,yshift=5cm,scale=0.9]
		\draw[-stealth, thick, bend left=40] (0,0) to (3,0);
	\end{scope}
	\begin{scope}[xshift=24cm,scale=0.9]
		\node[circle, draw=red,thick, fill=white, minimum size=10pt, inner sep=1.8pt] (v1) at
		({90+0*360/13}:6.5cm) {$v_1$};
		\node[circle, draw=red,thick, fill=white, minimum size=12pt, inner sep=1.8pt] (v2) at
		({90+1*360/13}:6.5cm) {$v_2$};
		\node[circle, draw=red,thick, fill=white, minimum size=12pt, inner sep=1.8pt] (v3) at
		({90+2*360/13}:6.5cm) {$v_8$};
		\node[circle, draw, fill=white, minimum size=12pt, inner sep=0.6pt] (v4) at
		({90+3*360/13}:6.5cm) {$v_{10}$};
		\node[circle, draw, fill=white, minimum size=12pt, inner sep=0.6pt] (v5) at
		({90+4*360/13}:6.5cm) {$v_{11}$};
		\node[circle, draw, fill=white, minimum size=12pt, inner sep=1.8pt] (v6) at
		({90+5*360/13}:6.5cm) {$v_4$};
		\node[circle, draw, fill=white, minimum size=12pt, inner sep=1.8pt] (v7) at
		({90+6*360/13}:6.5cm) {$v_7$};
		\node[circle, draw, fill=white, minimum size=12pt, inner sep=1.8pt] (v8) at
		({90+7*360/13}:6.5cm) {$v_9$};	
		\node[circle, draw=red,thick, fill=white, minimum size=12pt, inner sep=1.8pt] (v9) at
		({90+8*360/13}:6.5cm) {$v_3$};	
		\node[circle, draw=red,thick, fill=white, minimum size=12pt,inner sep=1.8pt] (v10) at
		({90+9*360/13}:6.5cm) {$v_{6}$};	
		\node[circle, draw=red,thick, fill=white, minimum size=12pt,inner sep=0.6pt] (v11) at
		({90+10*360/13}:6.5cm) {$v_{12}$};	
		\node[circle, draw, fill=white, minimum size=12pt,inner sep=0.6pt] (v12) at
		({90+11*360/13}:6.5cm) {$v_{13}$};	
		\node[circle, draw, fill=white, minimum size=12pt,inner sep=1.8pt] (v13) at
		({90+12*360/13}:6.5cm) {$v_{5}$};		
		\draw[red,thick] (v1)--(v2); \draw (v1) -- (v6); \draw[red,thick] (v1) --
		(v9); \draw (v1) -- (v13);
		\draw[red,thick] (v2) -- (v3); \draw (v2) -- (v7); \draw[red,thick] (v2) --
		(v10);
		\draw (v3) -- (v4);			\draw (v3) -- (v8);			\draw[red,thick] (v3) -- (v11);
		\draw (v4) -- (v5); \draw (v4) -- (v9);		\draw (v4) -- (v12);
		\draw (v5) -- (v6);		\draw (v5) -- (v10);	\draw (v5) -- (v13);
		\draw (v6) -- (v7);\draw (v6) -- (v11);			
		\draw (v7) -- (v8); \draw (v7) -- (v12);			\draw (v8) -- (v9);
		\draw (v8) -- (v13);			\draw[red,thick] (v9) -- (v10);
		\draw[red,thick] (v10) -- (v11); \draw (v11) -- (v12);			\draw (v12) --
		(v13);
	\end{scope}
	\begin{scope}[xshift=30cm,yshift=5cm,scale=1.2]
		\draw[red, line width=2.5pt, line cap=round, line join=round]
		(-0.5,0.95) -- (0.3,0) ;
		\draw[red, line width=2.5pt, line cap=round, line join=round]	(-0.5,0) --
		(0.3,0.95);
	\end{scope}
\end{tikzpicture}
\caption{(left) $3$-ladder graph $L_6$ --- a $6$-vertex graph without a FOR(3) [Lemma \ref{2-kwadratki}]. (middle) $L_6$
	highlighted in  $M_{10778}$. (right) $M_{10778}$ is isomorphic to a circulant
	graph; the graph shows the vertices rearranged to better visualize it.}
\label{kwadraty}
\end{figure}
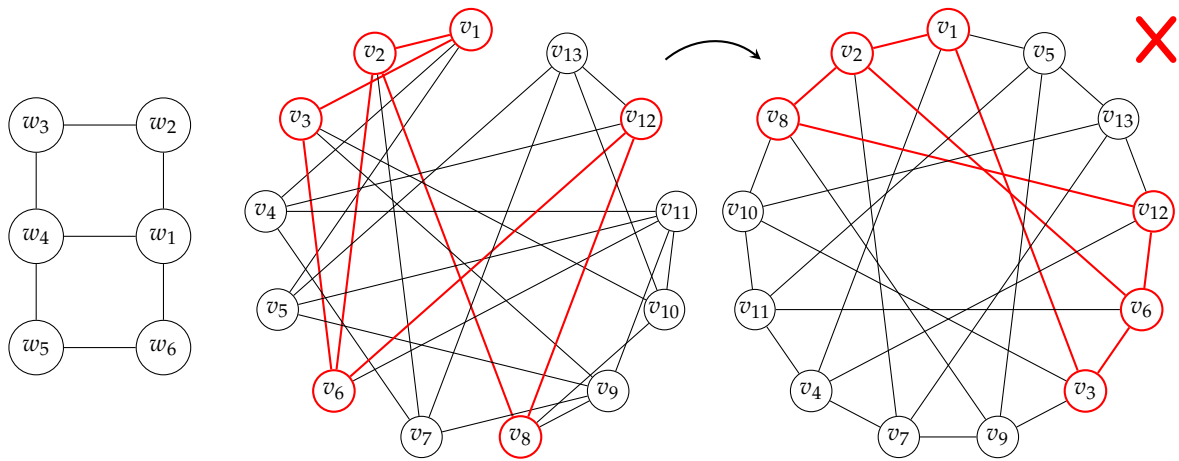
\begin{lem}\label{2-kwadratki}
$L_6$ does not admit a FOR(3). 
\end{lem}
\begin{proof}
It follows from Lemma \ref{4cykl} that it necessarily holds that
$\ket{w_2}=\ket{w_4}$  or
$\ket{w_3}=\ket{w_5}$. None of these cases is possible,
though, as $\mathcal{N}(w_2) \ne \mathcal{N}(w_4)$ and   $\mathcal{N}(w_3) \ne
\mathcal{N}(w_5)$.
\end{proof}

The adjacency matrix for the graph is
\begin{equation}
\mathcal{A}(L_6)=	\left(
\begin{array}{cccccc}
	0 & 1 & 0 & 1 &0& 1 \\
	1 & 0 & 1 & 0 &0&0\\
	0 & 1 & 0 & 1 &0&0\\
	1 & 0 & 1 & 0& 1&0\\
	0 &0 & 0 &1& 0&1\\
	1&0&0&0&1&0
\end{array} \right).
\end{equation}

\section{Analysis of graphs from Eq. \eqref{grafy-N}} \label{eny}

Here we analyze the $14$-vertex $4$-regular graphs not discarded by $\mathcal{O}_3$ [cf. Eq. \ref{grafy-N}]; they are depicted in Fig. \ref{14-vertices-appendix}. We will now show that there is only a single graph among them not having a  FOR(3). 
\begin{figure}[h!tp]
	\begin{tikzpicture}[scale=0.415]
		\centering
		\begin{scope}[xshift=5.8cm,yshift=-9.3cm,scale=1.2]
			\draw[red, line width=2.5pt, line cap=round, line join=round]
			(-0.5,0.95) -- (0.3,0) ;
			\draw[red, line width=2.5pt, line cap=round, line join=round]	(-0.5,0) --
			(0.3,0.95);
		\end{scope}
		\begin{scope}[xshift=5.8cm,yshift=5.7cm,scale=1.2]
			\draw[ufogreen, line width=2.5pt, line cap=round, line join=round]
			(-0.3,0.4) -- (0,0) -- (0.6,0.95);
		\end{scope}
		\begin{scope}[xshift=20.3cm,yshift=5.7cm,scale=1.2]
			\draw[ufogreen, line width=2.5pt, line cap=round, line join=round]
			(-0.3,0.4) -- (0,0) -- (0.6,0.95);
		\end{scope}
		\begin{scope}[xshift=34.7cm,yshift=5.7cm,scale=1.2]
			\draw[ufogreen, line width=2.5pt, line cap=round, line join=round]
			(-0.3,0.4) -- (0,0) -- (0.6,0.95);
		\end{scope}
		\begin{scope}[xshift=20.3cm,yshift=-9.3cm,scale=1.2]
			\draw[ufogreen, line width=2.5pt, line cap=round, line join=round]
			(-0.3,0.4) -- (0,0) -- (0.6,0.95);
		\end{scope}
		\begin{scope}[xshift=34.7cm,yshift=-9.3cm,scale=1.2]
			\draw[ufogreen, line width=2.5pt, line cap=round, line join=round]
			(-0.3,0.4) -- (0,0) -- (0.6,0.95);
		\end{scope}
		\begin{scope}[xshift=20.3cm,yshift=-24.3cm,scale=1.2]
			\draw[ufogreen, line width=2.5pt, line cap=round, line join=round]
			(-0.3,0.4) -- (0,0) -- (0.6,0.95);
		\end{scope}
		\begin{scope}
			\GraphInit[vstyle=Classic]
			\SetGraphUnit{2.5}
			\node[circle, draw, fill=white, minimum size=10pt, inner sep=1.8pt] (v1) at ({90+0*360/14}:6.5cm) {$v_1$};
			\node[circle, draw, fill=white, minimum size=12pt, inner sep=1.8pt] (v2) at ({90+1*360/14}:6.5cm) {$v_2$};
			\node[circle, draw, fill=brown, minimum size=12pt, inner sep=1.8pt] (v3) at ({90+2*360/14}:6.5cm) {$v_3$};
			\node[circle, draw, fill=brown, minimum size=12pt, inner sep=1.8pt] (v4) at ({90+3*360/14}:6.5cm) {$v_4$};
			\node[circle, draw, fill=brown, minimum size=12pt, inner sep=1.8pt] (v5) at ({90+4*360/14}:6.5cm) {$v_5$};
			\node[circle, draw, fill=white, minimum size=12pt, inner sep=1.8pt] (v6) at ({90+5*360/14}:6.5cm) {$v_6$};
			\node[circle, draw, fill=white, minimum size=12pt, inner sep=1.8pt] (v7) at ({90+6*360/14}:6.5cm) {$v_7$};
			\node[circle, draw, fill=white, minimum size=12pt, inner sep=1.8pt] (v8) at ({90+7*360/14}:6.5cm) {$v_8$};	
			\node[circle, draw, fill=white, minimum size=12pt, inner sep=1.8pt] (v9) at ({90+8*360/14}:6.5cm) {$v_9$};	
			\node[circle, draw, fill=lime, minimum size=12pt,inner sep=0.6pt] (v10) at ({90+9*360/14}:6.5cm) {$v_{10}$};	
			\node[circle, draw, fill=lime, minimum size=12pt,inner sep=0.6pt] (v11) at ({90+10*360/14}:6.5cm) {$v_{11}$};	
			\node[circle, draw, fill=lime, minimum size=12pt,inner sep=0.6pt] (v12) at ({90+11*360/14}:6.5cm) {$v_{12}$};	
			\node[circle, draw, fill=white, minimum size=12pt,inner sep=0.6pt] (v13) at ({90+12*360/14}:6.5cm) {$v_{13}$};	
			\node[circle, draw, fill=white, minimum size=12pt,inner sep=0.6pt] (v14) at ({90+13*360/14}:6.5cm) {$v_{14}$};		
			\draw (v1) -- (v2);		\draw (v1) -- (v3);		\draw (v1) -- (v4);		\draw (v1) -- (v5);		
			\draw (v2) -- (v3);		\draw (v2) -- (v4);		\draw (v2) -- (v5);		
			\draw (v3) -- (v6);		\draw (v3) -- (v7);		
			\draw (v4) -- (v6);		\draw (v4) -- (v7);		
			\draw (v5) -- (v6);		\draw (v5) -- (v7);		
			\draw (v6) -- (v8);				\draw (v7) -- (v9);
			\draw (v8) -- (v10);		\draw (v8) -- (v11);		\draw (v8) -- (v12);
			\draw (v9) -- (v10);		\draw (v9) -- (v11);		\draw (v9) -- (v12);				\draw (v10) -- (v13);		\draw (v10) -- (v14);		
			\draw (v11) -- (v13);		\draw (v11) -- (v14);				\draw (v12) -- (v13);		\draw (v12) -- (v14);				\draw (v13) -- (v14);
		\end{scope}
		\begin{scope}[xshift=14.5cm]
			\GraphInit[vstyle=Classic]
			\SetGraphUnit{2.5}
			\node[circle, draw, fill=white, minimum size=10pt, inner sep=1.8pt] (v1) at ({90+0*360/14}:6.5cm) {$v_1$};
			\node[circle, draw, fill=lime, minimum size=12pt, inner sep=1.8pt] (v2) at ({90+1*360/14}:6.5cm) {$v_2$};
			\node[circle, draw, fill=brown, minimum size=12pt, inner sep=1.8pt] (v3) at ({90+2*360/14}:6.5cm) {$v_3$};
			\node[circle, draw, fill=brown, minimum size=12pt, inner sep=1.8pt] (v4) at ({90+3*360/14}:6.5cm) {$v_4$};
			\node[circle, draw, fill=lime, minimum size=12pt, inner sep=1.8pt] (v5) at ({90+4*360/14}:6.5cm) {$v_5$};
			\node[circle, draw, fill=white, minimum size=12pt, inner sep=1.8pt] (v6) at ({90+5*360/14}:6.5cm) {$v_6$};
			\node[circle, draw, fill=white, minimum size=12pt, inner sep=1.8pt] (v7) at ({90+6*360/14}:6.5cm) {$v_7$};
			\node[circle, draw, fill=gray, minimum size=12pt, inner sep=1.8pt] (v8) at ({90+7*360/14}:6.5cm) {$v_8$};	
			\node[circle, draw, fill=gray, minimum size=12pt, inner sep=1.8pt] (v9) at ({90+8*360/14}:6.5cm) {$v_9$};	
			\node[circle, draw, fill=violet, minimum size=12pt,inner sep=0.6pt] (v10) at ({90+9*360/14}:6.5cm) {$v_{10}$};	
			\node[circle, draw, fill=violet, minimum size=12pt,inner sep=0.6pt] (v11) at ({90+10*360/14}:6.5cm) {$v_{11}$};	
			\node[circle, draw, fill=cyan, minimum size=12pt,inner sep=0.6pt] (v12) at ({90+11*360/14}:6.5cm) {$v_{12}$};	
			\node[circle, draw, fill=cyan, minimum size=12pt,inner sep=0.6pt] (v13) at ({90+12*360/14}:6.5cm) {$v_{13}$};	
			\node[circle, draw, fill=cyan, minimum size=12pt,inner sep=0.6pt] (v14) at ({90+13*360/14}:6.5cm) {$v_{14}$};		
			\draw (v1) -- (v2);\draw (v1) -- (v3);
			\draw (v1) -- (v4);\draw (v1) -- (v5);
			\draw (v2) -- (v3);\draw (v2) -- (v4);
			\draw (v2) -- (v6);\draw (v3) -- (v5);
			\draw (v3) -- (v7);\draw (v4) -- (v5);
			\draw (v4) -- (v7);\draw (v5) -- (v6);
			\draw (v6) -- (v8);\draw (v6) -- (v9);
			\draw (v7) -- (v10);\draw (v7) -- (v11);
			\draw (v8) -- (v12);\draw (v8) -- (v13);
			\draw (v8) -- (v14);\draw (v9) -- (v12);
			\draw (v9) -- (v13);\draw (v9) -- (v14);
			\draw (v10) -- (v12);\draw (v10) -- (v13);
			\draw (v10) -- (v14);\draw (v11) -- (v12);
			\draw (v11) -- (v13);\draw (v11) -- (v14);
		\end{scope}
		\begin{scope}[xshift=29cm]
			\GraphInit[vstyle=Classic]
			\SetGraphUnit{2.5}
			\node[circle, draw, fill=white, minimum size=10pt, inner sep=1.8pt] (v1) at ({90+0*360/14}:6.5cm) {$v_1$};
			\node[circle, draw, fill=white, minimum size=12pt, inner sep=1.8pt] (v2) at ({90+1*360/14}:6.5cm) {$v_2$};
			\node[circle, draw, fill=brown, minimum size=12pt, inner sep=1.8pt] (v3) at ({90+2*360/14}:6.5cm) {$v_3$};
			\node[circle, draw, fill=brown, minimum size=12pt, inner sep=1.8pt] (v4) at ({90+3*360/14}:6.5cm) {$v_4$};
			\node[circle, draw, fill=white, minimum size=12pt, inner sep=1.8pt] (v5) at ({90+4*360/14}:6.5cm) {$v_5$};
			\node[circle, draw, fill=white, minimum size=12pt, inner sep=1.8pt] (v6) at ({90+5*360/14}:6.5cm) {$v_6$};
			\node[circle, draw, fill=cyan, minimum size=12pt, inner sep=1.8pt] (v7) at ({90+6*360/14}:6.5cm) {$v_7$};
			\node[circle, draw, fill=cyan, minimum size=12pt, inner sep=1.8pt] (v8) at ({90+7*360/14}:6.5cm) {$v_8$};	
			\node[circle, draw, fill=lime, minimum size=12pt, inner sep=1.8pt] (v9) at ({90+8*360/14}:6.5cm) {$v_9$};	
			\node[circle, draw, fill=lime, minimum size=12pt,inner sep=0.6pt] (v10) at ({90+9*360/14}:6.5cm) {$v_{10}$};	
			\node[circle, draw, fill=white, minimum size=12pt,inner sep=0.6pt] (v11) at ({90+10*360/14}:6.5cm) {$v_{11}$};	
			\node[circle, draw, fill=white, minimum size=12pt,inner sep=0.6pt] (v12) at ({90+11*360/14}:6.5cm) {$v_{12}$};	
			\node[circle, draw, fill=white, minimum size=12pt,inner sep=0.6pt] (v13) at ({90+12*360/14}:6.5cm) {$v_{13}$};	
			\node[circle, draw, fill=white, minimum size=12pt,inner sep=0.6pt] (v14) at ({90+13*360/14}:6.5cm) {$v_{14}$};		
			\draw (v1) -- (v2);
			\draw (v1) -- (v3);		\draw (v1) -- (v4);
			\draw (v1) -- (v5);		\draw (v2) -- (v3);		\draw (v2) -- (v4);		\draw (v2) -- (v6);
			\draw (v3) -- (v7);		\draw (v3) -- (v8);
			\draw (v4) -- (v7);		\draw (v4) -- (v8);
			\draw (v5) -- (v9);		\draw (v5) -- (v10);
			\draw (v5) -- (v11);		
			\draw (v6) -- (v9);		\draw (v6) -- (v10);
			\draw (v6) -- (v12);		
			\draw (v7) -- (v13);		\draw (v7) -- (v14);
			\draw (v8) -- (v13);
			\draw (v8) -- (v14);		
			\draw (v9) -- (v11);		\draw (v9) -- (v12);
			\draw (v10) -- (v11);
			\draw (v10) -- (v12);		
			\draw (v11) -- (v13);		
			\draw (v12) -- (v14);
			\draw (v13) -- (v14);

		\end{scope}
		%
		%
		%
		\begin{scope}[yshift=-15cm]
			\GraphInit[vstyle=Classic]
			\SetGraphUnit{2.5}
			\node[circle, draw=red,thick, fill=white, minimum size=10pt, inner sep=1.8pt] (v1) at ({90+0*360/14}:6.5cm) {$v_1$};
			\node[circle, draw=red,thick, fill=white, minimum size=12pt, inner sep=1.8pt] (v2) at ({90+1*360/14}:6.5cm) {$v_2$};
			\node[circle, draw=red,thick, fill=white, minimum size=12pt, inner sep=1.8pt] (v3) at ({90+2*360/14}:6.5cm) {$v_3$};
			\node[circle, draw=red,thick, fill=white, minimum size=12pt, inner sep=1.8pt] (v4) at ({90+3*360/14}:6.5cm) {$v_4$};
			\node[circle, draw=red,thick, fill=white, minimum size=12pt, inner sep=1.8pt] (v5) at ({90+4*360/14}:6.5cm) {$v_5$};
			\node[circle, draw=red,thick, fill=white, minimum size=12pt, inner sep=1.8pt] (v6) at ({90+5*360/14}:6.5cm) {$v_6$};
			\node[circle, draw=red,thick, fill=white, minimum size=12pt, inner sep=1.8pt] (v7) at ({90+6*360/14}:6.5cm) {$v_7$};
			\node[circle, draw=red,thick, fill=white, minimum size=12pt, inner sep=1.8pt] (v8) at ({90+7*360/14}:6.5cm) {$v_8$};	
			\node[circle, draw=red,thick, fill=white, minimum size=12pt, inner sep=1.8pt] (v9) at ({90+8*360/14}:6.5cm) {$v_9$};	
			\node[circle, draw=red,thick, fill=white, minimum size=12pt,inner sep=0.6pt] (v10) at ({90+9*360/14}:6.5cm) {$v_{10}$};	
			\node[circle, draw, fill=white, minimum size=12pt,inner sep=0.6pt] (v11) at ({90+10*360/14}:6.5cm) {$v_{11}$};	
			\node[circle, draw=red,thick, fill=white, minimum size=12pt,inner sep=0.6pt] (v12) at ({90+11*360/14}:6.5cm) {$v_{12}$};	
			\node[circle, draw, fill=white, minimum size=12pt,inner sep=0.6pt] (v13) at ({90+12*360/14}:6.5cm) {$v_{13}$};	
			\node[circle, draw, fill=white, minimum size=12pt,inner sep=0.6pt] (v14) at ({90+13*360/14}:6.5cm) {$v_{14}$};		
			\draw[red,thick] (v1) -- (v2);	\draw[red,thick] (v1) -- (v3);	\draw[red,thick] (v1) -- (v4);
			\draw[red,thick] (v1) -- (v5);		\draw[red,thick] (v2) -- (v3);
			\draw[red,thick] (v2) -- (v6);	\draw[red,thick] (v2) -- (v7);	
			\draw[red,thick] (v3) -- (v8);	\draw[red,thick] (v3) -- (v9);		\draw[red,thick] (v4) -- (v5);
			\draw[red,thick] (v4) -- (v10);	\draw (v4) -- (v11);		\draw[red,thick] (v5) -- (v12);	\draw (v5) -- (v13);	
			\draw[red,thick] (v6) -- (v7);	\draw[red,thick] (v6) -- (v10);
			\draw (v6) -- (v11);		\draw[red,thick] (v7) -- (v12);
			\draw (v7) -- (v13);		\draw[red,thick] (v8) -- (v9);
			\draw[red,thick] (v8) -- (v10);	\draw (v8) -- (v11);		\draw[red,thick] (v9) -- (v12);	\draw (v9) -- (v13);	
			\draw (v10) -- (v14);	\draw (v11) -- (v14);	\draw (v12) -- (v14);	\draw (v13) -- (v14);

		\end{scope}
		\begin{scope}[xshift=14.5cm,yshift=-15cm]
			\GraphInit[vstyle=Classic]
			\SetGraphUnit{2.5}
			\node[circle, draw, fill=white, minimum size=10pt, inner sep=1.8pt] (v1) at ({90+0*360/14}:6.5cm) {$v_1$};
			\node[circle, draw, fill=brown, minimum size=12pt, inner sep=1.8pt] (v2) at ({90+1*360/14}:6.5cm) {$v_2$};
			\node[circle, draw, fill=brown, minimum size=12pt, inner sep=1.8pt] (v3) at ({90+2*360/14}:6.5cm) {$v_3$};
			\node[circle, draw, fill=brown, minimum size=12pt, inner sep=1.8pt] (v4) at ({90+3*360/14}:6.5cm) {$v_4$};
			\node[circle, draw, fill=white, minimum size=12pt, inner sep=1.8pt] (v5) at ({90+4*360/14}:6.5cm) {$v_5$};
			\node[circle, draw, fill=white, minimum size=12pt, inner sep=1.8pt] (v6) at ({90+5*360/14}:6.5cm) {$v_6$};
			\node[circle, draw, fill=white, minimum size=12pt, inner sep=1.8pt] (v7) at ({90+6*360/14}:6.5cm) {$v_7$};
			\node[circle, draw, fill=white, minimum size=12pt, inner sep=1.8pt] (v8) at ({90+7*360/14}:6.5cm) {$v_8$};	
			\node[circle, draw, fill=lime, minimum size=12pt, inner sep=1.8pt] (v9) at ({90+8*360/14}:6.5cm) {$v_9$};	
			\node[circle, draw, fill=lime, minimum size=12pt,inner sep=0.6pt] (v10) at ({90+9*360/14}:6.5cm) {$v_{10}$};	
			\node[circle, draw, fill=lime, minimum size=12pt,inner sep=0.6pt] (v11) at ({90+10*360/14}:6.5cm) {$v_{11}$};	
			\node[circle, draw, fill=white, minimum size=12pt,inner sep=0.6pt] (v12) at ({90+11*360/14}:6.5cm) {$v_{12}$};	
			\node[circle, draw, fill=white, minimum size=12pt,inner sep=0.6pt] (v13) at ({90+12*360/14}:6.5cm) {$v_{13}$};	
			\node[circle, draw, fill=white, minimum size=12pt,inner sep=0.6pt] (v14) at ({90+13*360/14}:6.5cm) {$v_{14}$};		
			\draw (v1) -- (v2);	\draw (v1) -- (v3);
			\draw (v1) -- (v4);	\draw (v1) -- (v5);	\draw (v2) -- (v6);	\draw (v2) -- (v7);	\draw (v2) -- (v8);	
			\draw (v3) -- (v6);	\draw (v3) -- (v7);	\draw (v3) -- (v8);		\draw (v4) -- (v6);
			\draw (v4) -- (v7);	\draw (v4) -- (v8);		\draw (v5) -- (v9);	\draw (v5) -- (v10);
			\draw (v5) -- (v11);	
			\draw (v6) -- (v12);\draw (v7) -- (v13);\draw (v8) -- (v14);	
			\draw (v9) -- (v12);	\draw (v9) -- (v13);	\draw (v9) -- (v14);	
			\draw (v10) -- (v12);
			\draw (v10) -- (v13);	\draw (v10) -- (v14);	
			\draw (v11) -- (v12);	\draw (v11) -- (v13);	\draw (v11) -- (v14);

		\end{scope}
		\begin{scope}[xshift=29cm,yshift=-15cm]
			\GraphInit[vstyle=Classic]
			\SetGraphUnit{2.5}
			\node[circle, draw, fill=cyan, minimum size=10pt, inner sep=1.8pt] (v1) at ({90+0*360/14}:6.5cm) {$v_1$};
			\node[circle, draw, fill=brown, minimum size=12pt, inner sep=1.8pt] (v2) at ({90+1*360/14}:6.5cm) {$v_2$};
			\node[circle, draw, fill=brown, minimum size=12pt, inner sep=1.8pt] (v3) at ({90+2*360/14}:6.5cm) {$v_3$};
			\node[circle, draw, fill=magenta, minimum size=12pt, inner sep=1.8pt] (v4) at ({90+3*360/14}:6.5cm) {$v_4$};
			\node[circle, draw, fill=magenta, minimum size=12pt, inner sep=1.8pt] (v5) at ({90+4*360/14}:6.5cm) {$v_5$};
			\node[circle, draw, fill=cyan, minimum size=12pt, inner sep=1.8pt] (v6) at ({90+5*360/14}:6.5cm) {$v_6$};
			\node[circle, draw, fill=gray, minimum size=12pt, inner sep=1.8pt] (v7) at ({90+6*360/14}:6.5cm) {$v_7$};
			\node[circle, draw, fill=gray, minimum size=12pt, inner sep=1.8pt] (v8) at ({90+7*360/14}:6.5cm) {$v_8$};	
			\node[circle, draw, fill=orange, minimum size=12pt, inner sep=1.8pt] (v9) at ({90+8*360/14}:6.5cm) {$v_9$};	
			\node[circle, draw, fill=orange, minimum size=12pt,inner sep=0.6pt] (v10) at ({90+9*360/14}:6.5cm) {$v_{10}$};	
			\node[circle, draw, fill=pink, minimum size=12pt,inner sep=0.6pt] (v11) at ({90+10*360/14}:6.5cm) {$v_{11}$};	
			\node[circle, draw, fill=pink, minimum size=12pt,inner sep=0.6pt] (v12) at ({90+11*360/14}:6.5cm) {$v_{12}$};	
			\node[circle, draw, fill=violet, minimum size=12pt,inner sep=0.6pt] (v13) at ({90+12*360/14}:6.5cm) {$v_{13}$};	
			\node[circle, draw, fill=violet, minimum size=12pt,inner sep=0.6pt] (v14) at ({90+13*360/14}:6.5cm) {$v_{14}$};		
			\draw (v1) -- (v2);	\draw (v1) -- (v3);	\draw (v1) -- (v4);
			\draw (v1) -- (v5);		\draw (v2) -- (v6);
			\draw (v2) -- (v7);	\draw (v2) -- (v8);	
			\draw (v3) -- (v6);	\draw (v3) -- (v7);
			\draw (v3) -- (v8);		\draw (v4) -- (v6);
			\draw (v4) -- (v9);	\draw (v4) -- (v10);	
			\draw (v5) -- (v6);	\draw (v5) -- (v9);
			\draw (v5) -- (v10);		\draw (v7) -- (v11);
			\draw (v7) -- (v12);		\draw (v8) -- (v11);
			\draw (v8) -- (v12);		\draw (v9) -- (v13);
			\draw (v9) -- (v14);		\draw (v10) -- (v13);
			\draw (v10) -- (v14);		\draw (v11) -- (v13);
			\draw (v11) -- (v14);		\draw (v12) -- (v13);
			\draw (v12) -- (v14);		
		\end{scope}
		%
		%
		%
		\begin{scope}[xshift=14.5cm,yshift=-30cm]
			\GraphInit[vstyle=Classic]
			\SetGraphUnit{2.5}
			\node[circle, draw, fill=lime, minimum size=10pt, inner sep=1.8pt] (v1) at ({90+0*360/14}:6.5cm) {$v_1$};
			\node[circle, draw, fill=brown, minimum size=12pt, inner sep=1.8pt] (v2) at ({90+1*360/14}:6.5cm) {$v_2$};
			\node[circle, draw, fill=brown, minimum size=12pt, inner sep=1.8pt] (v3) at ({90+2*360/14}:6.5cm) {$v_3$};
			\node[circle, draw, fill=new-color, minimum size=12pt, inner sep=1.8pt] (v4) at ({90+3*360/14}:6.5cm) {$v_4$};
			\node[circle, draw, fill=new-color, minimum size=12pt, inner sep=1.8pt] (v5) at ({90+4*360/14}:6.5cm) {$v_5$};
			\node[circle, draw, fill=lime, minimum size=12pt, inner sep=1.8pt] (v6) at ({90+5*360/14}:6.5cm) {$v_6$};
			\node[circle, draw, fill=white, minimum size=12pt, inner sep=1.8pt] (v7) at ({90+6*360/14}:6.5cm) {$v_7$};
			\node[circle, draw, fill=white, minimum size=12pt, inner sep=1.8pt] (v8) at ({90+7*360/14}:6.5cm) {$v_8$};	
			\node[circle, draw, fill=white, minimum size=12pt, inner sep=1.8pt] (v9) at ({90+8*360/14}:6.5cm) {$v_9$};	
			\node[circle, draw, fill=white, minimum size=12pt,inner sep=0.6pt] (v10) at ({90+9*360/14}:6.5cm) {$v_{10}$};	
			\node[circle, draw, fill=orange, minimum size=12pt,inner sep=0.6pt] (v11) at ({90+10*360/14}:6.5cm) {$v_{11}$};	
			\node[circle, draw, fill=orange, minimum size=12pt,inner sep=0.6pt] (v12) at ({90+11*360/14}:6.5cm) {$v_{12}$};	
			\node[circle, draw, fill=cyan, minimum size=12pt,inner sep=0.6pt] (v13) at ({90+12*360/14}:6.5cm) {$v_{13}$};	
			\node[circle, draw, fill=cyan, minimum size=12pt,inner sep=0.6pt] (v14) at ({90+13*360/14}:6.5cm) {$v_{14}$};		
			\draw (v1) -- (v2);	\draw (v1) -- (v3);
			\draw (v1) -- (v4);	\draw (v1) -- (v5);	
			\draw (v2) -- (v6);	\draw (v2) -- (v7);
			\draw (v2) -- (v8);		\draw (v3) -- (v6);
			\draw (v3) -- (v7);	\draw (v3) -- (v8);	
			\draw (v4) -- (v6);	\draw (v4) -- (v9);
			\draw (v4) -- (v10);		\draw (v5) -- (v6);
			\draw (v5) -- (v9);	\draw (v5) -- (v10);	
			\draw (v7) -- (v11);	\draw (v7) -- (v12);	
			\draw (v8) -- (v13);	\draw (v8) -- (v14);	
			\draw (v9) -- (v11);	\draw (v9) -- (v12);	
			\draw (v10) -- (v13);	\draw (v10) -- (v14);	
			\draw (v11) -- (v13);	\draw (v11) -- (v14);	
			\draw (v12) -- (v13);	\draw (v12) -- (v14);	
		\end{scope}
	\end{tikzpicture}
	\caption{The seven graphs not eliminated by $\mathcal{O}_3$ in case of $14$-vertex $4$-regular graphs. First row (from the left): $N_{2359}$, $N_{11743}$, $N_{36919}$; second row (from the left): $N_{80015}$, $N_{87949}$, $N_{87956}$; third row: $N_{87957}$. Only one of these graphs -- $N_{80015}$ -- does not admit a FOR(3); a forbidden induced subgraph is highlighted in red ($\hat{N}_{11}$, see Appendix \ref{n11}). Same color vertices in a graph have necessarily identical vectors in a FOR(3) (see the main text).} \label{14-vertices-appendix}
\end{figure}
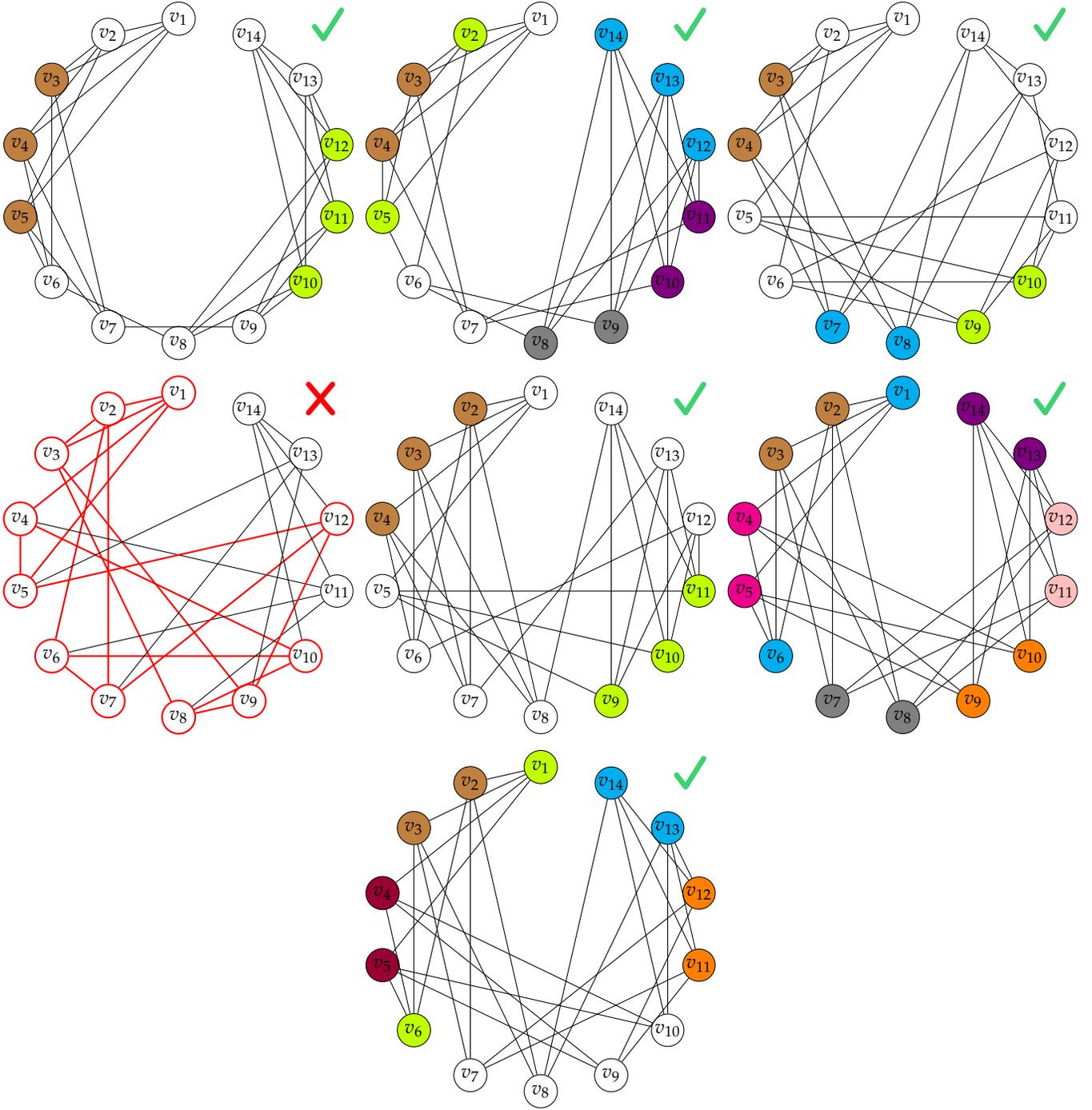

\subsection{Graph $N_{2359}$}

An exemplary FOR(3) is as follows:
\begin{align}
	&\repr{1}=\ket{0}, \quad \repr{2}=\ket{1}, \non &\repr{3}=\repr{4}=\repr{5}=\ket{2}, \nonumber \\ 
	&\repr{6}=\ket{0}+2\ket{1}, \quad\repr{7}=\ket{0}-\ket{1}, \nonumber\\
	&\repr{8}=2\ket{0}-\ket{1}+4\ket{2}, \quad \repr{9}=\ket{0}+\ket{1}-\ket{2}, \\
	& \repr{10}=\repr{11}=\repr{12} =\ket{0}-2\ket{1}-\ket{2}, \nonumber \\
	& \repr{13}= 7\ket{0}+4\ket{1}-\ket{2},\quad \repr{14}=\ket{0}+\ket{1}-3\ket{2}.\nonumber 
\end{align}

It is evident that it necessarily holds $\repr{3}=\repr{4}=\repr{5}$ (vertices marked in brown in Fig. \ref{14-vertices-appendix}) and $\repr{10}=\repr{11}=\repr{12}$ (lime vertices) in any FOR(3) of $N_{2359}$. In turn, this graph cannot be a LOG for a GUPB as it violates Fact \ref{general-fakt}, because $\dim \mathrm{span} \{\repr{3},\repr{4},\repr{5}, \repr{10},\repr{11},\repr{12}\}=2$.

\subsection{Graph $N_{11743}$}

 A FOR(3) for the graph is:
\begin{align}
	& \repr{1}= \ket{0},\non  &\repr{2}=\repr{5}=\ket{1}, \non &\repr{3}=\repr{4}=\ket{2}, \non
	& \repr{6}= \zero+\dwa, \quad \repr{7}=\zero+\jeden,\\
	& \repr{8}=\repr{9}=\zero+2\jeden-\dwa, \non 
	&\repr{10}=\repr{11}=\zero-\jeden+\dwa, \non
	&\repr{12}=\repr{13}=\repr{14}=\zero-2\jeden-3\dwa. \nonumber
\end{align}

Considering relevant diamond graphs, we immediately find that $\repr{2}=\repr{5}$ (lime vertices) and $\repr{3}=\repr{4}$ (brown vertices) in any FOR(3). Further, considering the square graphs with vertices $(v_8,v_{10},v_{13},v_{14})$ and $(v_9,v_{11},v_{12},v_{14})$ and the relevant neighborhoods we find $\repr{12}=\repr{13}=\repr{14}$ (blue vertices). Similar consideration for the square graphs $(v_6,v_8,v_9,v_{14})$ and $(v_7,v_{10},v_{11},v_{14})$ gives, respectively, $\repr{8}=\repr{9}$ (gray) and $\repr{10}=\repr{11}$ (violet). These FORs are admissible by Fact \ref{general-fakt}.

\subsection{Graph $N_{36919}$}

A FOR(3) is given by the following set of vectors:
\begin{align}
& \repr{1}=\ket{0}, \quad \repr{2}=\ket{1}, \non &\repr{3}=\repr{4}=\ket{2}, \non
&\repr{5}= \ket{1}-x\ket{2}, \quad \repr{6}=\ket{0}-\ket{2}, \non
&\repr{7}=\repr{8}=\ket{0}-\ket{1}, \non &\repr{9}=\repr{10}=\ket{0}+x\ket{1}+\ket{2}, \non
&\repr{11}=(1+x^2)\ket{0}-x\ket{1}-\ket{2}, \\
&\repr{12}=x\ket{0}-2\ket{1}+x\ket{2}, \non
&\repr{13}=(x-2)\ket{0}+(x-2)\ket{1}+2x\ket{2},\non
&\repr{14}=x\ket{0}+x\ket{1}-(x-2)\ket{2},\nonumber
\end{align}
where $x$ is the real solution of $x^3-3x^2+x-2=0$, which is found with standard methods to be
\begin{align}
	 x= 1+ \sqrt[3]{\frac{3}{2}+\sqrt{\frac{211}{108}}} +\sqrt[3]{\frac{3}{2}-\sqrt{\frac{211}{108}}} \simeq 2.89329
	\end{align}

It is trivial to see that $\repr{3}=\repr{4}$ (brown vertices), $\repr{7}=\repr{8}$ (blue), and $\repr{9}=\repr{10}$ (lime) in any FOR(3). 

We have not found any argument either supporting or refuting this graph's validity as a candidate LOG. Nevertheless, among the graphs considered, it appears to be the most promising, as it has the fewest repeated vectors—only three in total.

\subsection{Graph $N_{80015}$}\label{n11}

This graph has  an appealing structure, which may not be that easy to infer from its form shown in Fig. \ref{14-vertices-appendix}.
To show non-existence of a FOR(3) for this graph, we consider its $11$-vertex induced subgraph with vertices $(v_{1},v_2,v_3,v_4,v_5,v_6,v_7,v_8,v_9,v_{10},v_{12})$ (see Fig. \ref{redraw-graph}); we will call it $\widehat{N}_{11}$. 

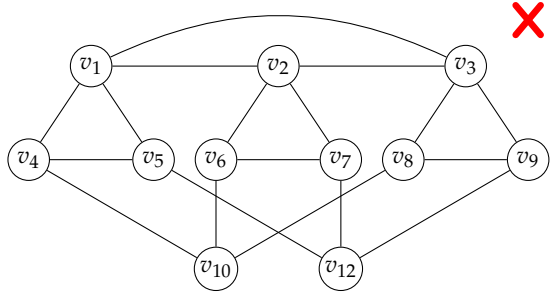
\begin{figure}[htp]
	\begin{tikzpicture}[scale=0.415]
		\centering
	\GraphInit[vstyle=Classic]
	\SetGraphUnit{2.5}
[every node/.style={circle, draw, fill=white, inner sep=2pt, minimum size=1cm}]
		\begin{scope}
		 \node[circle, draw, inner sep=2pt, minimum size=12pt, inner sep=1.8pt] (v1) at (-6,0) {$v_1$};
		\node[circle, draw, inner sep=2pt,minimum size=12pt, inner sep=1.8pt] (v2) at (0,0) {$v_2$};
		\node[circle, draw, inner sep=2pt, minimum size=12pt, inner sep=1.8pt] (v3) at (6,0) {$v_3$};
	\node[circle, draw, inner sep=2pt, minimum size=12pt, inner sep=1.8pt] (v4) at (-8,-3) {$v_4$};
	\node[circle, draw, inner sep=2pt, minimum size=12pt, inner sep=1.8pt] (v5) at (-4,-3) {$v_5$};
		\node[circle, draw, inner sep=2pt, minimum size=12pt, inner sep=1.8pt] (v6) at (-2,-3) {$v_6$};
			\node[circle, draw, inner sep=2pt, minimum size=12pt, inner sep=1.8pt] (v7) at (2,-3) {$v_7$};
				\node[circle, draw, inner sep=2pt, minimum size=12pt, inner sep=1.8pt] (v8) at (4,-3) {$v_8$};
					\node[circle, draw, inner sep=2pt, minimum size=12pt, inner sep=1.8pt] (v9) at (8,-3) {$v_9$};
					\node[circle, draw, inner sep=2pt, minimum size=10pt, inner sep=1pt] (v10) at (-2,-6.5) {$v_{10}$};
					\node[circle, draw, inner sep=2pt, minimum size=10pt, inner sep=1pt] (v12) at (2,-6.5) {$v_{12}$};
					%
		\draw (v1) -- (v2);
		\draw (v2) -- (v3);
		\draw (v1) -- (v4) -- (v5) -- (v1);
		\draw  (v2) -- (v6) -- (v7) -- (v2);
			\draw  (v3) -- (v8) -- (v9) -- (v3);
		\draw (v4) -- (v10) -- (v6);
			\draw (v8) -- (v10);
		\draw (v5) -- (v12) -- (v9);
			\draw (v7) -- (v12);
		\draw[bend left=25] (v1) to (v3);
	%
\end{scope}
	\begin{scope}[xshift=8.1cm,yshift=1cm,scale=1]
	\draw[red, line width=2.5pt, line cap=round, line join=round]
	(-0.5,0.95) -- (0.3,0) ;
	\draw[red, line width=2.5pt, line cap=round, line join=round]	(-0.5,0) --
	(0.3,0.95);
\end{scope}
\end{tikzpicture}
\caption{Induced subgraph $\widehat{N}_{11}$ of graph $N_{80015}$ without a FOR(3). } \label{redraw-graph}
\end{figure}
 Let $\repr{1}=\zero$, $\repr{2}=\jeden$, and $\repr{3}=\dwa$. We then have
 \begin{align}
& \repr{4}=\jeden + x\dwa, \quad \repr{5}=  x^* \jeden - \dwa  , \non
& \repr{6}=\zero+y\dwa, \quad \repr{7}=y^* \zero-\dwa, \\
&\repr{8}= \zero + z\jeden    , \quad \repr{9}= z^* \zero-\jeden    . \nonumber
\end{align}
with $x,y,z \ne 0$.  Further, it must hold
\begin{align}\label{war-dim}
	\dim \mathrm{span} \{\repr{4},\repr{6},\repr{8}\} =  \dim \mathrm{span} \{\repr{5},\repr{7},\repr{9}\} =2.
\end{align}
We thus  obtain conditions:
\begin{align}
\left| \begin{array}{ccc}
	0 & 1 & 1 \\
	1 & 0 & z \\
	x &y & 0
\end{array}   \right|=xz+y=0, \quad \left| \begin{array}{rrr}
0 & y^* & z^* \\
x^* & 0 & -1 \\
-1 & -1 & 0
\end{array}   \right|=-x^*z^*+y^*=0,
\end{align}
which are contradictory under our assumptions. Graph $\widehat{N}_{11}$ constitutes a forbidden induced subgraph. 

\subsection{Graph $N_{87949}$}

A FOR(3) for the graph in question is:
\begin{align}
	& \repr{1}=\ket{0}, \quad\repr{2}=\repr{3}=\repr{4}=\ket{1}, \non 	
	& \repr{5}=\ket{1}-\ket{2}, \quad \repr{6}=5\ket{0}-\ket{2}, \quad \repr{7}=3\ket{0}-\ket{2}, \quad \repr{8}=2\ket{0}-\ket{2}, \\
	& \repr{9}=\repr{10}=\repr{11}=\ket{0}-\ket{1}-\ket{2}, \non
	&\repr{12}=\ket{0}-4\ket{1}+5\ket{2},\quad \repr{13}=\ket{0}-2\ket{1}+3\ket{2},\quad \repr{14}=\ket{0}-\ket{1}+2\ket{2}. \nonumber
\end{align}

Considering the square graph with vertices $(v_3,v_4,v_6,v_7)$ we find that $\repr{3}=\repr{4}$ or/and $\repr{6}=\repr{7}$. Since $\mathcal{N}(v_6)\ne \mathcal{N}(v_7)$, we conclude $\repr{3}=\repr{4}$. Similar reasoning for the square graph $(v_2,v_3,v_6,v_7)$ leads to $\repr{2}=\repr{3}$, and in turn $\repr{2}=\repr{3}=\repr{4}$ (brown vertices). Analogous consideration for the square graphs $(v_9,v_{10}, v_{13},v_{14})$ and $(v_9,v_{10}, v_{12},v_{13})$ results in $\repr{9}=\repr{10}=\repr{11}$ (lime vertices). This necessarily holds in any FOR(3) of the graph, eliminating it -- by Fact \ref{general-fakt} -- from the set of candidate LOGs for GUPBs, as $\dim\mathrm{span}\{\repr{2},\repr{3},\repr{4},\repr{9},\repr{10},\repr{11}\}=2$.

\subsection{Graph $N_{87956}$}

This graph has very simple FORs(3). An exemplary one is as follows:

\begin{align}
& \repr{1}=\repr{6}=\ket{0}, \quad \repr{2}=\repr{3}=\ket{1}, \non
&\repr{4}=\repr{5}=2\ket{1}-\ket{2},  \quad	\repr{7}=\repr{8}=\ket{0}+\ket{2}, \\
&\repr{9}=\repr{10}=\ket{0}-\ket{1}-2\ket{2}, \non
&\repr{11}=\repr{12}=4\ket{0}-\ket{1}-4\ket{2}, \non
& \repr{13}=\repr{14}=2\ket{0}-4\ket{1}+3\ket{2}, \nonumber
\end{align}

Consideration of the following square graphs and relevant neighborhoods gives the properties holding for any FOR(3) (we first give the vertices of the square graph and then the consequence): $(v_2,v_3,v_6,v_7)$ --- $\repr{2}=\repr{3}$ (brown vertices), $(v_1,v_3,v_4,v_6)$ --- $\repr{1}=\repr{6}$ (blue);  $(v_9,v_{11},v_{13},v_{14})$ --- $\repr{13}=\repr{14}$ (violet), $(v_3,v_7,v_8,v_{12})$ --- $\repr{7}=\repr{8}$ (gray), $(v_4,v_9,v_{10},v_{13})$ --- $\repr{9}=\repr{10}$ (orange), $(v_1,v_4,v_5,v_9)$ --- $\repr{4}=\repr{5}$ (magenta), $(v_7,v_{11},v_{12},v_{14})$ --- $\repr{11}=\repr{12}$ (pink). The FORs are not filtered out by Fact \ref{general-fakt}.

\subsection{Graph $N_{87957}$}

A FOR(3) for this graph is:
\begin{align}
	& \repr{1}=\repr{6}=\ket{0}, \quad\repr{2}=\repr{3}=\ket{1}, \quad \repr{4}=\repr{5}=2\ket{1}-\ket{2}, \non & \repr{7}=\ket{0}+\ket{2}, \quad \repr{8}=7\ket{0}-12\ket{2}, \\
	&\repr{9}=\ket{0}-2\ket{1}-4\ket{2}, \quad \repr{10}=\ket{0}-\ket{1}-2\ket{2}, \non
	& \repr{11}=\repr{12}=2\ket{0}+5\ket{1}-2\ket{2}, \non 
	& \repr{13}=\repr{14}=12\ket{0}-2\ket{1}+7\ket{2}.\nonumber
	\end{align}

As above, we consider the square graphs and find the following properties (using the above notation): $(v_1,v_2,v_4,v_6)$ --- $\repr{1}=\repr{6}$ (lime), $(v_2,v_3,v_7,v_8)$ --- $\repr{2}=\repr{3}$ (brown), $(v_4,v_5,v_9,v_{10})$ --- $\repr{4}=\repr{5}$ (violet),  $(v_{7},v_{9},v_{11},v_{12})$ --- $\repr{11}=\repr{12}$ (orange), $(v_{10},v_{11},v_{13},v_{14})$ --- $\repr{13}=\repr{14}$, which  are necessarily shared by any FOR(3). The FORs are permissible by Fact \ref{general-fakt}.

\section{$14$-vertex $3$-regular graphs}\label{3-reg}

Here, we provide a brief analysis of the candidate fourteen-vertex three-regular graphs (see Sec. \ref{remarksy}).

\subsection{Connected graphs} \label{3-reg-connected}

We first analyze the connected graphs defined by set $\mathcal{O}_3\setminus C_4$.

\subsubsection{Girth-$3$ graphs}

These are graphs with the following GENREG indices: $5,14,17,254,264,265,267,269,270,271,272,275,276,288, 289$,  $292,293,294,301,305,306,308,310,313,314,320,322, 341,342,353,354,364,367,370,372,374,382,388,389,393,397,399$. They all have permissible FORs(3).

In Fig. \ref{3-reg-254}, we show graph $254$ for which a FOR(3) is as follows: 
\begin{align}\label{254-for}
	&\repr{1}=\zero,\quad \repr{2}=\jeden,\quad \repr{3}=\dwa, \quad \repr{4}=2\jeden-3\dwa, \non
	& \repr{5}= \zero+\dwa, \quad \repr{6}=2\zero-\jeden, \quad \repr{7}=\repr{8}=2\zero-3\jeden-\dwa, \non
	&  \repr{9}=\repr{10}= \zero+2\jeden-3\dwa, \quad \repr{11}=4\zero+2\jeden+\dwa, \\
	& \repr{12}=\repr{13}= \zero+\jeden+\dwa,\quad \repr{14}= \zero-3\jeden+2\dwa. \nonumber
\end{align}

\begin{figure}[htp]
	\begin{tikzpicture}[scale=0.5]
		\centering
		\begin{scope}[xshift=6.75cm,yshift=5.5cm,scale=1.3]
			\draw[ufogreen, line width=2.5pt, line cap=round, line join=round]
			(-0.3,0.4) -- (0,0) -- (0.6,0.95);
		\end{scope}
		\begin{scope}
			\GraphInit[vstyle=Classic]
			\SetGraphUnit{2.5}
			\node[circle, draw, fill=white, minimum size=10pt, inner sep=1.8pt] (v1) at ({90+0*360/14}:6.5cm) {$v_1$};
			\node[circle, draw, fill=white, minimum size=12pt, inner sep=1.8pt] (v2) at ({90+1*360/14}:6.5cm) {$v_2$};
			\node[circle, draw, fill=white, minimum size=12pt, inner sep=1.8pt] (v3) at ({90+2*360/14}:6.5cm) {$v_3$};
			\node[circle, draw, fill=white, minimum size=12pt, inner sep=1.8pt] (v4) at ({90+3*360/14}:6.5cm) {$v_4$};
			\node[circle, draw, fill=white, minimum size=12pt, inner sep=1.8pt] (v5) at ({90+4*360/14}:6.5cm) {$v_5$};
			\node[circle, draw, fill=white, minimum size=12pt, inner sep=1.8pt] (v6) at ({90+5*360/14}:6.5cm) {$v_6$};
			\node[circle, draw, fill=white, minimum size=12pt, inner sep=1.8pt] (v7) at ({90+6*360/14}:6.5cm) {$v_7$};
			\node[circle, draw, fill=white, minimum size=12pt, inner sep=1.8pt] (v8) at ({90+7*360/14}:6.5cm) {$v_8$};	
			\node[circle, draw, fill=white, minimum size=12pt, inner sep=1.8pt] (v9) at ({90+8*360/14}:6.5cm) {$v_9$};	
			\node[circle, draw, fill=white, minimum size=12pt,inner sep=0.6pt] (v10) at ({90+9*360/14}:6.5cm) {$v_{10}$};	
			\node[circle, draw, fill=white, minimum size=12pt,inner sep=0.6pt] (v11) at ({90+10*360/14}:6.5cm) {$v_{11}$};	
			\node[circle, draw, fill=white, minimum size=12pt,inner sep=0.6pt] (v12) at ({90+11*360/14}:6.5cm) {$v_{12}$};	
			\node[circle, draw, fill=white, minimum size=12pt,inner sep=0.6pt] (v13) at ({90+12*360/14}:6.5cm) {$v_{13}$};	
			\node[circle, draw, fill=white, minimum size=12pt,inner sep=0.6pt] (v14) at ({90+13*360/14}:6.5cm) {$v_{14}$};		
			\draw (v1)--(v2); \draw (v1)--(v3); \draw (v1)--(v4); \draw (v2)--(v3); \draw (v2)--(v5);
		\draw (v3)--(v6); \draw (v4)--(v7); \draw (v4)--(v8); \draw (v5)--(v7); \draw (v5)--(v8);
		\draw (v6)--(v9); \draw (v6)--(v10); \draw (v7)--(v11); \draw (v8)--(v11);
		\draw (v9)--(v12); \draw (v9)--(v13); 	\draw (v10)--(v12); \draw (v10)--(v13);
		\draw (v11)--(v14); \draw (v12)--(v14); \draw (v13)--(v14); 
		\end{scope}
		\end{tikzpicture}
		\caption{Graph $254$ -- one of the $14$-vertex  $3$-regular graphs with $\girth=3$. The graph has a FOR(3) shown in Eq. \eqref{254-for}.}\label{3-reg-254}
		\end{figure}
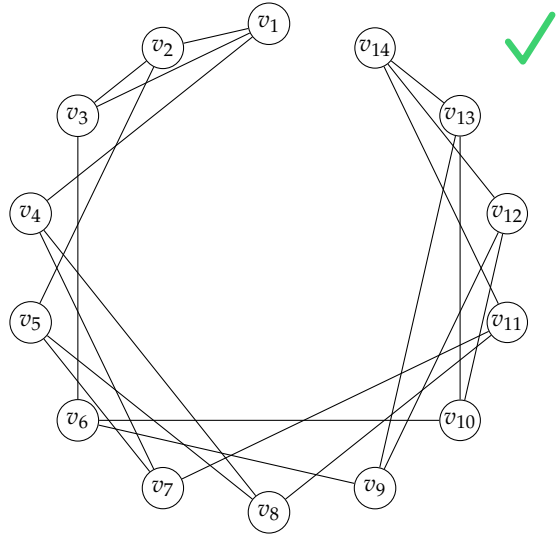

\subsubsection{Girth-$4$ graphs}
These are graphs with the following GENREG indices $400, 401, 402, 406, 411, 421$. We have found that all of them have permissible FORs(3). 

In Fig. \ref{3-reg-411}, we show graph $411$ for which a FOR(3) is as follows: 
\begin{align}\label{411-for}
	&\repr{1}=\zero, \quad\repr{2}=\repr{3}=\jeden, \quad\repr{4}=\jeden-\dwa, \quad\repr{5}=2\zero-\dwa,\non
	& \repr{6}=4\zero-\dwa,\quad \repr{7}=5\zero-\jeden-\dwa, \quad \repr{8}=\zero-\jeden-\dwa,\non
	&\repr{9}=\zero+3\jeden+2\dwa, \quad\repr{10}=\zero+\jeden+4\dwa, \quad\repr{11}=\repr{12}=\zero+\jeden, \\
	& \repr{13}=\zero-\jeden+\dwa, \quad\repr{14}=\zero-\jeden. \nonumber
\end{align}

\begin{figure}[htp]
	\begin{tikzpicture}[scale=0.5]
		\centering
		\begin{scope}[xshift=6.75cm,yshift=5.5cm,scale=1.3]
			\draw[ufogreen, line width=2.5pt, line cap=round, line join=round]
			(-0.3,0.4) -- (0,0) -- (0.6,0.95);
		\end{scope}
		\begin{scope}
			\GraphInit[vstyle=Classic]
			\SetGraphUnit{2.5}
			\node[circle, draw, fill=white, minimum size=10pt, inner sep=1.8pt] (v1) at ({90+0*360/14}:6.5cm) {$v_1$};
			\node[circle, draw, fill=white, minimum size=12pt, inner sep=1.8pt] (v2) at ({90+1*360/14}:6.5cm) {$v_2$};
			\node[circle, draw, fill=white, minimum size=12pt, inner sep=1.8pt] (v3) at ({90+2*360/14}:6.5cm) {$v_3$};
			\node[circle, draw, fill=white, minimum size=12pt, inner sep=1.8pt] (v4) at ({90+3*360/14}:6.5cm) {$v_4$};
			\node[circle, draw, fill=white, minimum size=12pt, inner sep=1.8pt] (v5) at ({90+4*360/14}:6.5cm) {$v_5$};
			\node[circle, draw, fill=white, minimum size=12pt, inner sep=1.8pt] (v6) at ({90+5*360/14}:6.5cm) {$v_6$};
			\node[circle, draw, fill=white, minimum size=12pt, inner sep=1.8pt] (v7) at ({90+6*360/14}:6.5cm) {$v_7$};
			\node[circle, draw, fill=white, minimum size=12pt, inner sep=1.8pt] (v8) at ({90+7*360/14}:6.5cm) {$v_8$};	
			\node[circle, draw, fill=white, minimum size=12pt, inner sep=1.8pt] (v9) at ({90+8*360/14}:6.5cm) {$v_9$};	
			\node[circle, draw, fill=white, minimum size=12pt,inner sep=0.6pt] (v10) at ({90+9*360/14}:6.5cm) {$v_{10}$};	
			\node[circle, draw, fill=white, minimum size=12pt,inner sep=0.6pt] (v11) at ({90+10*360/14}:6.5cm) {$v_{11}$};	
			\node[circle, draw, fill=white, minimum size=12pt,inner sep=0.6pt] (v12) at ({90+11*360/14}:6.5cm) {$v_{12}$};	
			\node[circle, draw, fill=white, minimum size=12pt,inner sep=0.6pt] (v13) at ({90+12*360/14}:6.5cm) {$v_{13}$};	
			\node[circle, draw, fill=white, minimum size=12pt,inner sep=0.6pt] (v14) at ({90+13*360/14}:6.5cm) {$v_{14}$};		
			%
				\draw (v1)--(v2); \draw (v1)--(v3); \draw (v1)--(v4);
				\draw (v2)--(v5); \draw (v2)--(v6);
				\draw (v3)--(v5);	\draw (v3)--(v6);
				\draw (v4)--(v7);
				\draw (v4)--(v8); \draw (v5)--(v9); \draw (v6)--(v10); \draw (v7)--(v9); \draw (v7)--(v10);
				\draw (v8)--(v11); \draw (v8)--(v12); \draw (v9)--(v13); \draw (v10)--(v14); \draw (v11)--(v13);
				\draw (v11)--(v14); \draw (v12)--(v13); \draw (v12)--(v14);
		\end{scope}
	\end{tikzpicture}
	\caption{Graph $411$ -- one of the $14$-vertex  $3$-regular graphs with $\girth=4$. The graph has a FOR(3) shown in Eq. \eqref{411-for}.}\label{3-reg-411}
\end{figure}
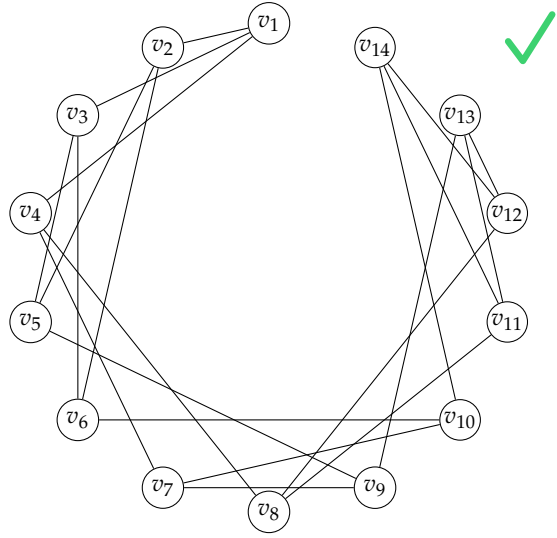

\subsubsection{Girth-$5$ graphs}\label{g-5}

We found that all the eight graphs (indices $501 - 508$ in GENREG) have permissible FORs(3). We give an exemplary representation for graph $501$ (see Fig. \ref{g-5-graph}):
\begin{align}\label{przykladowa-for-0}
	& \repr{1}=\zero, \quad \repr{2}=\jeden, \quad  \repr{3}= \jeden-2\dwa, \quad \repr{4}=\jeden-\dwa, \quad \non 	
	& \repr{5}=3\zero-2\dwa, \quad \repr{6}=\zero+2\dwa, \quad \repr{7}=2\zero+6\jeden+3\dwa, \\
	&  \repr{8}=2\zero-2\jeden-\dwa, \quad \repr{9}=2\zero+3\jeden+3\dwa,\quad\repr{10}=\zero-3\jeden-3\dwa, \non 
	& \repr{11}=6\zero+5\jeden-3\dwa, \quad
	\repr{12}=3\zero-3\jeden+4\dwa,\quad \repr{13}=\zero+\jeden,\quad \repr{14}=3\zero-3\jeden+\dwa. \nonumber
\end{align}
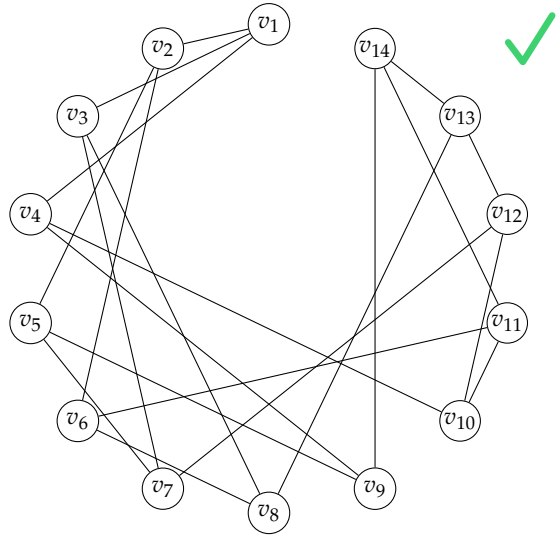
\begin{figure}[htp]
	\begin{tikzpicture}[scale=0.5]
		\centering
		\begin{scope}[xshift=6.75cm,yshift=5.5cm,scale=1.3]
			\draw[ufogreen, line width=2.5pt, line cap=round, line join=round]
			(-0.3,0.4) -- (0,0) -- (0.6,0.95);
		\end{scope}
		\begin{scope}
			\GraphInit[vstyle=Classic]
			\SetGraphUnit{2.5}
			\node[circle, draw, fill=white, minimum size=10pt, inner sep=1.8pt] (v1) at ({90+0*360/14}:6.5cm) {$v_1$};
		\node[circle, draw, fill=white, minimum size=12pt, inner sep=1.8pt] (v2) at ({90+1*360/14}:6.5cm) {$v_2$};
		\node[circle, draw, fill=white, minimum size=12pt, inner sep=1.8pt] (v3) at ({90+2*360/14}:6.5cm) {$v_3$};
		\node[circle, draw, fill=white, minimum size=12pt, inner sep=1.8pt] (v4) at ({90+3*360/14}:6.5cm) {$v_4$};
		\node[circle, draw, fill=white, minimum size=12pt, inner sep=1.8pt] (v5) at ({90+4*360/14}:6.5cm) {$v_5$};
		\node[circle, draw, fill=white, minimum size=12pt, inner sep=1.8pt] (v6) at ({90+5*360/14}:6.5cm) {$v_6$};
		\node[circle, draw, fill=white, minimum size=12pt, inner sep=1.8pt] (v7) at ({90+6*360/14}:6.5cm) {$v_7$};
		\node[circle, draw, fill=white, minimum size=12pt, inner sep=1.8pt] (v8) at ({90+7*360/14}:6.5cm) {$v_8$};	
		\node[circle, draw, fill=white, minimum size=12pt, inner sep=1.8pt] (v9) at ({90+8*360/14}:6.5cm) {$v_9$};	
		\node[circle, draw, fill=white, minimum size=12pt,inner sep=0.6pt] (v10) at ({90+9*360/14}:6.5cm) {$v_{10}$};	
		\node[circle, draw, fill=white, minimum size=12pt,inner sep=0.6pt] (v11) at ({90+10*360/14}:6.5cm) {$v_{11}$};	
		\node[circle, draw, fill=white, minimum size=12pt,inner sep=0.6pt] (v12) at ({90+11*360/14}:6.5cm) {$v_{12}$};	
		\node[circle, draw, fill=white, minimum size=12pt,inner sep=0.6pt] (v13) at ({90+12*360/14}:6.5cm) {$v_{13}$};	
		\node[circle, draw, fill=white, minimum size=12pt,inner sep=0.6pt] (v14) at ({90+13*360/14}:6.5cm) {$v_{14}$};		
		\draw (v1)--(v2);
		\draw (v1)--(v3);
		\draw (v1)--(v4);
		\draw (v2)--(v5);
		\draw (v2)--(v6);
		\draw (v3)--(v7);
		\draw (v3)--(v8);
		\draw (v4)--(v9);
		\draw (v4)--(v10);
		\draw (v5)--(v7);
		\draw (v5)--(v9);
		\draw (v6)--(v8);
		\draw (v6)--(v11);
		\draw (v7)--(v12);
		\draw (v8)--(v13);
		\draw (v9)--(v14);
		\draw (v10)--(v11);
		\draw (v10)--(v12);
		\draw (v11)--(v14);
		\draw (v12)--(v13);
		\draw (v13)--(v14);
		\end{scope}
	\end{tikzpicture}
	\caption{One of the eight girth-$5$ $3$--regular graphs on $14$ vertices; a FOR(3) for this graph is given in Eq. \eqref{przykladowa-for-0}.} \label{g-5-graph}
\end{figure}

 \subsubsection{Girth-$6$ graph --- Heawood graph}\label{g-6}
 
The Heawood graph \cite{heawood-math,heawood-house} (graph number $509$ in GENREG) is the unique $(3,6)$-cage, i.e., it is the smallest $3$-regular graph with girth $\girth=6$. We show this graph in Fig. \ref{heawood}. It follows from  Ref. \cite{parameters} that the minimal dimension for a FOR with real vectors is  $d=4$. An exemplary one is as follows:
\begin{align}\label{przykladowa-for}
	& \repr{1}=\zero, \quad \repr{2}=14\jeden+7\dwa-4\ket{3}, \quad  \repr{3}= 3\jeden-2\dwa, \quad \repr{4}=\jeden, \quad \non 	
	& \repr{5}=2\zero-2\jeden-7\ket{3}, \quad \repr{6}=10\zero-22\jeden+32\dwa-21\ket{3}, \quad \repr{7}=5\zero+2\jeden+3\dwa, \\
	&  \repr{8}=\zero+2\jeden+3\dwa, \quad \repr{9}=\zero-2\dwa-\ket{3},\quad\repr{10}=\zero+\dwa, \non 
	& \repr{11}=2\zero-5\jeden+2\ket{3}, \quad
	\repr{12}=\zero+\jeden-\dwa,\quad \repr{13}=\zero-\jeden-\dwa,\quad \repr{14}=2\zero-\jeden+2\ket{3}. \nonumber
\end{align}

We  have not been able to verify whether the complex case also requires four dimensions. 
However, in view of the results of Ref. \cite{parameters} it is likely that this is indeed the case.

\begin{figure}[htp]
	\begin{tikzpicture}[scale=0.6]
		\centering
		\begin{scope}[xshift=3.5cm,scale=0.9]
			\GraphInit[vstyle=Classic]
			\SetGraphUnit{2.5}
			\node[circle, draw, fill=white, minimum size=10pt, inner sep=1.8pt] (v1) at ({90+0*360/14}:6.5cm) {$v_1$};
			\node[circle, draw, fill=white, minimum size=12pt, inner sep=1.8pt] (v2) at ({90+1*360/14}:6.5cm) {$v_2$};
			\node[circle, draw, fill=white, minimum size=12pt, inner sep=1.8pt] (v3) at ({90+2*360/14}:6.5cm) {$v_3$};
			\node[circle, draw, fill=white, minimum size=12pt, inner sep=1.8pt] (v4) at ({90+3*360/14}:6.5cm) {$v_4$};
			\node[circle, draw, fill=white, minimum size=12pt, inner sep=1.8pt] (v5) at ({90+4*360/14}:6.5cm) {$v_5$};
			\node[circle, draw, fill=white, minimum size=12pt, inner sep=1.8pt] (v6) at ({90+5*360/14}:6.5cm) {$v_6$};
			\node[circle, draw, fill=white, minimum size=12pt, inner sep=1.8pt] (v7) at ({90+6*360/14}:6.5cm) {$v_7$};
			\node[circle, draw, fill=white, minimum size=12pt, inner sep=1.8pt] (v8) at ({90+7*360/14}:6.5cm) {$v_8$};	
			\node[circle, draw, fill=white, minimum size=12pt, inner sep=1.8pt] (v9) at ({90+8*360/14}:6.5cm) {$v_9$};	
			\node[circle, draw, fill=white, minimum size=12pt,inner sep=0.6pt] (v10) at ({90+9*360/14}:6.5cm) {$v_{10}$};	
			\node[circle, draw, fill=white, minimum size=12pt,inner sep=0.6pt] (v11) at ({90+10*360/14}:6.5cm) {$v_{11}$};	
			\node[circle, draw, fill=white, minimum size=12pt,inner sep=0.6pt] (v12) at ({90+11*360/14}:6.5cm) {$v_{12}$};	
			\node[circle, draw, fill=white, minimum size=12pt,inner sep=0.6pt] (v13) at ({90+12*360/14}:6.5cm) {$v_{13}$};	
			\node[circle, draw, fill=white, minimum size=12pt,inner sep=0.6pt] (v14) at ({90+13*360/14}:6.5cm) {$v_{14}$};		
	\draw (v1)--(v2);
	\draw (v1)--(v3);
	\draw (v1)--(v4);
	\draw (v2)--(v5);
	\draw (v2)--(v6);
	\draw (v3)--(v7);
	\draw (v3)--(v8);
	\draw (v4)--(v9);
	\draw (v4)--(v10);
	\draw (v5)--(v11);
	\draw (v5)--(v12);
	\draw (v6)--(v13);
	\draw (v6)--(v14);
	\draw (v7)--(v11);
	\draw (v7)--(v13);
	\draw (v8)--(v12);
	\draw (v8)--(v14);
	\draw (v9)--(v11);
	\draw (v9)--(v14);
	\draw (v10)--(v12);
	\draw (v10)--(v13);
		\end{scope}
		\begin{scope}[xshift=10cm,yshift=5cm,scale=0.9]
			\draw[-stealth, thick, bend left=40] (0,0) to (3,0);
		\end{scope}
		\begin{scope}[xshift=24cm,yshift=5.2cm,scale=0.15]
		\draw[line width=2.5,line cap=round] (1.5,0) .. controls ++(0,2) and ++(0,-2) .. (4,4)
		to[out=90,in=0] (2,6)
		to[out=180,in=90] (0,4);		
		\fill (1.5,-2) circle (0.75);
		\end{scope}
			\begin{scope}[xshift=18.5cm,scale=0.9]
			\GraphInit[vstyle=Classic]
			\SetGraphUnit{2.5}
				\node[circle, draw, fill=white, minimum size=10pt, inner sep=1.8pt] (v1) at ({90+0*360/14}:6.5cm) {$v_1$};
			\node[circle, draw, fill=white, minimum size=12pt, inner sep=1.8pt] (v2) at ({90+1*360/14}:6.5cm) {$v_2$};
			\node[circle, draw, fill=white, minimum size=12pt, inner sep=1.8pt] (v3) at ({90+2*360/14}:6.5cm) {$v_5$};
			\node[circle, draw, fill=white, minimum size=12pt, inner sep=0.6pt] (v4) at ({90+3*360/14}:6.5cm) {$v_{11}$};
			\node[circle, draw, fill=white, minimum size=12pt, inner sep=1.8pt] (v5) at ({90+4*360/14}:6.5cm) {$v_7$};
			\node[circle, draw, fill=white, minimum size=12pt, inner sep=1.8pt] (v6) at ({90+5*360/14}:6.5cm) {$v_3$};
			\node[circle, draw, fill=white, minimum size=12pt, inner sep=1.8pt] (v7) at ({90+6*360/14}:6.5cm) {$v_8$};
			\node[circle, draw, fill=white, minimum size=12pt, inner sep=0.6pt] (v8) at ({90+7*360/14}:6.5cm) {$v_{12}$};	
			\node[circle, draw, fill=white, minimum size=12pt, inner sep=0.6pt] (v9) at ({90+8*360/14}:6.5cm) {$v_{10}$};	
			\node[circle, draw, fill=white, minimum size=12pt,inner sep=0.6pt] (v10) at ({90+9*360/14}:6.5cm) {$v_{13}$};	
			\node[circle, draw, fill=white, minimum size=12pt,inner sep=1.8pt] (v11) at ({90+10*360/14}:6.5cm) {$v_{6}$};	
			\node[circle, draw, fill=white, minimum size=12pt,inner sep=0.6pt] (v12) at ({90+11*360/14}:6.5cm) {$v_{14}$};	
			\node[circle, draw, fill=white, minimum size=12pt,inner sep=1.8pt] (v13) at ({90+12*360/14}:6.5cm) {$v_{9}$};	
			\node[circle, draw, fill=white, minimum size=12pt,inner sep=1.8pt] (v14) at ({90+13*360/14}:6.5cm) {$v_{4}$};			
			\draw (v1)--(v2)--(v3)--(v4)--(v5)--(v6)--(v7)--(v8)--(v9)--(v10)--(v11)--(v12)--(v13)--(v14)--(v1); 	\draw (v1)--(v6);	\draw (v1)--(v14);	
			\draw (v2)--(v3); \draw (v2)--(v11);
			\draw (v3)--(v8); 	
			\draw (v4)--(v13); \draw (v5)--(v10); \draw (v7)--(v12); \draw (v9)--(v14);
			%
		\end{scope}
	\end{tikzpicture}
	\caption{The Heawood graph in the GENREG form (left) and in a more common layout (right).  The minimal dimension for a real FOR is $d=4$ [see Eq. \eqref{przykladowa-for}]; it is not clear whether this is also the minimal dimension for the complex case. } \label{heawood}
\end{figure}
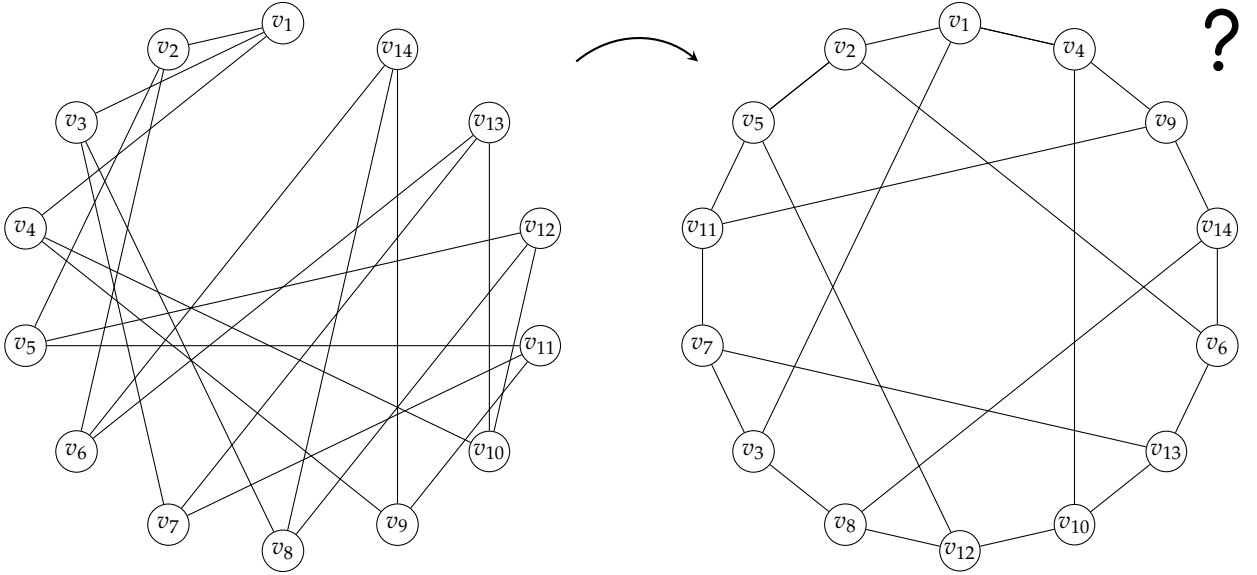

\subsection{Disconnected graphs}\label{3-reg-disconnected}

As noted in the main text, only type (b) graphs ($6 + 8$ vertices) need a more detailed investigation. We need to analyze (i) $6$-vertex and (ii) $8$-vertex $3$-regular graphs. (i) There are two graphs in this category. We find that the house graph $H_5$ is an induced subgraph of the first one. The second one is the complete bipartite graph $K_{3,3}$ (also called the utility graph). It actually admits a FOR(2) and in any FOR(d) there is at least one three-tuple of identical vectors (cf. Sec. \ref{14-vert-4-reg}, where $K_{4,4}$ is analyzed). (ii) This category consists of five graphs. We find that $\{K_5,L_6\}$ is an obstruction set defining  a single graph having a FOR(3) -- graph number $3$, denote it $P_3$ (see Fig. \ref{8-vert-3-reg}).  An exemplary FOR(3) of $P_3$ is as follows:

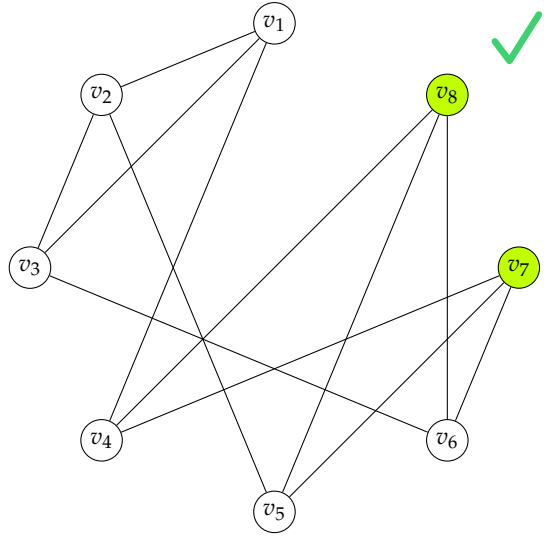
\begin{figure}[htp]
	\begin{tikzpicture}[scale=0.5]
		\centering
			\begin{scope}[xshift=6.25cm,yshift=5.5cm,scale=1.3]
			\draw[ufogreen, line width=2.5pt, line cap=round, line join=round]
			(-0.3,0.4) -- (0,0) -- (0.6,0.95);
		\end{scope}
		\begin{scope}
			\GraphInit[vstyle=Classic]
			\SetGraphUnit{2.5}
			\node[circle, draw, fill=white, minimum size=10pt, inner sep=1.8pt] (v1) at ({90+0*360/8}:6.5cm) {$v_1$};
			\node[circle, draw, fill=white, minimum size=12pt, inner sep=1.8pt] (v2) at ({90+1*360/8}:6.5cm) {$v_2$};
			\node[circle, draw, fill=white, minimum size=12pt, inner sep=1.8pt] (v3) at ({90+2*360/8}:6.5cm) {$v_3$};
			\node[circle, draw, fill=white, minimum size=12pt, inner sep=1.8pt] (v4) at ({90+3*360/8}:6.5cm) {$v_4$};
			\node[circle, draw, fill=white, minimum size=12pt, inner sep=1.8pt] (v5) at ({90+4*360/8}:6.5cm) {$v_5$};
			\node[circle, draw, fill=white, minimum size=12pt, inner sep=1.8pt] (v6) at ({90+5*360/8}:6.5cm) {$v_6$};
			\node[circle, draw, fill=lime, minimum size=12pt, inner sep=1.8pt] (v7) at ({90+6*360/8}:6.5cm) {$v_7$};
			\node[circle, draw, fill=lime, minimum size=12pt, inner sep=1.8pt] (v8) at ({90+7*360/8}:6.5cm) {$v_8$};			
			\draw (v1)--(v2); \draw (v1)--(v4); \draw (v4)--(v8);  \draw (v4)--(v7);
			\draw (v2)--(v3);
			\draw (v3)--(v1); \draw (v2)--(v5); \draw (v6)--(v8); \draw (v3)--(v6); \draw (v5)--(v8); \draw (v6)--(v7);
			\draw (v5)--(v7);
		\end{scope}
		\end{tikzpicture}
		\caption{The only $8$-vertex $3$-regular graph with a FOR(3), see Eq. \eqref{for-8-3}. Vectors corresponding to colored vertices are necessarily the same.}\label{8-vert-3-reg}
		\end{figure}

\begin{align}\label{for-8-3}
	&\ket{v_1}=\zero,\quad \repr{2}=\jeden,\quad \repr{3}=\dwa, \nonumber\\
	&\ket{v_4}= \jeden-\dwa,\quad \repr{5}=\zero+\dwa,\quad \repr{6}=\zero+\jeden,  \\
	&\ket{v_7}=\repr{8}=\zero-\jeden-\dwa. \nonumber
\end{align}

It immediately follows that $\repr{7}=\repr{8}$ holds for any FOR(3) of $P_3$ leading to the elimination of $K_{3,3}+P_3$ as a LOG.

\subsection{Petersen graph}

As an addition to the analysis of the fourteen-vertex three-regular graphs, we also consider the Petersen graph.
The Petersen graph (see Fig. \ref{petersen}) is the unique $(3,5)$-cage. It is a ten-vertex graph, which has a FOR(3) and we show one below:
\begin{align}\label{petersen-for3}
	&\ket{v_1}=\jeden, \quad \ket{v_2}=\zero, \quad \ket{v_3}=2\jeden+3\dwa, \quad \ket{v_4}=2\zero+3\jeden-2\dwa, \nonumber\\
		&\ket{v_5}= \zero+\dwa, \quad \ket{v_6}=\zero-\jeden-\dwa, \quad \ket{v_7}= 2\zero+\dwa, \\
			&\ket{v_8}= \jeden-\dwa, \quad \ket{v_9}= \zero+3\jeden-2\dwa, \quad\ket{v_{10}}= \zero -2\jeden-2\dwa, \nonumber
\end{align}
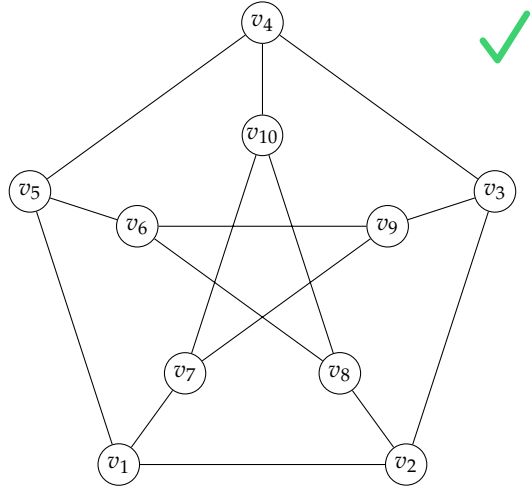
\begin{figure}[htp]
	\begin{tikzpicture}[scale=0.5]
		\centering
		\begin{scope}
		\GraphInit[vstyle=Classic]
		\SetGraphUnit{2.5}
		\node[circle, draw, fill=white, minimum size=10pt, inner sep=1.8pt] (v1) at ({90+0*360/5}:6.5cm) {$v_4$};
		\node[circle, draw, fill=white, minimum size=12pt, inner sep=1.8pt] (v2) at ({90+1*360/5}:6.5cm) {$v_5$};
		\node[circle, draw, fill=white, minimum size=12pt, inner sep=1.8pt] (v3) at ({90+2*360/5}:6.5cm) {$v_1$};
		\node[circle, draw, fill=white, minimum size=12pt, inner sep=1.8pt] (v4) at ({90+3*360/5}:6.5cm) {$v_2$};
		\node[circle, draw, fill=white, minimum size=12pt, inner sep=1.8pt] (v5) at ({90+4*360/5}:6.5cm) {$v_3$};
		\node[circle, draw, fill=white, minimum size=12pt, inner sep=0.6pt] (v6) at ({90+5*360/5}:3.5cm) {$v_{10}$};
		\node[circle, draw, fill=white, minimum size=12pt, inner sep=1.8pt] (v7) at ({90+6*360/5}:3.5cm) {$v_6$};
		\node[circle, draw, fill=white, minimum size=12pt, inner sep=1.8pt] (v8) at ({90+7*360/5}:3.5cm) {$v_7$};	
		\node[circle, draw, fill=white, minimum size=12pt, inner sep=1.8pt] (v9) at ({90+8*360/5}:3.5cm) {$v_8$};	
		\node[circle, draw, fill=white, minimum size=12pt,inner sep=1.8pt] (v10) at ({90+9*360/5}:3.5cm) {$v_9$};	
		\draw (v1)--(v2)--(v3)--(v4)--(v5)--(v1);
		\draw (v1)--(v6)--(v8)--(v10)--(v7)--(v9)--(v6);
		\draw (v5)--(v10); \draw (v4)--(v9); \draw (v3)--(v8); \draw (v2)--(v7);
		\end{scope}
			\begin{scope}[xshift=6.25cm,yshift=5.5cm,scale=1.3]
				\draw[ufogreen, line width=2.5pt, line cap=round, line join=round]
				(-0.3,0.4) -- (0,0) -- (0.6,0.95);
			\end{scope}
	\end{tikzpicture}
	\caption{Petersen graph; the graph admits a FOR(3), see Eq. \eqref{petersen-for3}.} \label{petersen}
\end{figure}

\section{Thirteen-vertex four-regular graphs with the considered induced subgraphs}\label{details-13-4}

Table \ref{tabelka-summary} below provides indices of graphs $M_i$ admitting
certain induced subgraphs.
Ladder graph $L_6$ (see Appendix \ref{drabina}), which does not belong to the
obstruction set $\mathcal{O}_3$ [Eq. \eqref{obstrukcja}], has been included.
%
%
	\begin{table}[h!]
		\begin{tabular}{cccc}
			\hline\hline 
			\makecell{forbidden \\ induced   subgraph }
			&\makecell{\# graphs w/ \\ induced subgraph } & \makecell{indices of graphs
				\\ with induced subgraph}     \\ \hline \\
			\vspace{0.4cm}$4$-clique $C_4$ & $671$  & \footnotesize{\makecell{$1-401$,
					$439$, $449$, $511$, $517$, $537$, $568$,   $578$, $606$,  $618$, $643$, 
					$686$,  $700$,   $716$,  \\$740$,  $774$, $781$, $867$, $882$, $931$,
					$945$,
					$956$,  $1164$,  $1166$,  $1316$,  $1343$,  $1350-1594$}}  \\ \\
			A-graph $A_6$ & $10\:672$  & \footnotesize{\makecell{$5-7$, $9-10$,
					$12-14$, $16-24$, $26-30$, $32$, $34$, $38-47$, $49$, $51$, $53-55$, $57-67$,
					$69-71$, \\
					$75-83$,  $85-115$, $117-121$, $123-125$, $128-196$, $198-239$,
					$243-278$, $280-286$, $288-329$, \\
					$331-332$, $335-336$, $338-361$, $363-396$, $398-401$, $406-422$, $425$,
					$427$, $429-430$, $432-434$, \\
					$436-440$, $442-455$, $457-476$, $478-483$, $485-494$, $496-683$,
					$685-690$, $692-696$, $698-1135$, \\
					$1137$, $1139-1140$, $1142-1207$, $1209-1224$, $1226-1230$, $1232$,
					$1234-1238$, $1240-1243$, \\
					$1245-1286$, $1288-1342$, $1344-1454$, $1456-1459$, $1461-1468$,
					$1470-1594$, $1596-1598$, \\
					$1600-1604$, $1606-1625$, $1627-1648$, $1650-1657$, $1659-1759$,
					$1761-1784$, $1786$, \\
					$1788-1816$, $1818-1820$, $1822-1992$, $1995-2057$, $2059-2061$, $2063$,
					$2065-2099$, \\
					$2101-2103$, $2105-2115$, $2117-2293$, $2295-2353$, $2355-2415$,
					$2417-2532$, $2534-3422$, \\
					$3424-4880$, $4882-4936$, $4938$, $4940-5007$, $5009-5045$, $5047-5056$,
					$5058-7554$, \\
					$7556-7570$, $7572-9047$, $9049-9081$, $9083-9161$, $9163$, $9165-10778$
			}}   \\\\
			House graph $H_5$ & $10\:662$  & \footnotesize{\makecell{$1-3$, $5-7$,
					$9-14$, $16-34$, $36-40$, $42-46$, $49-52$, $54-55$,  $57-83$, $85-278$,\\
					$280-286$, $288-293$, $295-312$,  $314-329$,  $331-333$,  $335-340$,
					$342-386$,\\ $388-396$,  $398-401$,  $403-409$, $411-412$, $414-423$, $425-431$,
					$433-434$, \\ $436-440$, $442-476$, $478-532$, $534-537$, $539-608$,  
					$610-612$,  $614-643$,  \\ $646-668$,  $670-675$,  $677-873$,  $875-879$,
					$881-899$, $901-934$, $936-962$, \\ $964-990$,  $992-993$, $995-1003$,
					$1005-1007$, $1009-1039$, $1041-1078$, \\ $1080-1137$,  $1139-1140$,
					$1142-1179$, $1181-1207$, $1209-1224$,  $1226-1230$,\\ $1232-1286$, $1288-1303$,
					$1305-1315$, $1317-1331$,  $1333-1342$, $1345-1347$,\\ $1349-1579$,$1581-1584$,
					$1586-1820$,  $1822-1901$, $1903-1909$, $1911-1919$,\\ $1921-2057$, 
					$2059-2063$,  $2065-2103$, $2105-2134$, $2136-2415$, $2417-2532$, \\ 
					$2534-4880$, $4882$, $4884-4921$, $4923-5056$,  $5058$, $5060-5064$,
					$5066-6701$, \\ $6703$, $6705-9038$, $9041-9042$, $9044$, $9046-9992$,
					$9995-10747$}}   \\\\
			Kite graph $K_5$ & $8\:919$  & \footnotesize{\makecell{$1-2$, $5-1352$,
					$1354-1356$, $1358-1361$, $1363-1367$,  $1369-1374$, $1376-1381$, \\
					$1383-1396$, $1398-1399$, $1401-1418$, $1420-1442$, $1444-1445$, $1451$,
					$1453-1468$,\\ $1470-1478$, $1480$, $1482-1483$, $1487-1490$, $1492-1505$, 
					$1507-1513$, $1516-1520$,  \\ $1522-1523$, $1525-1527$, $1530$, $1533$,
					$1535-1536$, $1540-1541$,  $1545-1549$, \\ $1551-1552$,  $1554-1556$,
					$1560-1562$, $1568-1571$, $1574-1579$, $1581$,  $1583-1584$, \\ $1586$,
					$1590-1592$, $1595-1784$, $1786-4926$,  $4930-4935$, $4937-4938$,
					$4940-4948$, \\ $4950$, $4952-4953$,  $4957-4959$, $4961$, $4967-4970$,
					$4972-4975$, $4979$, $4981$,  $4983$,\\ $4985$, $4990-4993$, $4995-4997$,
					$4999-5002$, $5007$,  $5009-5024$, $5027-5029$, \\ $5031-5036$,  $5040$,
					$5042-5047$, $5049-5055$, $5058$, $5063-5064$, $5066-5070$, \\ $5074-5077$,
					$5082-5083$, $5086-5087$, $5089-5090$, $5092$, $5094-5099$,  $5101$, \\
					$5103-5107$, $5109-5114$, $5116$, $5118-9047$} }   \\
			\\
			Ladder graph $L_6$ & $9\:933$ & \tiny{\makecell{
					$5$, $10$, $13-14$, $16-17$, $20$, $22$, $26-27$, $32$, $34$, $38-40$,
					$45$, $49$, $51$, $53-55$, $57-58$,  $61$, $63$, $67$, $70-71$, $75-78$, $80$,
					$83$,
					$86-88$, $93$,\\ $97-101$,$105-115$, $118$, $121$, $123-124$, $128-130$,
					$133-137$, $139$, $142-150$, $156-158$, $160$, $162-177$,$181-188$, $193-195$,
					$197$,\\
					$199-200$, $202-203$, $205-210$,$213-227$, $229-239$, $244-246$,
					$248$,$250-254$, $256-270$,$273-278$, $280$, $282-283$, $285-286$,\\ $288-291$,
					$293$, 
					$295-296$,$298$, $300-303$, $305-308$, $311-315$, $317-318$, $320-321$,
					$324$, $326$,$328-329$, $331-332$, $335-336$,\\ $338-343$, $345-347$, $349-351$,
					$353-356$,$359$, $361$, $364$, $366-370$, $372-376$, $378-385$,
					$389-390$, $392$, $396$, $398-401$, $408-415$,\\ $427$, $430$, $434$, $436-438$,
					$440$, $442$, $444-447$,  $449-455$, $458-459$, $464$, $466$,
					$469-473$,$475-476$, $479-480$, $482$, $485-490$, $501-503$,\\ $510$, $512-514$,
					$517-527$, $529-532$, $534-544$, $550-551$, $553$, $555-559$,
					$561-683$,$686$, $688-690$, $696$, $698-702$, $704-714$, $716-726$, \\$728-729$,				
					$731-733$,$735-736$, $739$, $741-748$, $750-753$, $756-766$, $768-769$,
					$771-777$, $779-780$, $782-787$, $790-821$, $825$, $828-831$,\\ $834$, 
					$837-843$, $845-862$, $864-865$, $867-873$,$875-901$, $903-915$,
					$918-930$, $933-990$, $992$, $994-1005$, $1008-1039$, $1042-1043$,\\
					$1045-1128$,
					$1130$, $1134-1135$, $1137$, $1139-1140$, $1142-1148$,$1150-1153$,
					$1155-1163$, $1165$, $1171-1181$, $1183-1185$, $1187$, $1189-1195$,\\
					$1197-1202$, $1204-1207$, $1209-1210$, $1212-1224$, $1227-1230$, $1232$,
					$1234-1238$,$1241-1245$, $1247-1252$, $1255$, $1257-1258$, $1260-1262$, \\
					$1264-1269$, $1271-1274$, $1276-1286$, $1288-1310$, $1313-1315$, $1317$,
					$1319-1336$, $1338-1342$, $1344-1350$, $1352-1357$, $1360-1368$,\\ $1370-1399$,
					$1402$, $1404-1405$, $1407-1412$, $1415-1454$,$1456-1457$, $1461$, $1464$,
					$1466-1479$, $1481-1492$, $1494-1495$, $1498-1499$,\\ $1502-1511$, $1513-1514$,
					$1516-1517$, $1519-1522$, $1524$, $1526-1539$, $1541-1544$, $1546-1577$,
					$1582-1587$, $1589$, $1592$, $1594$, $1601$, $1604$,\\ $1607-1608$, $1611$,
					$1613-1629$, $1631$, $1641$, $1645$, $1648$, $1652-1654$, $1659-1660$,
					$1662-1669$, $1671-1677$, $1679-1688$, $1690-1691$,\\ $1693-1699$, $1701-1702$,
					$1704$, $1706-1712$, $1715-1716$, $1718-1725$,$1727-1754$, $1756$, $1758-1759$,
					$1761-1762$,$1764$, $1766-1767$,\\ $1769$, $1771-1781$,  $1784$, $1788$,
					$1790-1794$, $1796$, $1799$, $1801$, $1804$, $1806$, $1812$, $1818$, 
					$1822-1903$, $1907$, $1909$, $1912$, $1914$, $1916$, $1921-1922$,\\
					$1925-1926$, $1931$,$1935$, $1937-1944$, $1946-1947$, $1950-1952$, $1954-1977$, 
					$1988$, $1995-2001$, $2004$,$2006$, $2009-2012$, $2014-2033$,\\
					$2035-2036$, $2040-2042$, $2044$, $2046-2047$,$2050-2055$, $2059$, $2063$,
					$2066-2072$,
					$2074-2077$, $2079$, $2081-2089$, $2091-2096$, $2098-2099$,\\
					$2102-2103$, $2105-2113$, $2115$, $2117-2130$, $2132-2133$, $2136-2167$,
					$2169-2172$, 	$2174-2183$, $2186-2194$, $2196-2205$, $2208-2216$, \\$2218-2252$,
					$2254$, $2258-2269$, $2273-2274$, $2276-2291$, $2293$, $2295-2306$, $2308-2310$,
					$2312$, $2314-2323$, $2325-2339$, $2343-2344$,\\ $2346-2347$, $2349$,
					$2351-2353$, $2355-2357$, $2359-2364$, $2366-2391$, $2393-2398$, $2400-2411$,
					$2413-2415$, $2417-2428$, $2430-2441$,\\ $2443-2444$, $2446-2460$,
					$2462-2464$, $2467-2469$, $2472-2479$, $2481-2495$, $2497-2523$, $2526-2532$, 
					$2534-2542$, $2544-2671$, $2673-2677$,\\ $2679-2686$, $2688-2701$,
					$2704-2713$, $2715-2730$, $2732$, $2736-2743$, $2745$, $2747$, $2749-2785$,
					$2789-2830$, $2833-2847$, $2850$, $2852$,\\$2855$, $2857-3000$,
					$3003-3092$, $3094-3097$, $3099$, $3102$, $3104-3105$, $3107-3159$,$3165-3166$,
					$3168$, $3170-3171$, $3173$, $3175$, $3177-3193$,\\ $3195$, $3197-3271$,
					$3273-3285$,$3290-3342$, $3344$, $3346-3350$, $3352$, $3354-3382$, $3384-3394$,
					$3398-3400$, $3403-3420$,$3422$, $3425-3426$,\\ $3428-3576$, $3578$,
					$3580-3623$, $3625-3626$, $3628-3737$, $3739-3747$,$3749-3758$, $3761-3764$, 
					$3766-3784$, $3786-3787$, $3789-3807$,\\ $3809-3827$, $3829$,
					$3831-3849$, $3851-3869$, $3871-3888$, $3890-3947$, $3949-3958$, $3961-3964$,
					$3967-3975$,
					$3977-3987$, $3989-3994$,\\ $3996-4007$, $4009-4141$, $4144-4255$,
					$4257$, $4259-4267$, $4269-4280$, $4284-4285$, $4288-4291$, $4293-4312$,
					$4314-4325$, 
					$4327$, \\$4329-4330$,$4332-4334$, $4336-4342$, $4344-4348$, $4350-4355$,
					$4361-4362$, $4364-4383$, $4386-4391$, $4393-4397$, $4399-4403$,\\ $4405-4409$, 
					$4411-4430$, $4432$, $4434-4487$, $4489-4492$, $4494-4498$, $4500-4510$,
					$4514-4532$, $4534-4605$, $4607-4609$, $4611-4612$,\\ $4615$, $4617-4621$,
					$4623-4648$, $4650$, 
					$4652-4678$, $4680-4705$, $4707-4711$, $4714-4723$, $4725-4765$,
					$4767-4781$, $4783-4790$, $4792$, \\ $4794-4798$,$4800-4803$, $4805-4856$,
					$4858-4861$,
					$4863-4867$, $4869-4876$, $4878-4880$, $4882-4889$, $4891$, $4893$,
					$4895-4905$, $4908$,\\ $4910-4912$, $4914-4918$, $4920-4926$, $4928-4936$,
					$4938$, 
					$4940-4946$, $4948-4952$, $4954-4955$, $4957-4968$, $4970-4986$, $4988$,
					$4990-4995$, \\ $4997$, $4999-5007$, $5011$, $5013-5015$, $5017-5027$,
					$5029-5045$,
					$5047$, $5049-5050$, $5052-5056$, $5058$, $5060$,$5062$, $5064-5066$,
					$5070-5075$,\\ $5078$, $5080-5081$,$5083-5084$, $5086-5090$, $5092-5093$,
					$5095$, 
					$5098-5101$, $5103-5108$, $5110-5111$, $5113$, $5115-5527$, $5529-5530$,
					\\ $5532-5805$, $5816-6008$,  $6012$, $6014-6041$, $6043$, $6046$, $6048$, 
					$6051-6053$, $6056-6553$, $6555-6692$, $6694-6699$, $6701$, $6703-6708$,
					\\ $6711-6718$, $6720-6723$, $6725-6732$,  $6734-6755$, $6757-6771$, $6773$,
					$6776-6785$, $6787-7016$, $7018-7022$, $7025-7140$, $7142-7198$, \\
					$7200-7273$, $7275-7276$, $7278-7279$, $7281$,  $7283-7300$, $7302$,
					$7304-7305$, 
					$7307-7316$, $7318-7361$, $7363-7427$, $7430-7554$, \\$7557-7558$,
					$7560$, $7563-7581$, $7583-7585$, $7587-7897$,  $7899-7988$, $7990-8043$,
					$8045-8046$,
					$8048-8711$, $8713-8847$, $8849-8853$,\\ $8856-8885$, $8887-8933$, 
					$8935-8981$, $8983-8985$, $8987-9039$, $9041-9042$,$9044-9047$, $9050$,
					$9054-9081$,
					$9083-9161$,\\ $9166-9170$, $9172-9239$, $9242-9243$,$9245-9277$,
					$9279-9925$, $9927-10145$, $10147-10778$}}\\
			\hline\hline
		\end{tabular}
		\caption{Complete list of $4$-regular graphs on $13$ vertices sharing the considered induced
			subgraphs.}\label{tabelka-summary}
	\end{table}
%

\end{document}